\documentclass{article}

\usepackage{arxiv}

\usepackage{lineno,hyperref}
\modulolinenumbers[5]

\usepackage[english]{babel}
\usepackage{latexsym,lmodern}
\usepackage{amssymb}
\usepackage{amsfonts}
\usepackage{amsmath}
\usepackage{graphicx}
\usepackage{xcolor}

\usepackage{multicol}
\usepackage{enumitem}
\usepackage{amsthm}
\usepackage{tikz}
\usepackage{tikz-network}
\usepackage{mathdots}
\usepackage{diagbox}
\usepackage{adjustbox}
\usepackage{dashbox}
\usepackage{subcaption}
\usepackage{cancel}
\usepackage{color}
\usepackage{booktabs, makecell}
\usepackage{array}
\usepackage{multirow}
\usepackage{amssymb}
\usepackage{gensymb}
\usepackage{tabularx}
\usepackage{booktabs}
\usetikzlibrary{fadings}

\usetikzlibrary{patterns}
\usetikzlibrary{shadows.blur}
\usetikzlibrary{shapes}
\usetikzlibrary{arrows,automata,positioning}
\tikzset{>=latex}

\newcommand{\R}{\mathbb R}
\newcommand{\N}{\mathbb N}
\newcommand{\Q}{\mathbb Q}
\newcommand{\Z}{\mathbb Z}
\newcommand{\D}{\displaystyle}
\newcommand{\GG}{\mathcal{G}}
\newcommand{\DLOG}{\textbf{DLOGSPACE}}
\newcommand{\PREDU}{\textsf{U-PRED}}
\newcommand{\PREDB}{\textsf{B-PRED}}
\newcommand{\PREDC}{\textsf{PRED-CHG}}
\newcommand{\PRED}{\textsf{PRED}}
\newcommand{\REACH}{\textsf{REACH}}
\newcommand{\FNC}{\textbf{FNC}}
\newcommand{\NC}{\textbf{NC}}
\DeclareMathOperator{\lcm}{lcm}

\newtheorem{theorem}{Theorem}
\newtheorem{corollary}{Corollary}

\newtheorem{lemma}[theorem]{Lemma}
\newtheorem{proposition}{Proposition}

\newtheorem*{problem}{Problem}
\newtheorem{definition}[theorem]{Definition}
\newtheorem{remark}{Remark}

\usepackage{tabularx,lipsum,environ,amsmath,amssymb}

\usepackage{tabularx,lipsum,environ,amsmath,amssymb}

\makeatletter
\newcommand\TODOBIB[1]{{\color{red}{\textbf{BIBLIO:} #1}}}
\newcommand\TODO[1]{{\par\noindent\color{red}{\textbf{TODO:} #1}}}
\newcommand\TODOFIG[1]{{\par\noindent\color{red}{\textbf{FIGURE:} #1}}}
\newcommand{\problemtitle}[1]{\gdef\@problemtitle{#1}}
\newcommand{\probleminput}[1]{\gdef\@probleminput{#1}}
\newcommand{\problemquestion}[1]{\gdef\@problemquestion{#1}}
\newcounter{ProblemCounter}
\NewEnviron{problem1}{
	\problemtitle{}\probleminput{}\problemquestion{} 
	\BODY
	\par\addvspace{.5\baselineskip}
	\stepcounter{ProblemCounter}
	\centering
	\noindent
	\begin{tabularx}{0.8\textwidth}{@{\hspace{\parindent}} l X c}
		\hline
		\multicolumn{2}{@{\hspace{\parindent}}l}{{\bf Problem \arabic{ProblemCounter}: }\@problemtitle} \\
		\hline	
		\textbf{Input:} & \@probleminput \\
		\textbf{Question:} & \@problemquestion \\
		\hline	
	\end{tabularx}
	\par\addvspace{.5\baselineskip}
}

\title{Intrinsic Simulations and Universality in Automata Networks\thanks{This research was partially supported by French ANR project FANs ANR-18-CE40-0002 (G.T., M.R.W.), ECOS project C19E02 (G.T., M.R.W.) and ANID FONDECYT Postdoctorado 3220205 (M.R-W)}}

\author{ Mart\'in R\'ios-Wilson \\
Facultad de Ingenier\'ia y Ciencias, Universidad Adolfo Ib\'a\~nez.\\ 
	\texttt{martin.rios@uai.cl} \\
	\And
	Guillaume Theyssier \\
Aix-Marseille Universit\'e, CNRS, I2M (UMR 7373), Marseille, France.\\
	\texttt{guillaume.theyssier@cnrs.fr} \\
}

\date{}

\renewcommand{\shorttitle}{\textit{arXiv} Template}

\begin{document}

\maketitle
\begin{abstract}
  An automata network (AN) is a finite graph where each node holds a state from a finite alphabet and is equipped with a local map defining the evolution of the state of the node depending on its neighbors. They are studied both from the dynamical and the computational complexity point of view. Inspired from well-established notions in the context of cellular automata, we develop a theory of intrinsic simulations and universality for families of automata networks. We establish many consequences of intrinsic universality in terms of complexity of orbits (periods of attractors, transients, etc) as well as hardness of the standard well-studied decision problems for automata networks (short/long term prediction, reachability, etc). In the way, we prove orthogonality results for these problems: the hardness of a single one does not imply hardness of the others, while intrinsic universality implies hardness of all of them. As a complement, we develop a proof technique to establish intrinsic simulation and universality results which is suitable to deal with families of symmetric networks were connections are non-oriented. It is based on an operation of glueing of networks, which allows to produce complex orbits in large networks from compatible pseudo-orbits in small networks. As an illustration, we give a short proof that the family of networks were each node obeys the rule of the 'game of life' cellular automaton is strongly universal. This formalism and proof technique is also applied in a companion paper devoted to studying the effect of update schedules on intrinsic universality for concrete symmetric families of automata networks.
\end{abstract} 
\newpage

  \section{Introduction}
  An automata network is a (finite) graph where each node holds a state from a finite set $Q$ and is equipped with a local transition map that determines how the state of the node evolves depending on the states of neighboring nodes.
  Automata networks introduced in the 40s \cite{McCulloch_1943} are both a family of dynamical systems frequently used in the modeling of biological networks \cite{Thomas_1973,KAUFFMAN_1969} and a computational model \cite{GR15,ChM14,goles1997reaction,WR79ii,WR79i}.
  As such they can exhibit complexity from two very different point of view: complexity of orbits in their phase space (cycles, transient, etc) and computational complexity of canonical problems associated to them (prediction, reachability, etc).  
An automata network is completely described by its global map ${F:Q^V\rightarrow Q^V}$ that defines the collective evolution of all nodes $V$ of the network.
However, this global map hides two fundamental aspects at the heart of automata network literature \cite{DemongeotS20,Gadouleau_2019,chatain2018most,ARS17b,GN12,Aracena_2009,Robert_1969}: the interaction graph (knowing on which nodes effectively depends the behavior of a given node) and the update schedule (knowing in which order and with which degree of synchrony are local transition maps of each node applied).
Moreover, when considering computational complexity, the concrete representation of automata networks crucially maters.
An arbitrary global map $F$ requires as much information to be described as the complete description of its orbits.
On the contrary, many standard families of automata networks are actually described concretely by the interaction graph and the parameters of each local transition map (see below).
We are mainly interested in such natural families with concrete succinct representations and not in arbitrary global maps.

This paper is the first of a series of two that tackles core questions of automata networks theory:
\begin{itemize}
\item how dynamical complexity relates to computational complexity in automata networks? 
\item what families exhibit maximal complexity in both aspects and what are the key ingredients or sufficient conditions to achieve it?
\item how can update schedules compensate for the limitation coming from restrictions and symmetries of a family?
\item what hierarchy of different types of update schedules can be established with respect to the rise of complexity when applied to a family?
\end{itemize}

One of the key concept that we put forward to tackle these question is that of intrinsic universality of a family, \textit{i.e.} the ability to simulate arbitrary automata networks.
In this first paper, we focus on the formalism of universality of families of automata networks and the common roots of dynamics and computational complexity.
It aims both at studying the consequences of universality and establishing proof techniques for it.
Building upon these tools, the second paper is devoted to update schedules and their key influence on concrete families of automata networks, in particular how a non-universal family can recover universality under a particular update schedule.

\subsection{Motivating examples}

The initial motivation of this paper lies in the striking interplay established in some cases between the computational complexity and the dynamical properties of automata networks.
These two aspects are often measured in terms of periods of attractors or transient length and the computational complexity of various prediction or reachability problems respectively.
The prediction problem consists in predicting the future state of an objective node, given an initial condition and has been widely studied for different classes of automata networks \cite{goles2014computational,goles2016pspace,goles2021freezing,goles2022complexity} or cellular automata \cite{moore1998predicting,griffeath1996life,moore1997majority,goles2018complexity}.
The reachability problems generally asks whether a given input configuration will reach a given target configuration \cite{Barrett_2006,Barrett_2003}, sometimes giving only partial information on the target configuration \cite{Folschette_2015}. It has also been considered in a much generalized version as a benchmark for computational universality for arbitrary symbolic dynamical systems \cite{DelvenneKB06}.

However, if we consider only these measurements, the link between computational complexity and dynamical complexity in a particular automata network family is not clear. In fact, different families exhibit various types of behavior both from a computational and from a dynamical standpoint. Let us consider three of them: threshold networks, algebraic networks and the outer-totalistic networks.

In the first family, all the nodes hold a binary state: $0$ or $1$. In addition, an integer $\theta_v$ is assigned to each node representing its  \emph{threshold}. The dynamics of the network is defined locally by the sum of the states of its neighbors. If at least $\theta_v$ neighbors of $v$ are in state $1$ the node $v$ will change its internal state to $1$. In any other case, the node will change its state to $0$. A standard example of threshold networks are majority networks where $\theta_v$ is set to half the size of the neighborhood. 
A  seminal result \cite{GolesO80,PaperGoles} shows that symmetric threshold networks (the underlying interaction graph is non-directed) cannot have periodic orbits of period more than 2, and that they have polynomially bounded transients; this implies the existence of a polynomial time algorithm to predict the future of a node from any given initial configuration \cite{goles2014computational}, however the prediction problem can be P-complete on symmetric majority networks \cite{moore1997majority}. Here, the strong limitation on attractor periods does not totally exclude computational complexity.
On the other hand, majority networks under partially asynchronous updates (precisely block-sequential update modes) were shown to have super-polynomial periodic orbits and a PSPACE-complete prediction problem \cite{GolesMST16,ChM14}. A similar result was obtained recently for conjunctive networks under a more general update mode called 'firing memory' \cite{goles2020firing}.

The second family, the algebraic networks, are networks that have a linear global map and can be represented by a matrix. Thus, an orbit of the system is completely determined by the powers of the corresponding matrix. A notorious example is the case of the elementary 1D cellular automata rule $90$ in which the local rule in each cell is simply the sum modulo $2$ of the states of the left and the right neighbors. On one hand it is shown in \cite{martin1984algebraic} that the largest period $\pi_N$ of a periodic orbits of this CA on a network of size $N$ is exponential, precisely: ${\Pi_N\leq 2^{\frac{N-1}{2}} - 1}$ for all $N$ and ${\limsup_N\frac{\log\Pi_N}{N}=\frac{1}{2}}$, but the value of ${\Pi_N}$ highly depends on the multiplicative number theoretic properties of $N$. On the other hand, since an orbit is determined by the powers of some matrix, the prediction problem can be solved by an efficient parallel algorithm (see for example \cite{goles2014computational} for an efficient algorithm for the prediction problem on disjunctive networks and \cite{joseph1992introduction} for more details on efficient parallel algorithms for the prefix sum problem). Here, the strongly limited computational complexity does not exclude long periods of attractors.

As a last example, the outer-totalistic networks are characterized by dynamics depending only in the state of the node and the sum of the states of the neighbors. A very famous example is the cellular automaton known as Conway's \emph{Game of Life}. In this two dimensional cellular automaton,  each cell can be either dead or alive and the state of each cell depends on how many alive cells are nearby. If a cell is dead and has exactly $3$ alive neighbors, or if it it alive and has $2$ or $3$ alive neighbors, then the cell is alive at the next step. In any other case, the cell is dead at the next step.
In \cite{liferokadur}, this cellular automaton is shown to be intrinsically universal, \textit{i.e.} able to closely simulate any other cellular automaton.
It has consequences in terms of computational complexity but also dynamics \cite{bulk2,goles2011communication}.
These two forms of complexity can be observed on periodic configurations, where the cellular automaton can be seen as a family of (finite) automata networks, and through hardness of the prediction problem and exponential periods of attractors for instance.

\subsection{Our contributions}

The previous examples show that focusing on a single aspect (like exponential periods of some orbits, or the hardness of a prediction problem) is not enough to accurately describe the complexity of families of automata networks.
However, complexity of different parameters or problems for a family can have a common root with many consequences in terms of complexity: intrinsic universality, \textit{i.e.} the ability to simulate arbitrary behaviors of automata networks.
This notion has been successfully developed in the context of cellular automata \cite{goles2011communication,bulk2,OllingerRichard4states,ollinger2008universalities,boyertheyssier09,mazoyer1999inducing} and for other models like self-assembly tilings \cite{MeunierWoods2017,DLPSSW2012,PSTWW14}.
Although it appears implicitly in some results like in \cite{goles1997reaction}, the notion has never been explicitly ported to automata networks to our knowledge.
In this paper, as our first main contribution, we explore this idea in depth: we develop a formalism of intrinsic simulation and universality for families of automata networks, and we study its implication for various notions of complexity, both of dynamical and computational nature.

A common aspect to most hardness results for automata networks is the formulation of gadgets simulating logic gates in order to perform a reduction to classic computational complexity problems such as the circuit value problem or the boolean satisfiability problem.
Since there are similarities in the way these gadgets are constructed, it is usually accepted that the mere existence of these gadgets imply that the reductions are correct because a 'simulation of Boolean circuit' is achieved.
To our knowledge, no formalism for gadgets and their composition has been proposed so far, and we see two problems in this situation:
\begin{itemize}
\item First, it is not clear what 'simulating a Boolean circuit with gadgets' means since the asynchronous nature of Boolean circuit evaluation does not match with the extreme sensitivity to synchronism of automata networks; moreover, in a symmetric (non-oriented) network family like in the examples above, gadgets have no clear notion of input and output nodes and, due to potential unwanted feedback behaviors, it is not as simple to correctly connect them as it is to connect logic gates in a Boolean circuit.
\item Second, the lack of gadget formalism pushes the authors to state results about hardness of a particular decision problem or lower bound of a particular dynamical parameter: this approach is not modular since, as illustrated above, complexity of one particular aspect does not generally imply complexity of another. On the contrary, subsequent literature can possibly derive several consequences of a suitable statement about existence of gadgets. We believe for instance that proving intrinsic universality with gadgets is a much more far reaching objective than proving the hardness of a particular decision problem.
\end{itemize}

In this paper, as our second main contribution, we establish a general framework to make proofs based on gadget, suitable to prove intrinsic simulation and universality, even for families with symmetric networks which are prone to feedback problems.
At the core of our approach is a concept of glueing of small unoriented networks, which permits to build a large network with a prescribed oriented flow of information through the dynamics (pseudo-orbits).
If we consider two families of automata networks $A$ and $B$, our approach allows to prove intrinsic simulation in order to hit two targets with one bullet: with a proof that A simulates a previously analyzed B we show complexity lower bounds on A both in the dynamical and in the computational sense; in particular, intrinsic universality of A can be shown like this.
Our approach is not a priori limited to a small set of benchmark problems or properties: universality results can be used as a black box to then prove new corollaries on the complexity of other decision problems or other dynamical aspects.

We aim at making our approach general and modular, while giving concrete and relevant examples that were considered in literature.
That is why, behind the main goals of the paper presented above, we carefully deal with many aspects of the formalism that are often left implicit in the literature or treated in a specific way for selected examples. A more detailed list of our contributions is as follows:
\begin{enumerate}
\item We define general notion of family of automata networks (Definition~\ref{def:standdef}) that takes care of the size and complexity of representations.
\item As a running example we consider a set of concrete families (Definition~\ref{def:globalcsan}) that capture many natural examples from the literature (Section~\ref{sec:exampl-csan-netw}) and will be at the heart of the companion paper.
\item We define a notion of intrinsic simulation between families of automata networks (Definition~\ref{def:simu-family}) and define two variants of intrinsic universality (Definition~\ref{def:universal}).
\item We study the theoretical implications of universality in the computational complexity of various problems (Corollary~\ref{cor:universality}) as well as dynamical complexity (Theorem~\ref{them:univ-rich-dynamics}).
\item We show that, while universality implies hardness of all problems, there exist non-universal families which are hard for one problem but easy for another (Theorems~\ref{theo:orthogonalpredictions} and~\ref{theo:orthogonalreach}). We also show that neither polynomially bounded periods nor polynomially bounded transients can discard a maximal complexity of the prediction problem alone, but both bounds together does (Theorem~\ref{theo:orthodyncompu}).
\item We show that the automata network family associated to an intrinsically universal of cellular automata lies between our two notions of intrinsic universality for automata networks (Theorems~\ref{theo:iuca} and~\ref{theo:ballgrowthnouniv}).
\item We consider several families of automata networks based on a finite set of local update rules called $\GG$-networks (Definition~\ref{def:g-network}), and show various simulations or (non-)universality results on standard examples (Section~\ref{sec:gmonuniv} and Theorems~\ref{theo:Gconj-networks} and~\ref{theo:non-poly-transient} and~\ref{theo:transient-nonuniversal}).
\item We develop a proof tool for intrinsic simulation results in the context of automata networks on non-directed graphs: it is based on a way to compose small networks into a larger one called \emph{glueing} (Definition~\ref{def:glueing}) that preserves pseudo-orbits (Lemma~\ref{lem:pseudo-orbit-glueing}) and can be directly used to prove simulation of $\GG$-network families through a concept of gadgets (Lemma~\ref{lem:from-gadgets-to-networks}).
\item Based on previously analyzed families of $\GG$-networks, we obtain a sufficient condition for intrinsic universality that boils down to the existence of a coherent finite set of pseudo-orbits of a finite set of networks from the considered family (Corollary~\ref{cor:univfrommon}).
\item We apply this techniques on the family of automata networks on arbitrary undirected graphs where each node behaves like the cells of the famous 'Game of Life' cellular automaton, and establish strong universality of the family (Theorem~\ref{theo:goluniv}) by new gadgets that are smaller and more time-efficient than the classical ones used to prove intrinsic universality of the cellular automaton on the grid.
\end{enumerate}

As a main by-product of this framework, we obtain a precisely formalized proof technique that from a finite set of conditions to checks on a finite set of networks of a family, deduces intrinsic universality of 
the family and therefore several hardness complexity results as well as several lower bounds on various parameters of the possible dynamics within the family.
This proof technique will be heavily used in the companion paper dedicated to update schedules.

\section{Automata networks and families}
\label{sec:automataandfamilies}

A \emph{graph} is a pair $G = (V,E)$  where $V$ and $E$ are finite sets satisfying $E \subseteq V \times V.$ We will call $V$ the set of \emph{nodes} and the set $E$ of \emph{edges}. We call $|V|$ the \emph{order} of $G$ and we usually identify this quantity by the letter $n$. Usually, as $E$ and $V$ are finite sets we will implicitly assume that there exists an ordering of the vertices in $V$ from $1$ to $n$ (or from $0$ to $n-1$). Sometimes we will denote the latter set as $[n].$ If $G = (V,E)$ and $V' \subseteq V, E' \subseteq E$ we say that $G'$ is a \emph{subgraph} of $G.$ We call a graph $P=(V,E)$ of the form $V= \{v_{1},\hdots, v_{n}\}$ $E = \{(v_{1}v_{2}),\hdots,(v_{n-1},v_{n})\}$  a \emph{path graph}, or simply a \emph{path}. We often refer to a path by simply denoting its sequence of vertices $\{v_{1},\hdots,v_{n}\}$. We denote the \emph{length} of a path by its number of edges. Whenever $P = (V= \{v_{1},\hdots, v_{n}\},E = \{(v_{1}v_{2}),\hdots,(v_{n-1},v_{n})\}$ is a path we call the graph in which we add the edge $\{v_{n},v_{1}\}$
a \emph{cycle graph} or simply a \emph{cycle} and we call it $C$ where $C = P + \{v_{n},v_{1}\}.$ Analogously, a cycle  is denoted usually by a sequence of nodes and its length is also given by the amount of edges (or vertices) in the cycle. Depending of the length of $C$ we call it a $k$-cycle when $k$ is its length. A non-empty graph is called \emph{connected} if any pair of two vertices $u,v$ are linked by some path. Given any non-empty graph, a maximal connected subgraph is called a connected component.

We call \emph{directed graph} a pair $G = (V,E)$ together with two functions $\text{init}:E \to V$ and $\text{ter}:E \to V$ where each edge $e \in E$ is said to be directed from $\text{init}(e)$ to $\text{ter}(e)$ and we write $e=(u,v)$ whenever $\text{init}(e) = u$ and $\text{ter}(e) = v.$ There is also a natural extension of the definition of paths, cycles and connectivity for directed graphs in the obvious way. We say a directed graph is strongly connected if there is a directed path between any two nodes. A strongly connected component of a directed graph $G = (V,E)$  is a maximal strongly connected subgraph.

Given a (non-directed) graph $G=(V,E)$ and two vertices $u,v$ we say that $u$ and $v$ are neighbors if $(u,v) \in E$.  Remark that abusing notations, an edge $(u,v)$ is also denoted by $uv$. Let $v \in V,$ we call $N_{G}(v) = \{u \in V: uv \in E\}$ (or simply $N(v)$ when the context is clear)  the set of neighbors (or \emph{neighborhood}) of $v$ and $\delta(G)_v = |N_{G}(v)|$ to the \emph{degree} of $v$. Observe that if $G'=(V',E')$ is a subgraph of $G$ and $v \in V'$, we can also denote by $N_{G'}(v)$  the set of its neighbors in $G'$ and the degree of $v$ in $G'$ as $\delta(G')_{v} =|N_{G'}(v)|.$  In addition, we define the \emph{closed neighborhood} of $v$ as the set $N[v] = N(v) \cup \{v\}$ and we use the following notation $\Delta(G) = \max \limits_{v \in V} \delta_v$  for the \emph{maximum degree} of $G$. Additionally, given $v \in V$, we will denote by $E_{v}$ to its set of \emph{incident edges}, i.e., $E_{v} = \{e \in E: e=uv\}.$ We will use the letter $n$ to denote the order of $G$, i.e. $n = |V|$.  Also, if $G$ is a graph whose sets of nodes and edges are not specified, we use the notation $V(G)$ and $E(G)$ for the set of vertices and the set of edges of $G$ respectively. In the case of a directed graph $G = (V,E)$ we define for a node $v \in V$ the set of its \emph{in-neighbors} by $N^{-}(v) = \{ u \in V: (u,v) \in E\}$ and its \emph{out-neighbors} as $N^{+}(v) = \{ u \in V: (v,u) \in E\}.$ We have also in this context the indegree of $v$ given by $ \delta^{-} = |N^{-}(v)|$ and its \emph{outdegree} given by $ \delta^{+} = |N^{+}(v)|$

During the most part of of the text, and unless explicitly stated otherwise,  every graph $G$ will be assumed to be connected and undirected.
We start by stating the following basic definitions, notations and properties that we will be using in the next sections. In general, $Q$ and $V$ will denote finite sets representing the alphabet and the set of  nodes respectively. We define $\Sigma(Q)$ as the set of all possible permutations over alphabet $Q$. We call an \textit{abstract automata network} any function $F:Q^V \to Q^V$. Note that $F$ induces a dynamics in $Q^V$ and thus we can see $(Q^V,F)$ as dynamical system.
In this regard, we recall some classical definitions.  We call a \emph{configuration} to any element $x \in Q^{V}.$ If $S \subseteq V$ we define the restriction of a configuration $x$ to $V$ as the function $x|_{S} \in Q^{S}$ such that $(x|_{S})_{v} = x_{v}$ for all $v \in S$. In particular, if $S = \{v\},$ we write $x_{v}.$
 
Given an initial configuration $x \in Q^V$,  we define the \textit{orbit} of $x$ as  the sequence $\mathcal{O}(x) = (F^t(x))_{t\geq 0}$. 
We define the set of \emph{limit configurations} or \emph{recurrent configurations} of $F$ as $L(F) = \bigcap_{t\geq 0}F^t(Q^V)$.  Observe that since $Q$ is finite and $F$ is deterministic, each orbit is eventually periodic, i.e. for each $x \in Q^{V}$ there exist some $\tau, p \in \N$ such that $F^{\tau+p}(x) = F^{\tau}(x)$ for all $x \in Q^{V}$. Note that if $x$ is a limit configuration then, its orbit is periodic. In addition, any configuration $x \in Q^{V}$ eventually reaches a limit configuration in finite time. We denote the set of orbits corresponding to periodic configurations as $\text{Att}(F)=\{\mathcal{O}(x): x \in L(F)\}$ and we call it the set of $\emph{attractors}$ of $F.$
We define the \emph{global period} or simply the \emph{period} of $\overline{x} \in \text{Att}(F)$  by $p(\overline{x}) = \min \{p \in \N : \overline{x}(p) = \overline{x}(0)\}$. If $p(\overline{x}) = 1$ we say that $\overline{x}$ is a \emph{fixed point} and otherwise, we say that $\overline{x}$ is a \emph{limit cycle}.

Given a node $v$, its behavior ${x\mapsto F(x)_v}$ might depend or not on another node $u$.
This dependencies can be captured by a graph structure which plays an important role in the theory of automata networks (see \cite{Gadouleau_2019} for a review of known results on this aspect).
This motivates the following definitions.
\begin{definition}
 Let $F: Q^{V} \to Q^{V}$ be an abstract automata network and $G = (V,E)$ a directed graph. We say $G$ is a \emph{communication graph} of $F$ if for all $v \in V$ there exist $D \subseteq N^{-}_v$ and some function $f_v:  Q^{D} \to Q$ such that $F(x)_v = f_v(x|_{D}).$ The \emph{interaction graph} of $F$ is its minimal communication graph.
\end{definition}
Note that by minimality, for any node $v$ and any in-neighbor $u$ of $v$ in the interaction graph of some $F$, then the next state at node $v$ effectively depends on the actual state at node $u$.
More precisely, there is some configuration $c\in Q^V$ and some $q\in Q$ with ${q\neq c_u}$ such that ${F(c)_v\neq F(c')_v}$ where $c'$ is the configuration $c$ where the state of node $u$ is changed to $q$.
This notion of effective dependency is sometimes taken as a definition of edges of the interaction graph.


From now on, for an abstract automata network $F$ and some communication graph $G$ of $F$ we use the notation $\mathcal{A} = (G,F)$. In addition, by abuse of notation. we also call $\mathcal{A}$ an  abstract automata network. We define a set of automata networks or an \emph{abstract family} of automata networks on some alphabet $Q$ as a set $\mathcal{F} \subseteq \bigcup \limits_{n \in \N} \{F:Q^{V} \to Q^{V}:  V\subseteq [n]\}.$ 
Note that the latter definition provides a general framework of study as it allows us to analyze an automata network as an abstract dynamical system.
However, as we are going to be working also with a computational complexity framework, it is necessary to be more precise in how we represent them.
In this regard, one possible slant is to start defining an automata network from a communication graph.
It turns out that the most studied examples can be seen as labeled graphs: linear networks are given by a matrix (which is nothing else than a edge-labeled graph) and threshold networks


One of the main definition used all along this paper is that of \emph{concrete symmetric automata network}.
Roughly, they are non-directed labeled graph $G$ (both on nodes and edges) that represent an automata network.
They are \emph{concrete} because the labeled graph is a natural concrete representation upon which we can formalize decision problems and develop a computational complexity analysis.
They are \emph{symmetric} in two ways: first their communication graph is non-directed, meaning that an influence of node $u$ on node $v$ implies an equivalent influence of node $v$ on node $u$; second, the behavior of a given node is blind to the ordering of its neighbors in the communication graph, and it can only differentiate its dependence on neighbors when the labels of corresponding edges differ.

A \emph{multiset} over $Q$ is a map ${m:Q\to\N}$ (recall that ${0\in\N}$).
A $k$-bounded multiset over $Q$ is a map ${m:Q\to[k]=\{0,\ldots,k\}}$, the set of such multisets is denoted ${[k]^Q}$.
For instance a multiset in ${[2]^Q}$ is actually a set.
Note that when $Q$ is finite (which will always be the case below), any multiset is actually a bounded multiset.
\newcommand\mset[1]{m({#1})}
To any (partial) configuration $c\in Q^A$, we associate the multiset $\mset{c}$ which to any ${q\in Q}$ associates its number of occurrences in $c$, \textit{i.e.} 
\[\mset{c} = q\mapsto \#\{a\in A : c(a)=q\}.\]




\begin{definition}\label{def:csan}
  Given a non-directed graph $G = (V,E)$, a vertex label map $ \lambda: V \to ( Q\times \N^Q\to Q)$ and an edge label map $\rho: E \to (Q\to Q)$, we define the \textit{concrete symmetric automata network} (CSAN) $\mathcal{A} = (G,\lambda,\rho)$.
  A \emph{family of concrete symmetric automata networks} (CSAN family) $\mathcal{F}$ is given by an alphabet $Q$, a set of local labeling constraints ${\mathcal C\subseteq \Lambda\times R}$ where ${\Lambda = \{\phi : Q\times \N^Q\to Q\}}$ is the set of possible vertex labels and ${R = 2^{\{\psi : Q\to Q\}}}$ is the set of possible sets of neighboring edge labels.
  We say a CSAN ${(G,\lambda,\rho)}$ belongs to ${\mathcal{F}}$ if for any vertex $v$ of $G$ with incident edges $E_v$ it holds ${(\lambda(v),\rho(E_v))\in \mathcal C}$.
\end{definition}

Note that the labeling constraints defining a CSAN family are local.
In particular, the communication graph structure is a priori free.
This aspect will play an important role later when building arbitrarily complex objects by composition of simple building blocks inside a CSAN family.

Let us now define the abstract automata network associated to a CSAN, by describing the semantics of labels defined above.
Intuitively, labels on edges are state modifiers, and labels on nodes give a map that describes how the node changes depending on the set of sates appearing in the neighborhood, after application of state modifiers.
We use the following notation: given $\sigma \in \Sigma(Q)^V$ a collection of permutation and $x \in Q^A$ a partial configuration with ${A\subseteq V}$, we denote
by $x^{\sigma} = a\mapsto \sigma_a(x_a)$
In addition, given $x \in  Q^n$ we define the restriction of $x$  to some subset $U \subseteq V$ as the partial configuration $x|_U \in Q^{|U|}$ such that $(x|_U)_u = x_u$ for all $u \in U.$ 
\begin{definition}\label{def:globalcsan}
  Given a CSAN $(G,\lambda,\rho)$, its associated global map $F: Q^V \to Q^V$ is defined as follows.
  For all node $v \in V$ and  for all $x\in Q^n$:
  \begin{equation*}
    F(x)_v = \lambda_v (x_v,\mset{(x|_{N(v)})^{\rho}}).
  \end{equation*}
\end{definition}

Note that if $(G,\lambda,\rho)$ is a concrete automata network and $F$ its global rule then, $F$ is an abstract automata network with interaction graph included in $G$.
Note also that for any vertex $v$ and any configuration $x$, the multiset ${\mset{x|_{N(v)}}}$ is actually $k$-
bounded where $k$ is the degree of $v$.

\subsection{Examples of CSAN networks}
\label{sec:exampl-csan-netw}

In this subsection, we provide some examples from the literature that can be modeled as CSAN networks. In each case we provide the elements that define the corresponding labeled graph representation of each family.


\begin{enumerate}

\item \textbf{Symmetric Linear networks.} 
  When endowing the alphabet $Q$ with a finite ring structure, when can consider maps ${F:Q^n\to Q^n}$ defined by a ${n\times n}$ matrix $M=(m_{i,j})$ as ${F(x)=M\cdot x}$ when seeing configurations as column vectors.
  When $M$ is symmetric, one can see $F$ has the global map of a CSAN network defined as follows: the labeled graph $(G,\lambda,\rho)$ is defined by non-zero entries of $M$ and an edge between nodes $i$ and $j$ is labeled by the map ${q\mapsto m_{i,j}q}$, and, for each node $i$, the map $\lambda_i$ is just the map computing the sum in the ring $Q$ from a multiset, \textit{i.e.} such that ${\lambda_i (x_i,\mset{(x|_{N(i)})^{\rho}})= \sum_jm_{i,j}x_j}$. One of the most studied case is when $Q$ is actually endowed with a field structure, in particular in the Boolean case where the choice of coefficient is unique and the automata network is entirely determined by the graph $G$. Note however, that non-Boolean finite fields are worth studying since the graphs for which a given global property holds might differ from the Boolean case (see for instance Proposition 3.5 and Corollary~4.3 of \cite{Bridoux_2020}). A very well studied example of a network in this family is the case of additive elementary cellular automata. In particular, the rule $90$ in Wolfram notation, in which the new state of a cell is computed as the modulo 2 sum of its right and left neighbors. An example of its matrix representation for a small ring of $7$ cell is shown in Figure \ref{fig:linearnet}. In this case the matrix is defined over the field $\mathbb{F}_2.$ In the same Figure, it is shown an example of its dynamics starting from a uniformly random generated initial condition. 
  
  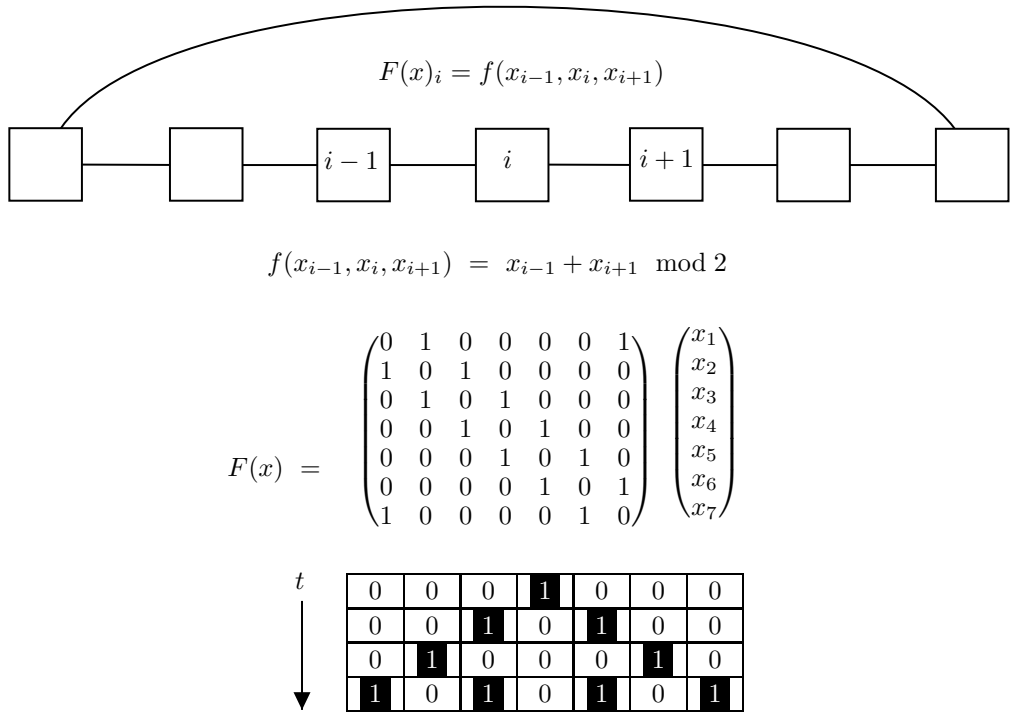
\begin{figure}

\centering
\tikzset{every picture/.style={line width=0.75pt}} 

\begin{tikzpicture}[x=0.75pt,y=0.75pt,yscale=-1,xscale=1]

\draw    (86.7,124.3) -- (162,124.6) ;
\draw    (162,124.6) -- (242,124.6) ;
\draw    (242,124.6) -- (322,124.6) ;
\draw    (318.7,124.3) -- (394,124.6) ;
\draw    (394,124.6) -- (474,124.6) ;
\draw    (474,124.6) -- (554,124.6) ;
\draw  [fill={rgb, 255:red, 255; green, 255; blue, 255 }  ,fill opacity=1 ] (149.4,106) -- (186,106) -- (186,142.6) -- (149.4,142.6) -- cycle ;
\draw  [fill={rgb, 255:red, 255; green, 255; blue, 255 }  ,fill opacity=1 ] (223.7,106.3) -- (260.3,106.3) -- (260.3,142.9) -- (223.7,142.9) -- cycle ;
\draw  [fill={rgb, 255:red, 255; green, 255; blue, 255 }  ,fill opacity=1 ] (303.7,106.3) -- (340.3,106.3) -- (340.3,142.9) -- (303.7,142.9) -- cycle ;
\draw  [fill={rgb, 255:red, 255; green, 255; blue, 255 }  ,fill opacity=1 ] (381.4,106) -- (418,106) -- (418,142.6) -- (381.4,142.6) -- cycle ;
\draw  [fill={rgb, 255:red, 255; green, 255; blue, 255 }  ,fill opacity=1 ] (455.7,106.3) -- (492.3,106.3) -- (492.3,142.9) -- (455.7,142.9) -- cycle ;
\draw    (86.7,124.3) .. controls (105,17.6) and (530,18.6) .. (554,124.6) ;
\draw  [fill={rgb, 255:red, 255; green, 255; blue, 255 }  ,fill opacity=1 ] (535.7,106.3) -- (572.3,106.3) -- (572.3,142.9) -- (535.7,142.9) -- cycle ;
\draw  [fill={rgb, 255:red, 255; green, 255; blue, 255 }  ,fill opacity=1 ] (68.4,106) -- (105,106) -- (105,142.6) -- (68.4,142.6) -- cycle ;
\draw    (231-15,424.6-80) -- (231-15,479.6-80) ;
\draw [shift={(231-15,479.6-80)}, rotate = 270] [fill={rgb, 255:red, 0; green, 0; blue, 0 }  ][line width=0.08]  [draw opacity=0] (8.93,-4.29) -- (0,0) -- (8.93,4.29) -- cycle    ;

\draw (197,166.4) node [anchor=north west][inner sep=0.75pt]    {$f( x_{i-1} ,x_{i} ,x_{i+1}) \ =\ x_{i-1} +x_{i+1} \ \bmod 2$};
\draw (242,206.4) node [anchor=north west][inner sep=0.75pt]    {$\begin{pmatrix}
0 & 1 & 0 & 0 & 0 & 0 & 1\\
1 & 0 & 1 & 0 & 0 & 0 & 0\\
0 & 1 & 0 & 1 & 0 & 0 & 0\\
0 & 0 & 1 & 0 & 1 & 0 & 0\\
0 & 0 & 0 & 1 & 0 & 1 & 0\\
0 & 0 & 0 & 0 & 1 & 0 & 1\\
1 & 0 & 0 & 0 & 0 & 1 & 0
\end{pmatrix}$};
\draw (177,270) node [anchor=north west][inner sep=0.75pt]    {$F( x) \ =$};
\draw (399,202) node [anchor=north west][inner sep=0.75pt]    {$\begin{pmatrix}
x_{1}\\
x_{2}\\
x_{3}\\
x_{4}\\
x_{5}\\
x_{6}\\
x_{7}
\end{pmatrix}$};
\draw (316,116) node [anchor=north west][inner sep=0.75pt]    {$i$};
\draw (253,70) node [anchor=north west][inner sep=0.75pt]    {$F( x)_{i} =f( x_{i-1} ,x_{i} ,x_{i+1}) \ $};
\draw (225.7,116.3) node [anchor=north west][inner sep=0.75pt]    {$i-1$};
\draw (384.7,115.3) node [anchor=north west][inner sep=0.75pt]    {$i+1$};
\draw (257-20,389.4-60) node [anchor=north west][inner sep=0.75pt]    {$\begin{array}{|c|c|c|c|c|c|c|}
\hline
0 & 0 & 0 & \colorbox{black!}{{\color{white} 1}} & 0 & 0 & 0\\
\hline
0 & 0 & \colorbox{black!}{{\color{white} 1}} & 0 & \colorbox{black!}{{\color{white} 1}}& 0 & 0\\
\hline
0 & \colorbox{black!}{{\color{white} 1}} & 0 & 0 & 0 & \colorbox{black!}{{\color{white} 1}} & 0\\
\hline
\colorbox{black!}{{\color{white} 1}} & 0 & \colorbox{black!}{{\color{white} 1}} & 0 & \colorbox{black!}{{\color{white} 1}} & 0 & \colorbox{black!}{{\color{white} 1}}\\
\hline
\end{array}$};
\draw (231-20,399-70) node [anchor=north west][inner sep=0.75pt]    {$t$};

\end{tikzpicture}

  \caption{\label{fig:linearnet} Representation and dynamics of Rule $90$ for a ring of $7$ cells. From up to bottom: (1) interaction graph; (2)  matrix representation and (3) dynamics.}
  \end{figure}

\item  \textbf{Symmetric  Threshold networks.}
The threhsold network family is a classical boolean network model that has been broadly studied both as a discrete dynamical system and also because of its computational properties.  In this case, we have that  $Q=\{0,1\}$, $\rho = \{\text{Id}\}$ where Id is the identity function and we have that $\lambda_{\theta_v}(x,m((x|_{N(v)}))) = \begin{cases}

1 & \text{ if }  m((x|_{N(v)}))[1] - \theta_v \geq 0  \\
0 & \text{ otherwise. }
\end{cases}$

where $\theta_v \in \mathbb{Z}$ is called the \emph{threshold} associated to each node in the network. An example of a threshold network represented as labeled graph is shown in Figure \ref{fig:threshold}. In that example, the threshold values are chosen so each node will change its state to $1$ if and only if the strict majority of their neighbors are in state $1$. Another example in which we exihibit the dynamics of a threshold network is the one Figure \ref{fig:dymthresh}. In this example the threshold are all the same and equal to $1$. In other words, each node will update its state to one if it has at least one neighbor in state $1$. In the same Figure, it is shown a particular case in which the dynamics reaches an attractor of period $2$.

\begin{figure}[t]
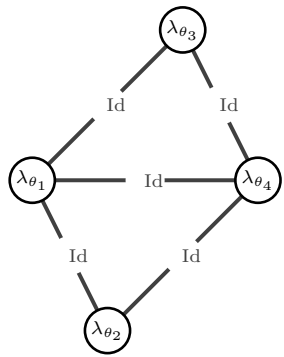

   \centering
 
\caption{A threshold network in which $\theta_1 = 3,\theta_2 = 2,\theta_3 = 2$ and $\theta_4 = 3.$ } 
\label{fig:threshold}  
\end{figure}

\begin{figure}
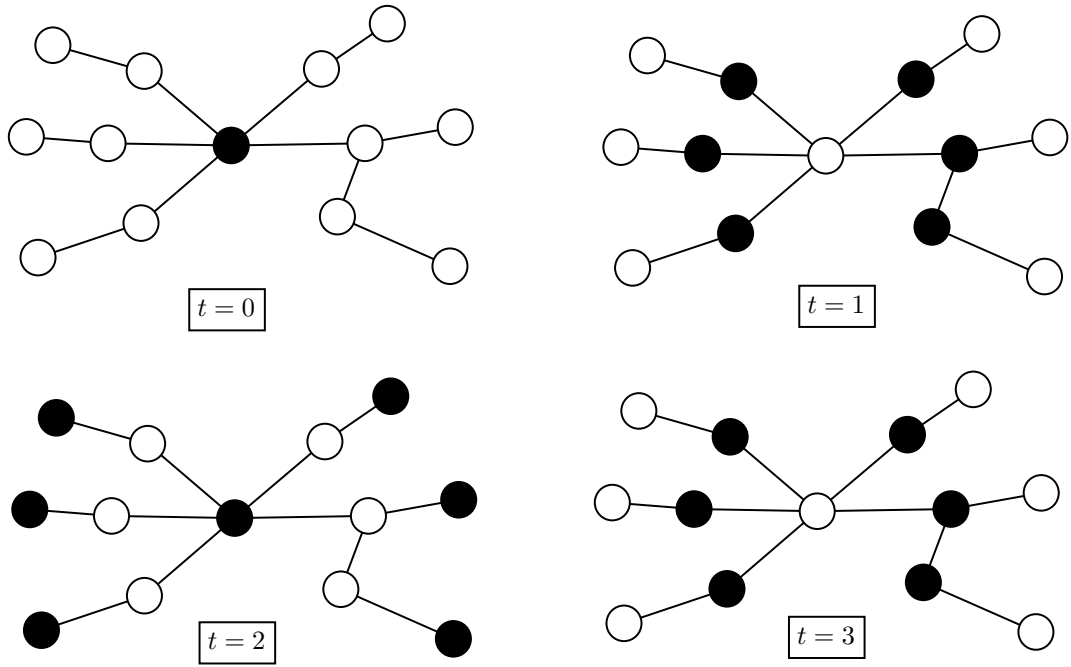

   \centering

\tikzset{every picture/.style={line width=0.75pt}} 

%

\caption{Dynamics of a threshold network where $\theta_v = 0$ for every node $v$ in the graph. Black nodes are in state $1$ and white nodes are in state $1$.}
\label{fig:dymthresh}
\end{figure}

\item \textbf{Symmetric Min/Max networks.}
In this example, the local functions of each node in the network will compute the minimum or the maximum state value in the set of different states. Formally, Let $Q$ be a totally ordered set. Given a multiset $X\in\N^Q$ and an order on $Q$, we denote by ${\min X = \min \{q\in Q: X(q)>0\}}.$ The family of \emph{min-max automata networks} over $Q$ is the set of CSAN ${(G,\lambda,\rho)}$ such that for each edge $e$, $\rho_e$ is the identity map and, for each node $v$, $\lambda_v(q,X)= \max X$ or $\lambda_v(q,X) = \min X$. An interesting particular case is the case in which the set of states is just $0$ and $1$. In this case, the min-max automata networks are just the AND-OR networks, which have been broadly studied in the literature. More precisely, if  $Q = \{0,1\}$) then,  $\lambda(x)_i = \bigwedge \limits_{j \in N(i)} x_j $ or $\lambda(x)_i = \bigvee \limits_{j \in N(i)} x_j $, for each node $i \in V.$ In Figure \ref{fig:AND-OR} there is an example of an AND-OR network defined on a graph of $4$ nodes.

\begin{figure}[b]
   \centering
\begin{tikzpicture}
\Vertex[label=$ \vee $,color=white]{1}
\Vertex[y=-2,x=1,label=$\wedge$,color=white]{2}
\Vertex[y=2,x=2,label=$\vee$,color=white]{3}	
\Vertex[y=0,x=3,label=$\wedge$,color=white]{4}

\Edge[label=Id](1)(2)
\Edge[label=Id](1)(3)
\Edge[label=Id](1)(4)
\Edge[label=Id](4)(3)
\Edge[label=Id](4)(2)
\end{tikzpicture} 
\caption{An AND-OR network. }  
\label{fig:AND-OR} 
\end{figure}
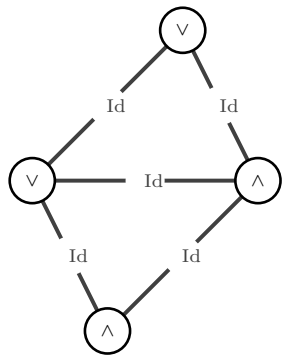

\item \textbf{Symmetric Outer-totalistic Boolean automata networks.} In this case, the alphabet is binary and local function $\lambda_v$ associated to each node in the network depends only the truncated sum of neighboring states, \textit{i.e.} on the number of occurrences of each state counted up to some constant number. In addition the edge labels $\rho$ are simply given by the identity function. 

\emph{Life-like rules} A classical example are the so-called life-like rules. These rules are defined for a binary alphabet where $1$ represents 'alive' nodes and $0$ dead ones, and the local rule is determined by two subsets $B$ (for birth) and $S$ (for survive) as follows: a dead node will change to state $1$ if the sum of the states of its neighbors is in $B$, and an alive nodes stays alive if and only if the sum of neighboring states belongs to $S$. A notable example (that actually inspired the definition of this class) is the famous Conway's Game of Life defined by ${B=\{3\}}$ and ${S=\{2,3\}}$. 

\emph{Interval rules}. Another natural class of examples are so-called interval rules defined by an interval ${[\alpha,\beta]}$ where ${0\leq\alpha\leq\beta}$ and a cell will change to state 1 if the sum the states of its neighbors is a number that is at least $\alpha$ and at most $\beta$. Observe that this case is similar to the latter family. In fact, the only difference is that for life-like rules, the transition rule depends on the state of each cell and the states of the neighbors and in the interval rules a cell depends only in the set of states of states of its neighbors. An example of a dynamics induced by an interval rule is shown in Figure \ref{fig:dymint}. In the case shown in the figure, the interval is given by $\alpha = 1$ and $\beta=2$. In other words, a node will update its state to $1$ if it has exactly one or two neighbors in state $1$ and it will update its state to $0$ otherwise.

\begin{figure}
   \centering
\begin{tikzpicture}
\Vertex[label=$ \lambda_{3,4} $,color=white]{1}
\Vertex[y=-2,x=1,label=$\lambda_{3,4}$,color=white]{2}
\Vertex[y=1,x=2,label=$\lambda_{3,4}$,color=white]{3}
\Vertex[y=0,x=3,label=$\lambda_{3,4}$,color=white]{4}

\Vertex[y=1,x=5,label=$\lambda_{3,4}$,color=white]{5}

\Edge[label=Id](1)(2)
\Edge[label=Id](1)(3)
\Edge[label=Id](1)(4)
\Edge[label=Id](4)(3)
\Edge[label=Id](4)(2)
\Edge[label=Id](5)(3)
\Edge[label=Id,bend=10](3)(2)
\Edge[label=Id](5)(4)
\Edge[label=Id, bend=20](5)(2)
\Edge[label=Id, bend=-90](5)(1)

\end{tikzpicture} 
\caption{Interval rule network for $\alpha=3, \beta=4$. }   
\label{fig:tot}
\end{figure}
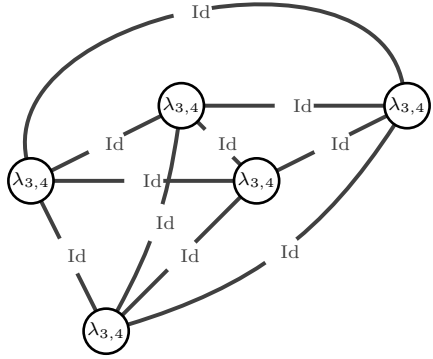
\begin{figure}[b]
\centering

\tikzset{every picture/.style={line width=0.75pt}} 

\begin{tikzpicture}[x=0.65pt,y=0.65pt,yscale=-1,xscale=1]

\draw    (82,138) -- (118,138) ;
\draw    (118,138) -- (166,138) ;
\draw    (72,138) -- (72.1,174.4) ;
\draw    (118,138) -- (118.1,174.4) ;
\draw  [fill={rgb, 255:red, 0; green, 0; blue, 0 }  ,fill opacity=1 ] (156.91,174.4) .. controls (156.91,168.68) and (161.47,164.04) .. (167.1,164.04) .. controls (172.73,164.04) and (177.29,168.68) .. (177.29,174.4) .. controls (177.29,180.12) and (172.73,184.76) .. (167.1,184.76) .. controls (161.47,184.76) and (156.91,180.12) .. (156.91,174.4) -- cycle ;
\draw    (167,138) -- (167.1,174.4) ;
\draw    (118.16,137.76) -- (118.12,101.36) ;
\draw  [fill={rgb, 255:red, 255; green, 255; blue, 255 }  ,fill opacity=1 ] (128.31,101.38) .. controls (128.31,107.1) and (123.73,111.73) .. (118.1,111.72) .. controls (112.47,111.71) and (107.92,107.07) .. (107.93,101.34) .. controls (107.93,95.62) and (112.51,90.99) .. (118.14,91) .. controls (123.77,91.01) and (128.32,95.65) .. (128.31,101.38) -- cycle ;
\draw  [fill={rgb, 255:red, 255; green, 255; blue, 255 }  ,fill opacity=1 ] (107.81,138) .. controls (107.81,132.28) and (112.37,127.64) .. (118,127.64) .. controls (123.63,127.64) and (128.19,132.28) .. (128.19,138) .. controls (128.19,143.72) and (123.63,148.36) .. (118,148.36) .. controls (112.37,148.36) and (107.81,143.72) .. (107.81,138) -- cycle ;
\draw    (165.16,136.76) -- (165.12,100.36) ;
\draw  [fill={rgb, 255:red, 0; green, 0; blue, 0 }  ,fill opacity=1 ] (175.31,100.38) .. controls (175.31,106.1) and (170.73,110.73) .. (165.1,110.72) .. controls (159.47,110.71) and (154.92,106.07) .. (154.93,100.34) .. controls (154.93,94.62) and (159.51,89.99) .. (165.14,90) .. controls (170.77,90.01) and (175.32,94.65) .. (175.31,100.38) -- cycle ;
\draw    (72.16,136.76) -- (72.12,100.36) ;
\draw  [fill={rgb, 255:red, 255; green, 255; blue, 255 }  ,fill opacity=1 ] (82.31,100.38) .. controls (82.31,106.1) and (77.73,110.73) .. (72.1,110.72) .. controls (66.47,110.71) and (61.92,106.07) .. (61.93,100.34) .. controls (61.93,94.62) and (66.51,89.99) .. (72.14,90) .. controls (77.77,90.01) and (82.32,94.65) .. (82.31,100.38) -- cycle ;
\draw  [fill={rgb, 255:red, 0; green, 0; blue, 0 }  ,fill opacity=1 ] (61.81,138) .. controls (61.81,132.28) and (66.37,127.64) .. (72,127.64) .. controls (77.63,127.64) and (82.19,132.28) .. (82.19,138) .. controls (82.19,143.72) and (77.63,148.36) .. (72,148.36) .. controls (66.37,148.36) and (61.81,143.72) .. (61.81,138) -- cycle ;
\draw    (165,138) -- (213,138) ;
\draw  [fill={rgb, 255:red, 0; green, 0; blue, 0 }  ,fill opacity=1 ] (203.91,174.4) .. controls (203.91,168.68) and (208.47,164.04) .. (214.1,164.04) .. controls (219.73,164.04) and (224.29,168.68) .. (224.29,174.4) .. controls (224.29,180.12) and (219.73,184.76) .. (214.1,184.76) .. controls (208.47,184.76) and (203.91,180.12) .. (203.91,174.4) -- cycle ;
\draw    (214,138) -- (214.1,174.4) ;
\draw  [fill={rgb, 255:red, 0; green, 0; blue, 0 }  ,fill opacity=1 ] (202.81,138) .. controls (202.81,132.28) and (207.37,127.64) .. (213,127.64) .. controls (218.63,127.64) and (223.19,132.28) .. (223.19,138) .. controls (223.19,143.72) and (218.63,148.36) .. (213,148.36) .. controls (207.37,148.36) and (202.81,143.72) .. (202.81,138) -- cycle ;
\draw    (212.16,136.76) -- (212.12,100.36) ;
\draw  [fill={rgb, 255:red, 0; green, 0; blue, 0 }  ,fill opacity=1 ] (222.31,100.38) .. controls (222.31,106.1) and (217.73,110.73) .. (212.1,110.72) .. controls (206.47,110.71) and (201.92,106.07) .. (201.93,100.34) .. controls (201.93,94.62) and (206.51,89.99) .. (212.14,90) .. controls (217.77,90.01) and (222.32,94.65) .. (222.31,100.38) -- cycle ;
\draw  [fill={rgb, 255:red, 255; green, 255; blue, 255 }  ,fill opacity=1 ] (61.91,174.4) .. controls (61.91,168.68) and (66.47,164.04) .. (72.1,164.04) .. controls (77.73,164.04) and (82.29,168.68) .. (82.29,174.4) .. controls (82.29,180.12) and (77.73,184.76) .. (72.1,184.76) .. controls (66.47,184.76) and (61.91,180.12) .. (61.91,174.4) -- cycle ;
\draw  [fill={rgb, 255:red, 255; green, 255; blue, 255 }  ,fill opacity=1 ] (107.91,174.4) .. controls (107.91,168.68) and (112.47,164.04) .. (118.1,164.04) .. controls (123.73,164.04) and (128.29,168.68) .. (128.29,174.4) .. controls (128.29,180.12) and (123.73,184.76) .. (118.1,184.76) .. controls (112.47,184.76) and (107.91,180.12) .. (107.91,174.4) -- cycle ;
\draw  [fill={rgb, 255:red, 255; green, 255; blue, 255 }  ,fill opacity=1 ] (155.81,138) .. controls (155.81,132.28) and (160.37,127.64) .. (166,127.64) .. controls (171.63,127.64) and (176.19,132.28) .. (176.19,138) .. controls (176.19,143.72) and (171.63,148.36) .. (166,148.36) .. controls (160.37,148.36) and (155.81,143.72) .. (155.81,138) -- cycle ;

\draw    (72,138) .. controls (-34.9,-6.8) and (309.1,-9.8) .. (214,138) ;

\draw    (285,139) -- (321,139) ;
\draw    (321,139) -- (369,139) ;
\draw    (275,139) -- (275.1,175.4) ;
\draw    (321,139) -- (321.1,175.4) ;
\draw    (370,139) -- (370.1,175.4) ;
\draw    (321.16,138.76) -- (321.12,102.36) ;
\draw  [fill={rgb, 255:red, 255; green, 255; blue, 255 }  ,fill opacity=1 ] (331.31,102.38) .. controls (331.31,108.1) and (326.73,112.73) .. (321.1,112.72) .. controls (315.47,112.71) and (310.92,108.07) .. (310.93,102.34) .. controls (310.93,96.62) and (315.51,91.99) .. (321.14,92) .. controls (326.77,92.01) and (331.32,96.65) .. (331.31,102.38) -- cycle ;
\draw  [fill={rgb, 255:red, 0; green, 0; blue, 0 }  ,fill opacity=1 ] (310.81,139) .. controls (310.81,133.28) and (315.37,128.64) .. (321,128.64) .. controls (326.63,128.64) and (331.19,133.28) .. (331.19,139) .. controls (331.19,144.72) and (326.63,149.36) .. (321,149.36) .. controls (315.37,149.36) and (310.81,144.72) .. (310.81,139) -- cycle ;
\draw    (368.16,137.76) -- (368.12,101.36) ;
\draw  [fill={rgb, 255:red, 255; green, 255; blue, 255 }  ,fill opacity=1 ] (378.31,101.38) .. controls (378.31,107.1) and (373.73,111.73) .. (368.1,111.72) .. controls (362.47,111.71) and (357.92,107.07) .. (357.93,101.34) .. controls (357.93,95.62) and (362.51,90.99) .. (368.14,91) .. controls (373.77,91.01) and (378.32,95.65) .. (378.31,101.38) -- cycle ;
\draw    (275.16,137.76) -- (275.12,101.36) ;
\draw  [fill={rgb, 255:red, 0; green, 0; blue, 0 }  ,fill opacity=1 ] (285.31,101.38) .. controls (285.31,107.1) and (280.73,111.73) .. (275.1,111.72) .. controls (269.47,111.71) and (264.92,107.07) .. (264.93,101.34) .. controls (264.93,95.62) and (269.51,90.99) .. (275.14,91) .. controls (280.77,91.01) and (285.32,95.65) .. (285.31,101.38) -- cycle ;
\draw  [fill={rgb, 255:red, 0; green, 0; blue, 0 }  ,fill opacity=1 ] (264.81,139) .. controls (264.81,133.28) and (269.37,128.64) .. (275,128.64) .. controls (280.63,128.64) and (285.19,133.28) .. (285.19,139) .. controls (285.19,144.72) and (280.63,149.36) .. (275,149.36) .. controls (269.37,149.36) and (264.81,144.72) .. (264.81,139) -- cycle ;
\draw    (368,139) -- (416,139) ;
\draw    (417,139) -- (417.1,175.4) ;
\draw    (415.16,137.76) -- (415.12,101.36) ;
\draw  [fill={rgb, 255:red, 255; green, 255; blue, 255 }  ,fill opacity=1 ] (425.31,101.38) .. controls (425.31,107.1) and (420.73,111.73) .. (415.1,111.72) .. controls (409.47,111.71) and (404.92,107.07) .. (404.93,101.34) .. controls (404.93,95.62) and (409.51,90.99) .. (415.14,91) .. controls (420.77,91.01) and (425.32,95.65) .. (425.31,101.38) -- cycle ;
\draw  [fill={rgb, 255:red, 0; green, 0; blue, 0 }  ,fill opacity=1 ] (264.91,175.4) .. controls (264.91,169.68) and (269.47,165.04) .. (275.1,165.04) .. controls (280.73,165.04) and (285.29,169.68) .. (285.29,175.4) .. controls (285.29,181.12) and (280.73,185.76) .. (275.1,185.76) .. controls (269.47,185.76) and (264.91,181.12) .. (264.91,175.4) -- cycle ;
\draw  [fill={rgb, 255:red, 255; green, 255; blue, 255 }  ,fill opacity=1 ] (310.91,175.4) .. controls (310.91,169.68) and (315.47,165.04) .. (321.1,165.04) .. controls (326.73,165.04) and (331.29,169.68) .. (331.29,175.4) .. controls (331.29,181.12) and (326.73,185.76) .. (321.1,185.76) .. controls (315.47,185.76) and (310.91,181.12) .. (310.91,175.4) -- cycle ;
\draw  [fill={rgb, 255:red, 255; green, 255; blue, 255 }  ,fill opacity=1 ] (358.81,139) .. controls (358.81,133.28) and (363.37,128.64) .. (369,128.64) .. controls (374.63,128.64) and (379.19,133.28) .. (379.19,139) .. controls (379.19,144.72) and (374.63,149.36) .. (369,149.36) .. controls (363.37,149.36) and (358.81,144.72) .. (358.81,139) -- cycle ;
\draw    (275,139) .. controls (168.1,-5.8) and (512.1,-8.8) .. (417,139) ;
\draw  [fill={rgb, 255:red, 255; green, 255; blue, 255 }  ,fill opacity=1 ] (359.91,175.4) .. controls (359.91,169.68) and (364.47,165.04) .. (370.1,165.04) .. controls (375.73,165.04) and (380.29,169.68) .. (380.29,175.4) .. controls (380.29,181.12) and (375.73,185.76) .. (370.1,185.76) .. controls (364.47,185.76) and (359.91,181.12) .. (359.91,175.4) -- cycle ;
\draw  [fill={rgb, 255:red, 255; green, 255; blue, 255 }  ,fill opacity=1 ] (406.91,175.4) .. controls (406.91,169.68) and (411.47,165.04) .. (417.1,165.04) .. controls (422.73,165.04) and (427.29,169.68) .. (427.29,175.4) .. controls (427.29,181.12) and (422.73,185.76) .. (417.1,185.76) .. controls (411.47,185.76) and (406.91,181.12) .. (406.91,175.4) -- cycle ;
\draw  [fill={rgb, 255:red, 255; green, 255; blue, 255 }  ,fill opacity=1 ] (405.81,139) .. controls (405.81,133.28) and (410.37,128.64) .. (416,128.64) .. controls (421.63,128.64) and (426.19,133.28) .. (426.19,139) .. controls (426.19,144.72) and (421.63,149.36) .. (416,149.36) .. controls (410.37,149.36) and (405.81,144.72) .. (405.81,139) -- cycle ;
\draw    (477,139) -- (513,139) ;
\draw    (513,139) -- (561,139) ;
\draw    (467,139) -- (467.1,175.4) ;
\draw    (513,139) -- (513.1,175.4) ;
\draw    (562,139) -- (562.1,175.4) ;
\draw    (513.16,138.76) -- (513.12,102.36) ;
\draw  [fill={rgb, 255:red, 0; green, 0; blue, 0 }  ,fill opacity=1 ] (523.31,102.38) .. controls (523.31,108.1) and (518.73,112.73) .. (513.1,112.72) .. controls (507.47,112.71) and (502.92,108.07) .. (502.93,102.34) .. controls (502.93,96.62) and (507.51,91.99) .. (513.14,92) .. controls (518.77,92.01) and (523.32,96.65) .. (523.31,102.38) -- cycle ;
\draw  [fill={rgb, 255:red, 0; green, 0; blue, 0 }  ,fill opacity=1 ] (502.81,139) .. controls (502.81,133.28) and (507.37,128.64) .. (513,128.64) .. controls (518.63,128.64) and (523.19,133.28) .. (523.19,139) .. controls (523.19,144.72) and (518.63,149.36) .. (513,149.36) .. controls (507.37,149.36) and (502.81,144.72) .. (502.81,139) -- cycle ;
\draw    (560.16,137.76) -- (560.12,101.36) ;
\draw  [fill={rgb, 255:red, 255; green, 255; blue, 255 }  ,fill opacity=1 ] (570.31,101.38) .. controls (570.31,107.1) and (565.73,111.73) .. (560.1,111.72) .. controls (554.47,111.71) and (549.92,107.07) .. (549.93,101.34) .. controls (549.93,95.62) and (554.51,90.99) .. (560.14,91) .. controls (565.77,91.01) and (570.32,95.65) .. (570.31,101.38) -- cycle ;
\draw    (467.16,137.76) -- (467.12,101.36) ;
\draw  [fill={rgb, 255:red, 0; green, 0; blue, 0 }  ,fill opacity=1 ] (477.31,101.38) .. controls (477.31,107.1) and (472.73,111.73) .. (467.1,111.72) .. controls (461.47,111.71) and (456.92,107.07) .. (456.93,101.34) .. controls (456.93,95.62) and (461.51,90.99) .. (467.14,91) .. controls (472.77,91.01) and (477.32,95.65) .. (477.31,101.38) -- cycle ;
\draw    (560,139) -- (608,139) ;
\draw    (609,139) -- (609.1,175.4) ;
\draw    (607.16,137.76) -- (607.12,101.36) ;
\draw  [fill={rgb, 255:red, 255; green, 255; blue, 255 }  ,fill opacity=1 ] (617.31,101.38) .. controls (617.31,107.1) and (612.73,111.73) .. (607.1,111.72) .. controls (601.47,111.71) and (596.92,107.07) .. (596.93,101.34) .. controls (596.93,95.62) and (601.51,90.99) .. (607.14,91) .. controls (612.77,91.01) and (617.32,95.65) .. (617.31,101.38) -- cycle ;
\draw  [fill={rgb, 255:red, 0; green, 0; blue, 0 }  ,fill opacity=1 ] (456.91,175.4) .. controls (456.91,169.68) and (461.47,165.04) .. (467.1,165.04) .. controls (472.73,165.04) and (477.29,169.68) .. (477.29,175.4) .. controls (477.29,181.12) and (472.73,185.76) .. (467.1,185.76) .. controls (461.47,185.76) and (456.91,181.12) .. (456.91,175.4) -- cycle ;
\draw  [fill={rgb, 255:red, 0; green, 0; blue, 0 }  ,fill opacity=1 ] (502.91,175.4) .. controls (502.91,169.68) and (507.47,165.04) .. (513.1,165.04) .. controls (518.73,165.04) and (523.29,169.68) .. (523.29,175.4) .. controls (523.29,181.12) and (518.73,185.76) .. (513.1,185.76) .. controls (507.47,185.76) and (502.91,181.12) .. (502.91,175.4) -- cycle ;
\draw  [fill={rgb, 255:red, 0; green, 0; blue, 0 }  ,fill opacity=1 ] (550.81,139) .. controls (550.81,133.28) and (555.37,128.64) .. (561,128.64) .. controls (566.63,128.64) and (571.19,133.28) .. (571.19,139) .. controls (571.19,144.72) and (566.63,149.36) .. (561,149.36) .. controls (555.37,149.36) and (550.81,144.72) .. (550.81,139) -- cycle ;
\draw    (467,139) .. controls (360.1,-5.8) and (704.1,-8.8) .. (609,139) ;
\draw  [fill={rgb, 255:red, 255; green, 255; blue, 255 }  ,fill opacity=1 ] (551.91,175.4) .. controls (551.91,169.68) and (556.47,165.04) .. (562.1,165.04) .. controls (567.73,165.04) and (572.29,169.68) .. (572.29,175.4) .. controls (572.29,181.12) and (567.73,185.76) .. (562.1,185.76) .. controls (556.47,185.76) and (551.91,181.12) .. (551.91,175.4) -- cycle ;
\draw  [fill={rgb, 255:red, 255; green, 255; blue, 255 }  ,fill opacity=1 ] (598.91,175.4) .. controls (598.91,169.68) and (603.47,165.04) .. (609.1,165.04) .. controls (614.73,165.04) and (619.29,169.68) .. (619.29,175.4) .. controls (619.29,181.12) and (614.73,185.76) .. (609.1,185.76) .. controls (603.47,185.76) and (598.91,181.12) .. (598.91,175.4) -- cycle ;
\draw  [fill={rgb, 255:red, 255; green, 255; blue, 255 }  ,fill opacity=1 ] (597.81,139) .. controls (597.81,133.28) and (602.37,128.64) .. (608,128.64) .. controls (613.63,128.64) and (618.19,133.28) .. (618.19,139) .. controls (618.19,144.72) and (613.63,149.36) .. (608,149.36) .. controls (602.37,149.36) and (597.81,144.72) .. (597.81,139) -- cycle ;
\draw  [fill={rgb, 255:red, 255; green, 255; blue, 255 }  ,fill opacity=1 ] (456.81,139) .. controls (456.81,133.28) and (461.37,128.64) .. (467,128.64) .. controls (472.63,128.64) and (477.19,133.28) .. (477.19,139) .. controls (477.19,144.72) and (472.63,149.36) .. (467,149.36) .. controls (461.37,149.36) and (456.81,144.72) .. (456.81,139) -- cycle ;

\draw    (117,42.2) -- (161,42.2) -- (161,66.2) -- (117,66.2) -- cycle  ;
\draw (120,46.6) node [anchor=north west][inner sep=0.75pt]    {$t=0$};
\draw    (323,42.2) -- (367,42.2) -- (367,66.2) -- (323,66.2) -- cycle  ;
\draw (326,46.6) node [anchor=north west][inner sep=0.75pt]    {$t=1$};
\draw    (513,42.2) -- (557,42.2) -- (557,66.2) -- (513,66.2) -- cycle  ;
\draw (516,46.6) node [anchor=north west][inner sep=0.75pt]    {$t=2$};

\end{tikzpicture}
\caption{Dynamics of the interval rule $[1,2]$. Black cells are in state $1$ and white cells are in state $0$.}
\label{fig:dymint}
\end{figure}
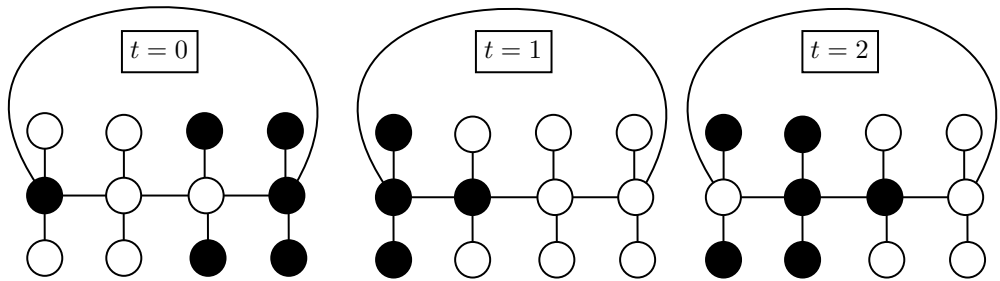

\item \textbf{Reaction-difussion networks.}  This family is similar to the case of Threshold networks. However, there is a non-trivial set of possible labels for the edges. In this case we have that $Q = \alpha \cup \epsilon$ where $\alpha = \{0,1\}$ and $\epsilon = \{2,\hdots, q'\}$ and $q'\geq2.$ Roughly speaking, $Q$ has reluctant states $\{0\} \cup \epsilon$ and an active state $1.$ If a node reaches state $1$ it will automatically go to the next state until it will reach $q'$. However, while the node goes from $2$ to $q'$ it will be considered as inactive for its neighbors. When a node $v$ is in state $0$ it needs at least $\theta_v$ active neighbors to turn into $1$. More precisely, we define these CSAN networks by asking $\rho(e)\}$ to be the map $a\mapsto \begin{cases} 1 & \text{ if } a=1\\ 0 & \text{otherwise} \end{cases}$ and $\lambda_{\theta_v}(x_v,m((x|_{N(v)}))) = \begin{cases}

1 & \text{ if }  m((x|_{N(v)})^{\rho})[1] - \theta_v \geq 0 \wedge x_v = 0  \\
q+1 & \text{ if } x_v \not = 0\\
0 & \text{ otherwise. }
\end{cases}$

An example of a reaction-diffusion dynamics is given in Figure \ref{fig:dymreact}. In this case, there is only one reluctant state ($q'=2$) and all the thresholds are set to $0$ as same as in the example shown in Figure \ref{fig:dymthresh}. In addition, observe that the graph is the same graph shown in Figure \ref{fig:dymthresh}. However, as reaction-diffusion dynamics considers a reluctant state, the dynamics reaches a uniform fixed point contrary to the case of Figure \ref{fig:dymthresh} in which the dynamics reaches an attractor of period two, starting from the same initial condition.

\begin{figure}
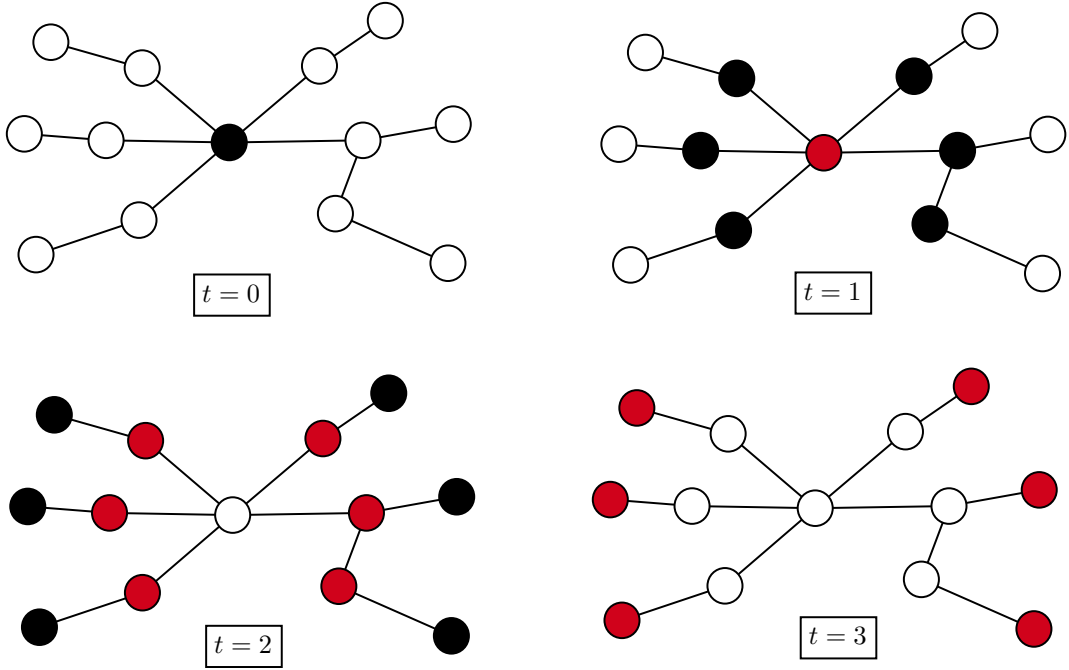


\centering
\tikzset{every picture/.style={line width=0.75pt}} 


\caption{An example of a dynamics of a reaction-difussion network with only one reluctant state and $\theta_v = 1$ for any node in the network. Nodes in black are nodes in state $1$, nodes in red are in reluctant state and nodes in white are in state $0$.}
\label{fig:dymreact}
\end{figure}


\end{enumerate}

\section{Representing Automata Networks}
\label{sec:representation}
As we are interested in measuring computational complexity of decision problems related to the dynamics of automata networks belonging to a particular family, we introduce hereunder a general notion for the representation of a family of automata networks.
We can always fix a canonical representation of automata networks as Boolean circuits.
However, as we show in the first part of this section, families can have different natural representation which are closely related to their particular properties.
Considering this fact, we introduce the notion of \textit{standard representation} in order to denote some representation from which we can efficiently obtain a circuit family computing original automata network family.
Finally, we resume previous discussion on different representations for some particular families, showing how difficult it is to transform one particular representation into another one.



\subsection{Standard representations}

\sloppy We fix for any alphabet $Q$ an injective map ${m_Q : Q\to \{0,1\}^{k_Q}}$ which we extend cell-wise for each $n$ to ${m_Q:Q^n\to\{0,1\}^{k_Qn}}$. Given an abstract automata network ${F:Q^n\to Q^n}$, a circuit encoding of $F$ is a Boolean circuit ${C:\{0,1\}^{k_Qn}\to \{0,1\}^{k_Qn}}$ such that ${m_Q\circ F = C\circ m_Q}$ on ${Q^n}$. We also fix a canonical way to represent circuits as words of ${\{0,1\}^*}$ (for instance given by a number of vertices, the list of gate type positioned at each vertex and the adjacency matrix of the graph of the circuit).




\begin{definition}\label{def:standdef}
Let $\mathcal{F}$ be a set of abstract automata network over alphabet $Q$. A standard representation $\mathcal{F}^*$ for $\mathcal{F}$ is a language $L_{\mathcal{F}} \subseteq \{0,1\}^*$ together with a \DLOG{} algorithm such that:
\begin{itemize}
\item the algorithm transforms any $w \in L_{\mathcal{F}}$ into the canonical
  representation of a circuit encoding ${C(w)}$ that code an abstract
  automata network ${F_w\in\mathcal{F}}$;
\item for any ${F\in\mathcal{F}}$ there is ${w\in L_{\mathcal{F}}}$ with ${F=F_w}$.
\end{itemize}
\end{definition}
The default general representations we will use are circuit representations, \textit{i.e.} representations where $w\in L_{\mathcal{F}}$ is just a canonical representation of a circuit. In this case the \DLOG{} algorithm is trivial (the identity map). However, we sometimes want to work with more concrete and natural representations for some families of networks: in such a case, the above definition allows any kind of coding as soon as it is easy to deduce the canonical circuit representation from it.

\subsection{Example of standard representations of some particular families}

Observe that communication graph is often an essential piece of information for describing an automata network, however, this information is usually not enough.
In this section three examples of canonical types of families are discussed: bounded degree networks, CSAN and algebraic families.
They all have the property that a complete description of an automata network can be done in a size comparable to the communication graph (more precisely, polynomial in the number of nodes), \textit{i.e.} much less than the set of all possible configurations.

\paragraph{The CSAN case}
A CSAN family is a collection of labeled graphs and thus is naturally represented as a graph $G$ together with some representations of local functions (i.e. $\lambda$ and $\rho$).
At each node $v\in V$ the local function $\lambda(v)$ will only be applied to pairs made of a state and a a $k$-bounded multiset of states where $k$ is the degree of $v$.
In any case, it holds ${k\leq |V|=n}$ so it is sufficient to consider only ${n}$-bounded multisets in order to completely specify the global map of the CSAN (Definition~\ref{def:globalcsan}).
For a fixed alphabet $Q$, there are ${(n+1)^{|Q|}}$ distinct ${n}$-bounded multisets so any $\lambda(v)$ can be described via a table of states of $Q$ of size polynomial in the graph $G$.
There are only finitely many possible $\rho$ maps and we represent them by an arbitrary numbering.
From the graph $G$, the table describing $\lambda(v)$ and the $\rho$ labels of adjacent edges, it is not difficult to produce in \DLOG{} a circuit that computes the transition map $F_v$ of node $v$ following Definition~\ref{def:globalcsan}.
We therefore have a standard representation for CSAN families where the encoding of an automata network is made of the encoding of $G$, the encoding of $\lambda(v)$ for each node $v$ and the encoding of $\rho_e$ for each edge $e$.
This encoding is of size polynomial in the number of nodes.

\paragraph{The case of bounded degree communication graphs}
Let us fix some positive constant $\Delta$. It is natural to consider the family of automata networks whose interaction graph has a maximum degree bounded by $\Delta$ (see Remark~\ref{rem:represent-g-networks} below). We associate to this family the following representation: an automata network $F$ is given as pair ${(G,(\tau_v)_{v\in V(g)})}$ where $G$ is a communication graph of $F$ of maximum degree at most $\Delta$ and ${(\tau_v)_{v\in V(G)}}$ is the list for all nodes of $G$ of its local transition map $F_v$ of the form ${Q^d\to Q}$ for ${d\leq\Delta}$ and represented as a plain transition table of size ${|Q|^d\log(|Q|)}$.

\begin{remark}\label{rem:csan-bounded-degree}
  Given any CSAN family, there is a \DLOG{} algorithm that transforms a bounded degree representation of an automata network of the family into a CSAN representation: indeed in this case all local maps are bounded objects, so it is just a matter of making a bounded computation for each node.
\end{remark}

\paragraph{The algebraic case}
When endowing the alphabet $Q$ with a finite field structure, the set of configurations $Q^n$ is a vector space and one can consider automata networks that are actually linear maps. In this case the natural representation is a ${n\times n}$ matrix. It is clearly a standard representation in the above sense since circuit encodings can be easily computed from the matrix. Moreover, as in the CSAN case, when a linear automata network is given as a bounded degree representation, it is easy to recover a matrix in \DLOG{}.

More generally, we can consider matrix representations without field structure on the alphabet. An interesting case is that of Boolean matrices: ${Q=\{0,1\}}$ is endowed with the standard Boolean algebra structure with operations ${\vee,\wedge}$ and matrix multiplication is defined by: 
\[\bigl(AB)_{i,j}=\bigvee_kA_{i,k}\wedge B_{k,j}.\]
They are a standard representation of disjunctive networks (and by switching the role of $0$ and $1$ conjunctive networks), \textit{i.e.} networks $F$ over alphabet ${\{0,1\}}$ whose local maps are of the form ${F_i(x) = \vee_{k\in N(i)}x_k}$ (respectively ${F_i(x) = \wedge_{k\in N(i)}x_k}$) . When their dependency graph is symmetric, disjunctive networks (resp. conjunctive networks) are a particular case of CSAN networks for which $\rho_v$ maps are the identity and $\lambda(v)$ are just max (resp. min) maps. For disjunctive networks (resp. conjunctive networks) the CSAN representation and the matrix representation are \DLOG{} equivalent. 


\subsection{Computing interaction graphs from representations}

One of the key differences between all the representations presented so far is in the information they give about the interaction graph of an automata network. For instance, it is straightforward to deduce the interaction graph of a linear network from its matrix representation in \DLOG{}: the non-zero entries of the matrix gives the edges of the interaction graph.
The situation doesn't change much if the linear network is given by a circuit representation: it is sufficient to evaluate the circuit on the $n$ input configurations that form a base of the vector space ${Q^n}$ to completely know the matrix of the linear network.

At the other extreme, one can see that it is NP-hard to decide whether a given edge belongs to the interaction graph of an automata network given by a circuit representation: indeed, one can build in \DLOG{} from any SAT formula $\phi$ with $n$ variables a circuit representation of an automata network ${F:\{0,1\}^{n+1}\to\{0,1\}^{n+1}}$ with 
\[F(x)_1 =
  \begin{cases}
    x_{n+1} &\text{ if }\phi(x_1,\ldots,x_n)\text{ is true,}\\
    0 &\text{ else.}
  \end{cases}
\]
This $F$ is such that node $1$ depends on node ${n+1}$ if and only if $\phi$ is satisfiable.

For automata networks with communication graphs of degree at most $\Delta$, there is a polynomial time algorithm to compute the interaction graph from a circuit representation: for each node $v$, try all the possible subsets $S$ of nodes of size at most $\Delta$ and find the largest one such that the following map 
\[x\in Q^S\mapsto F_v(\phi(x))\] effectively depends on each node of $S$, where ${\phi(x)_w}$ is $x_{w}$ if ${w\in S}$ and some arbitrary fixed state ${q\in Q}$ else. Note also, that we can compute a bounded degree representation in polynomial time with the same idea.

In the CSAN case, the situation is ambivalent. On one hand, the interaction graph can be computed in polynomial time (in the number of nodes) from a CSAN representation because for any given node $v$ there is only polynomially many possible multiset of the form ${\mset{(x|_{N(v)})^{\rho}}}$ that can appear in the neighborhood of $v$ and, for each neighbor $v'$ of $v$, we can compute (in polynomial time) the set of pairs of such multisets that can be realized by changing just the state of $v'$. This allows to determine whether $v$ depends on $v'$ by verifying whether any such pair can change the state of $v$ through the local map $\lambda_v$ . The dependence of $v$ on itself is also easy to check once the set of possible multisets is computed.

On the other hand, a polynomial time algorithm to compute the interaction graph from a circuit representation would give a polynomial algorithm solving Unambiguous-SAT (which is very unlikely following Valiant-Vazirani theorem \cite{Valiant_1986}). Indeed, any ``dirac'' map ${\delta : \{0,1\}^n\to\{0,1\}}$ with ${\delta(x)= 1}$ if and only if ${x_1\cdots x_n = b_1\cdots b_n}$ can be seen as the local map of a CSAN network because it can be written as ${\gamma(\{\rho_i(x_i) : 1\leq i\leq n\})}$ where $\rho_i(x_i)=x_i$ if $b_i=1$ and $\neg x_i$ else, and $\gamma$ is the map ${[2]^{\{0,1\}}\to\{0,1\}}$
\[S\subseteq Q\mapsto
  \begin{cases}
    1 &\text{ if }S=\{1\}\\
    0 &\text{ else.}
  \end{cases}
\]
A constant map can also be seen as the local map of some CSAN network.
Therefore, given a Boolean formula $\phi$ with the promise that is at has at most one satisfying assignment, one can easily compute the circuit representation of some CSAN network which has some edge in its interaction graph if and only if $\phi$ is satisfiable: indeed, the construction of $F$ above from a SAT formula always produce a CSAN given the promise on $\phi$.

It follows from the discussion above that a polynomial time algorithm to compute a CSAN representation of a CSAN represented by circuit would give a polynomial time algorithm to solve Unambiguous-SAT.

The following table synthesizes the computational hardness of representation conversions. It shall be read as follows: given a family ${(\mathcal{F},\mathcal{F}^*)}$ listed horizontally and a family ${(\mathcal{H},\mathcal{H}^*)}$ listed vertically, the corresponding entry in the table indicates the complexity of the problem of transforming ${w\in L_F}$ with the promise that ${F_w\in\mathcal{F}\cap\mathcal{H}}$ into $w'\in L_H$ such that ${F_w = H_{w'}}$.

\begin{center}
\begin{tabular}{c|c|c|c|c}
  \diagbox{output}{input}&circuit&CSAN&$\Delta$-bounded degree&matrix\\
  \hline
  circuit&trivial &\DLOG{} &\DLOG{}&\DLOG{}\\
  CSAN&USAT-hard&trivial&\DLOG{}&\DLOG{}\\
  $\Delta$-bounded degree&PTIME&PTIME&trivial&\DLOG{}\\
  matrix&PTIME&\DLOG{}&\DLOG{}&trivial\\
\end{tabular}
\end{center}

USAT-hard means that any PTIME algorithm would imply a PTIME algorithm for Unambiguous-SAT.

\section{Simulation and universality}
\label{sec:simuniv}

In this section we introduce a key tool used in this paper: simulations.
The goal is to easily prove computational or dynamical complexity of some family of automata networks by showing it can simulate some well-known reference family where the complexity analysis is already established.
It can be thought as a complexity or dynamical reduction.
Simulations of various kinds are often implicitly used in proofs of dynamical or computational hardness.
We are going instead to explicitly define a notion of simulation and establish hardness results as corollaries of simulation results later in the paper.
To be more precise, we will first define a notion of simulation between individual automata networks, and then extend it to a notion of simulation between families.
This latter notion, which is the one we are really interested in requires more care if we want to use it as a notion of reduction for computational complexity.
We introduce all the useful concepts progressively in the next subsections.

\subsection{Simulation between individual automata networks}

At the core of our formalism is the following definition of simulation where an automata network $F$ is simulated by an automata network $G$ with a constant time slowdown and using blocks of nodes in $G$ to represent nodes in $F$.
Our definition is rather strict and requires in particular an injective encoding of configurations of $F$ into configurations of $G$.
We are not aware of a published work with this exact same formal definition, but close variants certainly exist and it is a direct adaptation to finite automata networks of a classical definition of simulation for cellular automata \cite{bulk2}.

\begin{definition}\label{def:bloc-simu}
  Let ${F:Q_F^{V_F}\rightarrow Q_F^{V_F}}$ and ${G:Q_G^{V_G}\rightarrow Q_G^{V_G}}$ be abstract automata networks.
  A \emph{block embedding} of $Q_F^{V_F}$ into $Q_G^{V_G}$ is a collection of blocks ${D_i\subseteq V_G}$ for each ${i\in V_F}$ which forms a partition of ${V_G}$ together with a collection of patterns ${p_{i,q}\in Q_G^{D_i}}$ for each ${i\in V_F}$ and each ${q\in Q_F}$ such that ${p_{i,q}=p_{i,q'}}$ implies ${q=q'}$.
  This defines an injective map ${\phi:Q_F^{V_F}\rightarrow Q_G^{V_G}}$ by ${\phi(x)_{D_i} = p_{i,x_i}}$ for each $i\in V_F$.
  We say that $G$ simulates $F$ via block embedding $\phi$ if there is a time constant $T$ such that the following holds on ${Q_F^{V_F}}$: 
  \[\phi\circ F = G^T\circ\phi.\]
\end{definition}

See Figure  \ref{fig:bsimulation} for a scheme of block simulation. In the following, when useful we represent a block embedding as the list of blocks together with the list of patterns. The size of this representation is linear in the number of nodes (for fixed  alphabet). 


\begin{figure}

\centering
\begin{tikzpicture}[x=0.55pt,y=0.55pt,yscale=-1,xscale=1]

\draw  [fill={rgb, 255:red, 255; green, 255; blue, 255 }  ,fill opacity=1 ] (552,188) -- (637.5,188) -- (637.5,242) -- (552,242) -- cycle ;
\draw  [fill={rgb, 255:red, 255; green, 255; blue, 255 }  ,fill opacity=1 ] (524.25,25) -- (617.5,25) -- (617.5,80.75) -- (524.25,80.75) -- cycle ;
\draw  [fill={rgb, 255:red, 255; green, 255; blue, 255 }  ,fill opacity=1 ] (288,73) -- (373.5,73) -- (373.5,122) -- (288,122) -- cycle ;
\draw  [fill={rgb, 255:red, 255; green, 255; blue, 255 }  ,fill opacity=1 ] (362.5,163) -- (454.5,163) -- (454.5,212) -- (362.5,212) -- cycle ;
\draw  [dash pattern={on 4.5pt off 4.5pt}] (43.5,40) .. controls (63.5,30) and (167,56) .. (147,76) .. controls (127,96) and (198.5,164) .. (218.5,194) .. controls (238.5,224) and (35.5,199) .. (15.5,169) .. controls (-4.5,139) and (23.5,50) .. (43.5,40) -- cycle ;
\draw  [fill={rgb, 255:red, 0; green, 0; blue, 0 }  ,fill opacity=1 ] (18,102.75) .. controls (18,98.47) and (21.47,95) .. (25.75,95) .. controls (30.03,95) and (33.5,98.47) .. (33.5,102.75) .. controls (33.5,107.03) and (30.03,110.5) .. (25.75,110.5) .. controls (21.47,110.5) and (18,107.03) .. (18,102.75) -- cycle ;
\draw  [fill={rgb, 255:red, 0; green, 0; blue, 0 }  ,fill opacity=1 ] (64,146.75) .. controls (64,142.47) and (67.47,139) .. (71.75,139) .. controls (76.03,139) and (79.5,142.47) .. (79.5,146.75) .. controls (79.5,151.03) and (76.03,154.5) .. (71.75,154.5) .. controls (67.47,154.5) and (64,151.03) .. (64,146.75) -- cycle ;
\draw  [fill={rgb, 255:red, 0; green, 0; blue, 0 }  ,fill opacity=1 ] (119,97.75) .. controls (119,93.47) and (122.47,90) .. (126.75,90) .. controls (131.03,90) and (134.5,93.47) .. (134.5,97.75) .. controls (134.5,102.03) and (131.03,105.5) .. (126.75,105.5) .. controls (122.47,105.5) and (119,102.03) .. (119,97.75) -- cycle ;
\draw  [fill={rgb, 255:red, 0; green, 0; blue, 0 }  ,fill opacity=1 ] (162,170.75) .. controls (162,166.47) and (165.47,163) .. (169.75,163) .. controls (174.03,163) and (177.5,166.47) .. (177.5,170.75) .. controls (177.5,175.03) and (174.03,178.5) .. (169.75,178.5) .. controls (165.47,178.5) and (162,175.03) .. (162,170.75) -- cycle ;
\draw    (25.75,102.75) -- (71.75,146.75) ;
\draw    (71.75,146.75) -- (126.75,97.75) ;
\draw    (126.75,97.75) -- (169.75,170.75) ;
\draw    (71.75,146.75) -- (169.75,170.75) ;
\draw    (25.75,102.75) -- (126.75,97.75) ;

\draw [color={rgb, 255:red, 74; green, 144; blue, 226 }  ,draw opacity=1 ][line width=2.25]    (541.75,71.75) -- (571.5,200) ;
\draw [color={rgb, 255:red, 74; green, 144; blue, 226 }  ,draw opacity=1 ][line width=2.25]    (441.75,184.75) -- (571.5,200) ;
\draw [color={rgb, 255:red, 74; green, 144; blue, 226 }  ,draw opacity=1 ][line width=2.25]    (441.75,184.75) -- (541.75,71.75) ;
\draw [color={rgb, 255:red, 74; green, 144; blue, 226 }  ,draw opacity=1 ][line width=2.25]    (354.5,95) -- (541.75,71.75) ;
\draw [color={rgb, 255:red, 74; green, 144; blue, 226 }  ,draw opacity=1 ][line width=2.25]    (354.5,95) -- (441.75,184.75) ;
\draw    (324.25,113.25) -- (352.25,94.25) ;
\draw    (303.25,92.25) -- (352.25,94.25) ;
\draw    (303.25,92.25) -- (324.25,113.25) ;
\draw  [fill={rgb, 255:red, 155; green, 155; blue, 155 }  ,fill opacity=1 ] (297,92.25) .. controls (297,88.8) and (299.8,86) .. (303.25,86) .. controls (306.7,86) and (309.5,88.8) .. (309.5,92.25) .. controls (309.5,95.7) and (306.7,98.5) .. (303.25,98.5) .. controls (299.8,98.5) and (297,95.7) .. (297,92.25) -- cycle ;
\draw    (411.25,201.25) -- (441.75,184.75) ;
\draw    (381.25,173.25) -- (439.25,182.25) ;
\draw  [fill={rgb, 255:red, 155; green, 155; blue, 155 }  ,fill opacity=1 ] (435.5,184.75) .. controls (435.5,181.3) and (438.3,178.5) .. (441.75,178.5) .. controls (445.2,178.5) and (448,181.3) .. (448,184.75) .. controls (448,188.2) and (445.2,191) .. (441.75,191) .. controls (438.3,191) and (435.5,188.2) .. (435.5,184.75) -- cycle ;
\draw    (381.25,173.25) -- (376.5,190) ;
\draw    (568.25,199.25) -- (612.82,201.25) ;
\draw    (571.5,200) -- (585,225.75) ;
\draw  [fill={rgb, 255:red, 155; green, 155; blue, 155 }  ,fill opacity=1 ] (565.25,200) .. controls (565.25,196.55) and (568.05,193.75) .. (571.5,193.75) .. controls (574.95,193.75) and (577.75,196.55) .. (577.75,200) .. controls (577.75,203.45) and (574.95,206.25) .. (571.5,206.25) .. controls (568.05,206.25) and (565.25,203.45) .. (565.25,200) -- cycle ;
\draw    (573.25,56.25) -- (601.25,37.25) ;
\draw    (541.75,71.75) -- (571.5,57) ;
\draw  [fill={rgb, 255:red, 155; green, 155; blue, 155 }  ,fill opacity=1 ] (535.5,71.75) .. controls (535.5,68.3) and (538.3,65.5) .. (541.75,65.5) .. controls (545.2,65.5) and (548,68.3) .. (548,71.75) .. controls (548,75.2) and (545.2,78) .. (541.75,78) .. controls (538.3,78) and (535.5,75.2) .. (535.5,71.75) -- cycle ;
\draw    (376.5,196.25) -- (411.25,201.25) ;
\draw  [fill={rgb, 255:red, 155; green, 155; blue, 155 }  ,fill opacity=1 ] (370.25,196.25) .. controls (370.25,192.8) and (373.05,190) .. (376.5,190) .. controls (379.95,190) and (382.75,192.8) .. (382.75,196.25) .. controls (382.75,199.7) and (379.95,202.5) .. (376.5,202.5) .. controls (373.05,202.5) and (370.25,199.7) .. (370.25,196.25) -- cycle ;
\draw  [fill={rgb, 255:red, 155; green, 155; blue, 155 }  ,fill opacity=1 ] (405,201.25) .. controls (405,197.8) and (407.8,195) .. (411.25,195) .. controls (414.7,195) and (417.5,197.8) .. (417.5,201.25) .. controls (417.5,204.7) and (414.7,207.5) .. (411.25,207.5) .. controls (407.8,207.5) and (405,204.7) .. (405,201.25) -- cycle ;
\draw    (585.25,230.25) -- (613.25,230.25) ;
\draw  [fill={rgb, 255:red, 155; green, 155; blue, 155 }  ,fill opacity=1 ] (579,230.25) .. controls (579,226.8) and (581.8,224) .. (585.25,224) .. controls (588.7,224) and (591.5,226.8) .. (591.5,230.25) .. controls (591.5,233.7) and (588.7,236.5) .. (585.25,236.5) .. controls (581.8,236.5) and (579,233.7) .. (579,230.25) -- cycle ;
\draw  [fill={rgb, 255:red, 155; green, 155; blue, 155 }  ,fill opacity=1 ] (607,230.25) .. controls (607,226.8) and (609.8,224) .. (613.25,224) .. controls (616.7,224) and (619.5,226.8) .. (619.5,230.25) .. controls (619.5,233.7) and (616.7,236.5) .. (613.25,236.5) .. controls (609.8,236.5) and (607,233.7) .. (607,230.25) -- cycle ;
\draw    (573.25,56.25) -- (612.82,201.25) ;
\draw  [fill={rgb, 255:red, 155; green, 155; blue, 155 }  ,fill opacity=1 ] (567,56.25) .. controls (567,52.8) and (569.8,50) .. (573.25,50) .. controls (576.7,50) and (579.5,52.8) .. (579.5,56.25) .. controls (579.5,59.7) and (576.7,62.5) .. (573.25,62.5) .. controls (569.8,62.5) and (567,59.7) .. (567,56.25) -- cycle ;
\draw  [fill={rgb, 255:red, 155; green, 155; blue, 155 }  ,fill opacity=1 ] (607.13,201.25) .. controls (607.13,197.8) and (609.68,195) .. (612.82,195) .. controls (615.96,195) and (618.5,197.8) .. (618.5,201.25) .. controls (618.5,204.7) and (615.96,207.5) .. (612.82,207.5) .. controls (609.68,207.5) and (607.13,204.7) .. (607.13,201.25) -- cycle ;
\draw  [dash pattern={on 4.5pt off 4.5pt}] (217,32) .. controls (237,22) and (676.5,-5) .. (656.5,15) .. controls (636.5,35) and (639.5,207) .. (659.5,237) .. controls (679.5,267) and (419.5,291) .. (399.5,261) .. controls (379.5,231) and (197,42) .. (217,32) -- cycle ;
\draw  [fill={rgb, 255:red, 255; green, 255; blue, 255 }  ,fill opacity=1 ] (223,149.5) -- (256.9,149.5) -- (256.9,143) -- (279.5,156) -- (256.9,169) -- (256.9,162.5) -- (223,162.5) -- cycle ;
\draw  [fill={rgb, 255:red, 155; green, 155; blue, 155 }  ,fill opacity=1 ] (348.25,95) .. controls (348.25,91.55) and (351.05,88.75) .. (354.5,88.75) .. controls (357.95,88.75) and (360.75,91.55) .. (360.75,95) .. controls (360.75,98.45) and (357.95,101.25) .. (354.5,101.25) .. controls (351.05,101.25) and (348.25,98.45) .. (348.25,95) -- cycle ;
\draw  [fill={rgb, 255:red, 155; green, 155; blue, 155 }  ,fill opacity=1 ] (595,37.25) .. controls (595,33.8) and (597.8,31) .. (601.25,31) .. controls (604.7,31) and (607.5,33.8) .. (607.5,37.25) .. controls (607.5,40.7) and (604.7,43.5) .. (601.25,43.5) .. controls (597.8,43.5) and (595,40.7) .. (595,37.25) -- cycle ;
\draw    (324.25,113.25) -- (381.25,173.25) ;
\draw  [fill={rgb, 255:red, 155; green, 155; blue, 155 }  ,fill opacity=1 ] (318,113.25) .. controls (318,109.8) and (320.8,107) .. (324.25,107) .. controls (327.7,107) and (330.5,109.8) .. (330.5,113.25) .. controls (330.5,116.7) and (327.7,119.5) .. (324.25,119.5) .. controls (320.8,119.5) and (318,116.7) .. (318,113.25) -- cycle ;
\draw  [fill={rgb, 255:red, 155; green, 155; blue, 155 }  ,fill opacity=1 ] (375,173.25) .. controls (375,169.8) and (377.8,167) .. (381.25,167) .. controls (384.7,167) and (387.5,169.8) .. (387.5,173.25) .. controls (387.5,176.7) and (384.7,179.5) .. (381.25,179.5) .. controls (377.8,179.5) and (375,176.7) .. (375,173.25) -- cycle ;

\draw (539,29.4) node [anchor=north west][inner sep=0.75pt]    {$B_{i}$};
\draw (113,61.4+7) node [anchor=north west][inner sep=0.75pt]    {$v_{i}$};
\draw (94,171.4) node [anchor=north west][inner sep=0.75pt]    {$F$};
\draw (475,243.4) node [anchor=north west][inner sep=0.75pt]    {$G$};

\end{tikzpicture}

\caption{Scheme of one-to-one block simulation. In this case, network $F$ is simulated by $G$. Each node in $F$ is assigned to a block in $G$ and state coding is injective. Observe that blocks are connected (one edge in the original graph may be represented by a path in the communication graph of $G$) according to connections between nodes in the original network $F$. This connections are represented by blue lines}
\label{fig:bsimulation}
\end{figure}
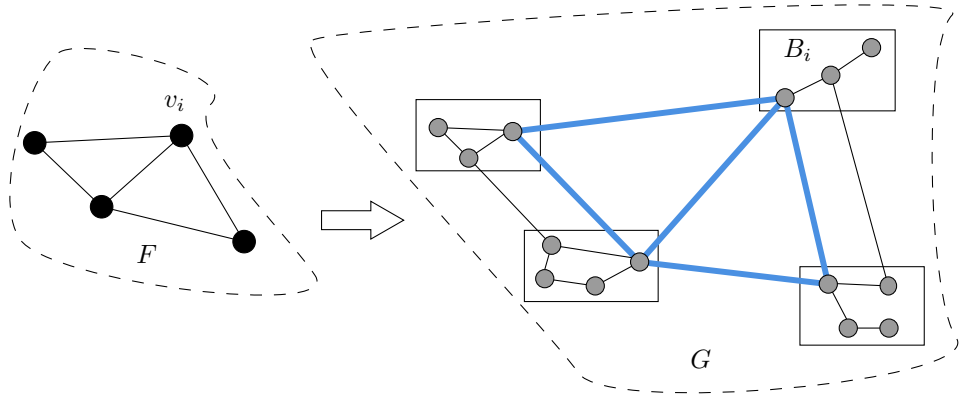

\begin{remark}\label{rem:bloc-simul}
  It is convenient in many concrete cases to define a block embedding through blocks $D_i$ that are disjoint but do not cover ${V_G}$ and add a context block $C$ disjoint from the $D_i$ that completes the covering of ${V_G}$. In this variant a block embedding of $Q_F^{V_F}$ into ${Q_G^{V_G}}$ is given by patterns $p_{i,q}$ and a constant context pattern $p_C\in Q_ G^C$ which define an injective map ${\phi:Q_F^{V_F}\rightarrow Q_G^{V_G}}$ by ${\phi(x)_{D_i} = p_{i,x_i}}$ for each $i\in V_F$ and ${\phi(x)_C = p_C}$. This variant is actually just a particular instance of Definition~\ref{def:bloc-simu} because we can include $C$ in an arbitrary block (${D_i\leftarrow D_i\cup C}$) and define the block embedding as in Definition~\ref{def:bloc-simu}.


  Another natural particular case of Definition~\ref{def:bloc-simu} corresponding to localized information is when in each block $D_i$, there is a special node $v_i\in D_i$ such that the map ${q\mapsto p_{i,q}(v_i)}$ is injective. It is only possible when $Q_G$ is larger than $Q_F$, but it will be the case in several examples of Boolean automata networks below. Interestingly, this local coding phenomena is forced when some automate network $G$ simulates some Boolean automata network $G$: indeed, in any block $D_i$ of $G$ at least one node $v_i$ must change between patterns $p_{i,0}$ and $p_{i,1}$, but the map ${x\in\{0,1\}\mapsto p_{i,x}}$ being injective, it means that ${x\mapsto p_{i,x}(v_i)}$ is injective too.
\end{remark}

\begin{remark}
  The simulation relation of Definition~\ref{def:bloc-simu} is a pre-order on automata networks.
\end{remark}

The orbit graph $G_F$ associated to network $F$ with nodes $V$ and alphabet $Q$ is the digraph with vertices ${Q^V}$ and an edge from $x$ to $F(x)$ for each $x\in Q^V$.
We also denote $G_F^t = G_{F^t}$.

\begin{lemma}\label{lem:simu-dynamic}
  If $G$ simulates $F$ via block embedding with time constant $T$ then the orbit graph $G_F$ of $F$ is a subgraph of $G_G^T$. In particular if $F$ has an orbit with transient of length $t$ and period of length $p$, then $G$ has an orbit with transient of length $Tt$ and period $Tp$.
\end{lemma}
\begin{proof}
  The embedding of $G_F$ inside $G_G^T$ is realized by definition by the block embedding of the simulation. The consequence on the length of periods and transients comes from the fact that the embedding $\phi$ verifies: $x$ is in a periodic orbit if and only if $\phi(x)$ is in a periodic orbit.   
\end{proof}

\subsection{Simulation between automata network families}

From now on, a family of automata networks will be given as a pair ${(\mathcal{F},\mathcal{F}^*)}$ where $\mathcal{F}$ is the set of abstract automata networks and $\mathcal{F}^*$ a standard representation.
We can now present our notion of simulation between families: a family $\mathcal{A}$ can simulate another family $\mathcal{B}$ if we are able to \textit{effectively} construct for any $B\in\mathcal{B}$ some automata network ${A\in\mathcal{A}}$ that is able to simulate $B$ in the sense of Definition~\ref{def:bloc-simu}.
More precisely, we ask on one hand that the automata network which performs the simulation do this task in reasonable time and reasonable space in the \emph{size of the simulated automata network}, and, on the other hand, that the construction of the simulator is efficient in the \emph{size of the representation of the simulated one}.

\begin{definition}\label{def:simu-family}
  Let $(\mathcal{F},\mathcal{F}^*)$ and $(\mathcal{H},\mathcal{H}^*)$ be two families with standard representations on alphabets $Q_F$ and $Q_H$ respectively. Let $T,S: \N \to \N$ be two functions. We say that $\mathcal{F}^*$ simulates $\mathcal{H}^*$ in time $T$ and space $S$  if there exists a \DLOG{} Turing machine $M$ such that for each $w \in L_{\mathcal{H}}$ representing some automata network ${H_w \in \mathcal{H} : Q_H^n\to Q_H^n}$, the machine produces a pair $M(w)$ which consists in:
  \begin{itemize}
  \item $w'\in L_{\mathcal{F}}$ with $F_{w'}:Q_F^{n_F}\to Q_F^{n_F}$,
  \item ${T(n)}$ and a representation of a block embedding $\phi:Q^{n_F}\to Q^{n}$, 
  \end{itemize}
  such that ${n_F=S(n)}$ and $F_{w'}$ simulates $H_w$ in time $T= T(n)$ under block embedding $\phi$. 
\end{definition}
From now on, whenever $\mathcal{F}^*$ simulates $\mathcal{H}^*$ in time $T$ and space $S$ we write $\mathcal{H}^* \preccurlyeq^T_S \mathcal{F}^*.$ 
In the above definition, the map $T(n)$ represents the temporal rescaling factor for a simulated network of size $n$.
The spatial rescaling factor is ${S(n)/n}$.
In the sequel we will mostly consider two cases: when the spatio-temporal rescaling factors are polynomial in $n$, and when they are constant.

\begin{remark}
  Note that both $T$ and $S$ maps must be \DLOG{} computable from this definition. Moreover, the simulation relation between families is transitive because the class \DLOG{} is closed under composition and simulation between individual automata networks is also transitive. When composing simulations time and space maps $S$ and $T$ get composed. 
\end{remark}

\subsection{Decision problems and automata network dynamics}
\label{sec:decis-probl-autom}
Studying the complexity of decision problems related to the dynamics of some discrete dynamical system is a very well known and interesting approach for measuring the complexity of the dynamics. In this section we introduce three variants of a classical decision problem that is closely related to the dynamical behavior of automata networks: the prediction problem.
This problem consists in predicting the state of one node of the network at a given time.
We study short term and long term versions of the problem depending on the way the time step is given in input.
In addition, we explore a variant in which we ask if some node has eventually changed without specifying any time step, but only a constant observation time rate $\tau$.
In other words, we check the system for any changement on the state of a particular node every multiple of $\tau$ time steps.
The main point of this subsection is to show that these problems are coherent with our simulation definition in the sense that if some family of automata networks $(\mathcal{F}_2,\mathcal{F}^*_2)$ simulates $(\mathcal{F}_1,\mathcal{F}^*_1)$ then, if some of the latter problem is hard for  $\mathcal{F}_1$ it will also be hard for $\mathcal{F'}_2$.
We will precise this result in the following lines.

Let $(\mathcal{F},\mathcal{F}^*)$ an automata network family and let $L \in \{0,1\}^* \times \{0,1\}^*$ be a parametrized language. We say that $L$ is parametrized by $\mathcal{F}$ if $L$  has $\mathcal{F}^*$ encoded as parameter. We note $L_{\mathcal{F}} \in \{0,1\}^*  $ as the language resulting on fixing $\mathcal{F}^*$ as a constant.

In particular, we are interested in studying prediction problems. We start by defining two variants of this well-known decision problem:
\begin{problem}[Unary Prediction (\PREDU)]\ 
	\label{prob:upred}
	\begin{description}
		\item[Parameters: ] alphabet $Q$, a standard representation $\mathcal{F}^*$ of an automata network family $\mathcal{F}$
		\item[Input: ]\ 
		\begin{enumerate}
			\item a word $w_F \in \mathcal{F}^*$ representing an automata network $F:Q^n \to Q^n$ on alphabet $Q$, with $F \in \mathcal{F}$,
			\item a node $v \in V(F) = [n]$,
			\item an initial condition $x \in Q^V$,
                        \item a state $q\in Q$,
			\item a natural number $t$ represented in unary.
		\end{enumerate}
		\item[Question: ] $F^t(x)_v = q$?
	\end{description}
\end{problem}
\begin{problem}[Binary Prediction (\PREDB)]\ 
	\label{prob:bpred}
	\begin{description}
		\item[Parameters: ] alphabet $Q$, a standard representation $\mathcal{F}^*$ of an automata network family $\mathcal{F}$
		\item[Input: ]\ 
		\begin{enumerate}
			\item a word $w_F \in \mathcal{F}^*$ representing an automata network $F:Q^n \to Q^n$ on alphabet $Q$, with $F \in \mathcal{F}$,
			\item a node $v \in V(F) = [n]$,
			\item an initial condition $x \in Q^V$,
                        \item a state $q\in Q$,
			\item a natural number $t$ represented in binary $t \in \{0,1\}^*$.
		\end{enumerate}
		\item[Question: ] $F^t(x)_v = q$?
	\end{description}
\end{problem}
Note that two problems are essentially the same, the only difference is the representation of time $t$ that we call the \textit{observation time}. We will also call node $v$ the objective node. Roughly,  as it happens with other decision problems, such as integer factorization,  the representation of observation time  will have an impact on the computation complexity of prediction problem. When the context is clear we will refer to both problems simply as $\PRED.$ In order to precise the latter observation we present now some general complexity results concerning $\PRED.$ 
\begin{proposition}
Let $\mathcal{F}$ be a concrete automata network family. The following statements hold:
\begin{enumerate}
	\item $\PREDU_{\mathcal{F}} \in \textbf{P}$
	\item $\PREDB_{\mathcal{F}} \in \textbf{PSPACE}$
\end{enumerate} 
\end{proposition}
Finally, we show that latter problem is coherent with our definition of simulation, in the sense that we can preserve the complexity of $\PRED$. Note that this give us a powerful tool in order to classify concrete automata rules according to the complexity of latter decision problem.

\begin{lemma}
Let $(\mathcal{F},\mathcal{F}^*)$ and $(\mathcal{H}, \mathcal{H}^*)$ be two automata network families. Let $T,S:\N \to \N$ be two polynomial functions such that $\mathcal{H}^* \preccurlyeq^T_S \mathcal{F}^*$ then, $\PRED_{\mathcal{H}^*} \leq^T_{\textbf{L}} \PRED_{\mathcal{F}^*}$\footnote{Here we denote $\leq^T_{\textbf{L}}$ as a $\DLOG$ Turing reduction. The capital letter ``T'' stands for \emph{Turing reduction} and it is not related to the simulation time function which is also denoted by T.} where {\PRED} denotes either \PREDU or \PREDB
\label{lemma:predsim}
\end{lemma}

\begin{proof}
	Let $(w_H,v,x,q,t)$ be an instance of $\PRED_{\mathcal{H}^*}$. By definition of simulation, there exists a $\DLOG$ algorithm which takes $w_H$ and produces a word $w_F \in L_{\mathcal{F}}$ with $F:Q^{n_F} \to Q^{n_F}$ and a block representation $\phi: Q^{n_F} \to Q^n$ such that $n_F = S(n)$ and $F$ simulates $H$ in time $T(n)$ under block embedding $\phi$. Particularly, there exists a partition of blocks $D_v \subseteq V(F) = [n_F]$ for each $v \in V(H) = [n]$ and a collection of injective patterns, i.e. patterns $p_{i,q} \in Q_F^{D_i}$ such that $p_{i,q} = p_{i,q'} \implies q=q'.$ In addition, we have $\phi \circ H = F^T \circ \phi.$ Let us define the configuration $y \in Q_F^{n_F}$ as $y_{D_i} = p_{i,x_i}$, i.e., $\phi(x)_{D_i} = y_{D_i}$. Note that $y$ is well-defined as the block map is injective. In addition, let us choose an arbitrary vertex $v' \in D_v$ and an arbitrary ${q'\in Q_F}$ and let us consider now the instance of $\PRED_{\mathcal{F}^*}$ given by $(w_F,v',y,t \times T)$. Note that for each $v' \in D_v$ the transformation $(w_H,v,x,q,t) \to (w_F,v',y,q',t \times T)$ can be done in $\DLOG(|w_H|)$ because we can read the representation of $\phi$ for each block $p_{i,x_i}$ and then output the configuration $y$. We claim that  there exists a $\DLOG(|w_H|)$ algorithm that decides if $(w_H,v,x,t) \in L_{\PRED_{\mathcal{H}^*}}$  with oracle calls to $\PRED_{\mathcal{F}^*}$. More precisely, as a consequence of the injectivity of block embedding, it is sufficient to runs oracle calls of $\PRED_{\mathcal{F}^*}$ for $(w_F,v',q',y,t \times T)$ for a set of pairs $v' \in D_v$ and ${q'\in Q_F}$ of size at most ${|Q|}$ in order to determine which pattern ${p_{v,q}}$ verifies ${F^{t\times T}(y)_{D_v}=p_{v,q}}$. Finally, all of this can be done in $\DLOG$ as $n_F = S(n) = n^{\mathcal{O}(1)}$ and $T =  n^{\mathcal{O}(1)}$ and thus, a polynomial amount of calls to each oracle is needed.
\end{proof}

Finally, we would like to study the case in which the observation time is not unique and ask whether the state of some node eventually changes.
However, in order to preserve complexity properties under simulation, we still need to have some sort of restriction on observation times.
This will allow us to avoid giving misleading answers when the simulating network is performing one step of simulation: indeed, it could take several time steps for the simulating network in order to represent one step of the dynamics of the simulated network, so some state change could happen in the intermediate steps while the simulated dynamics involve no state change.
In order to manage this sort of time dilation phenomenon between simulating and simulated systems, we introduce the following decision problem.


\begin{problem}[Prediction change $\PREDC_{\mathcal{F^*}}$]\ 
	\label{prob:cpred}
	\begin{description}
		\item[Parameters: ]  alphabet $Q$, a standard representation $\mathcal{F}^*$ of an automata network family $\mathcal{F}$
		\item[Input: ]\ 
		\begin{enumerate}
			\item a word $w_F \in \mathcal{F}^*$ representing an automata network $F:Q^n \to Q^n$ on alphabet $Q$, with $F \in \mathcal{F}$,
			\item a node $v \in V(G)$,
			\item an initial condition $x \in Q^V$,
			\item a time gap $k \in \mathbb{N}$ in unary.
		\end{enumerate}
		\item[Question: ] $\exists t \in \N: F^{kt}(x)_v \not = x_v$
	\end{description}
\end{problem}

As we did with previous versions of prediction problem, we introduce a general complexity result and then, we show computation complexity is consistent under simulation.

\begin{proposition}
	Let $(\mathcal{F}, \mathcal{F}^*)$ be a automata network family. $\PREDC_{\mathcal{F}} \in \textbf{PSPACE}.$
\end{proposition}

The injectivity of block encodings in our definition of simulation is essential for the following lemma as it guaranties that a state change in the simulating network always represent a state change in the simulated network at the corresponding time steps.

\begin{lemma}
Let $(\mathcal{F},\mathcal{F}^*)$ and $(\mathcal{H}, \mathcal{H}^*)$ be two automata network families and $T,S:\N \to \N$ two polynomial functions such that $\mathcal{H^*} \preccurlyeq^T_{S} \mathcal{F^*}$ then, $\PREDC_{\mathcal{H^*}} \leq^T_{\textbf{L}} \PREDC_{\mathcal{F^*}}$.
\label{lemma:predcsim}
\end{lemma}
\begin{proof}
	Proof is analogous to short term prediction case. Let $(w_H,v,x,k)$ be an instance of $\PREDC_{\mathcal{H}^*}$. Again, by the definition of simulation, there exists a $\DLOG$ algorithm which takes $w_H$ and produces a word $w_F \in L_{\mathcal{F}}$ with $F:Q^{n_F} \to Q^{n_F}$ and a block representation $\phi: Q^{n_F} \to Q^n$ such that $n_F = S(n)$ and $F$ simulates $H$ in time $T(n)$ under block embedding $\phi$. The latter statements means, particularly, that there exists a partition of blocks $D_v \subseteq V(F) = [n_F]$ for each $v \in V(H) = [n]$ and a collection of injective patterns, i.e. patterns $p_{i,q} \in Q_F^{D_i}$ such that $p_{i,q} = p_{i,q'} \implies q=q'$ and also that $\phi \circ H = F^{T(n)} \circ \phi.$ Let us define the configuration $y \in Q_F^{n_F}$ as $y_{D_i} = p_{i,x_i}$, i.e., $\phi(x)_{D_i} = y_{D_i}$. Note, again, that $y$ is well-defined as the block map is injective. Now we proceed in using the same approach than before: for each $v' \in D_v$ we can produce an instance $(w_F,v',y,kT(n))$ of $\PREDC_{\mathcal{F^*}}.$  There exists a $\DLOG$ machine which produces $(w_F,v',y,kT(n))$ for each $v' \in D_v$ and calls for an oracle solving $\PREDC_{\mathcal{F^*}}(w_F,v',y,kT(n))$ and outputs $1$ if there is at least one YES-instance for some $v'$. By definition of simulation and injectivity of block embedding function we have that this algorithm outputs $1$ if and only if $(w_H,v,x,k) \in \PREDC_{\mathcal{H}^*}.$
\end{proof}

To end this subsection, let us show that problems $\PREDC$ and $\PREDB$ are actually orthogonal: depending of the family of automata networks considered, one can be harder than the other and reciprocally.

\begin{theorem}\label{theo:orthogonalpredictions}
  The exists a family with circuit representation ${(\mathcal{F},\mathcal{F}^*)}$ such that ${\PREDB_{\mathcal{F}}}$ is solvable in polynomial time while ${\PREDC_{\mathcal{F}}}$ is NP-hard.
  Conversely, the exists a family with circuit representation ${(\mathcal{G},\mathcal{G}^*)}$ such that ${\PREDB_{\mathcal{G}}}$ is PSPACE-complete while ${\PREDC_{\mathcal{G}}}$ is solvable in polynomial time.
\end{theorem}
\begin{proof}
  Given a SAT formula $\phi$ with $n$ variables, let us define the automata network $F_\phi$ on ${\{0,1\}^{n+1}}$ which interprets any configuration as a pair ${(b,v)\in \{0,1\}\times\{0,1\}^n}$ where $b$ is the state of node $1$ and $v$ is both a number represented in base 2 and a valuation for $\phi$ and does the following:
  \[F(b,i) =
    \begin{cases}
      (1, i+1\bmod 2^n) &\text{ if $\phi$ is true on valuation $v$,}\\
      (0, i+1\bmod 2^n) &\text{ else.}
    \end{cases}
  \]
  A circuit representation of size polynomial in $n$ can be computed in $\DLOG$ from $\phi$ and we define ${(\mathcal{F},\mathcal{F}^*)}$ as the family obtained by considering all $F_\phi$ for all SAT formulas $\phi$.
  First, ${\PREDB_{\mathcal{F}}}$ can be solved in polynomial time: given $F_\phi$, an initial configuration ${(b,v)}$ and a time $t$, it is sufficient to compute ${v'=v+t-1\bmod 2^n}$ and verify the truth $b'$ of $\phi$ on valuation $v'$ and we have ${(b',v'+1\bmod 2^n)=F^t(b,v)}$.
  To see that ${\PREDC_{\mathcal{F}}}$ is NP-hard, it suffices to note that, on input ${(0,0\cdots0)}$, $F_\phi$ will test successively each possible valuation for $\phi$ and the state of node $1$ will change to $1$ at some time if and only if formula $\phi$ is satisfiable.

  For the second part of the proposition, the key is the construction for any $n$ of an automata network $H_n$ on $Q^n$ that completely trivializes problem ${\PREDC}$ in the following sense: for any configuration ${c\in Q^n}$ and any ${k\leq 2^n}$ and any node $v$, there is some $t$ such that ${c_v\neq H_n^{kt}(c)_v}$.
  Taking any automata network $F$ with $n$ nodes, the product automata network ${F\times H_n}$ (working on the product of alphabets in such a way that each component evolves independently) has the same property, namely that all instances of ${\PREDC}$ with ${k\leq 2^n}$ have a positive answer.
  From this, taking any family with a PSPACE-hard ${\PREDC}$ problem (they are known to exist, see Corollary~\ref{cor:universality} for details), and replacing each automata network $F$ with $n$ nodes by the product ${F\times H_n}$ (the circuit representation of the product is easily deduced from the representations of each component), we get a family ${(\mathcal{G},\mathcal{G}^*)}$ such that ${\PREDB_{\mathcal{G}}}$ is \textbf{PSPACE}-complete while ${\PREDC_{\mathcal{G}}}$ is easy: on one hand, taking products does not simplify ${\PREDB}$ problem (because deciding whether node $v$ is in state $q$ on the $F$ component reduces to deciding whether node $v$ is in state ${(q,q')}$ for some state $q'$ of the $H_n$ component); on the other hand, ${\PREDC}$ becomes trivial (always true) on inputs where the observation interval $k$ is less than ${2^n}$, and if ${k\geq 2^n}$ then the size of the whole orbit graph of the input network is polynomial in $k$ (since $k$ is given in unary), so the entire orbit of the input configuration can be computed explicitly in polynomial time and the ${\PREDC}$ can be answered in polynomial time.
  
  Let us complete the proof by giving an explicit construction of the automata networks $H_n$ over $Q^n$ with the desired property.
  ${Q=\{0,1\}\times\{0,1\}\times\{0,1\}}$ and $H_n$ interprets any configuration as a triplet of Boolean configurations ${(c,i,k)}$ with the following meaning: $k$ is a global counter that will take all possible values between $0$ and $2^n-1$ and loop, $i$ is a local counter that will run from $0$ to $2k$ and $c$ is the component where state changes will be realized at precise time steps to ensure the desired property of $H_n$. The goal is to produce in any orbit and for any $k$ and at any node the sequence of states ${O^k1^k}$ on the $c$-component: such a behavior is sufficient to ensure the desired property on $H_n$. This is obtained by defining ${H_n(c,i,k)=(c',i',k')}$ as follows:
  \begin{itemize}
  \item for any node $v$, ${c'_v=0}$ if $i<k$ or ${i\geq 2k}$, and ${c'_v=1}$ else,
  \item $i'=0$ if ${i\geq 2k}$ and $i'+1$ else,
  \item ${k'=k+1\bmod 2^n}$ if ${i\geq 2k}$ and ${k'=k}$ else.
  \end{itemize}
  It is clear that such an $H_n$ admits a polynomial circuit representation \DLOG{} computable from $n$.
\end{proof}
Finally, besides prediction problems, a classical type of problems studied in the literature are reachability problems \cite{Barrett_2006,Barrett_2003,Folschette_2015}.
Studying each variant of this type of problems goes beyond the scope of the present paper, but let us show that the most natural one is actually equivalent to $\PREDB$ for general automata networks.

\begin{problem}[Reachability (\REACH)]\ 
	\label{prob:reach}
	\begin{description}
		\item[Parameters: ] alphabet $Q$, a standard representation $\mathcal{F}^*$ of an automata network family $\mathcal{F}$
		\item[Input: ]\ 
		\begin{enumerate}
			\item a word $w_F \in \mathcal{F}^*$ representing an automata network $F:Q^n \to Q^n$ on alphabet $Q$, with $F \in \mathcal{F}$.
			\item an initial configuration $x \in Q^V$.
			\item a target configuration $y \in Q^V$.
		\end{enumerate}
		\item[Question: ] is there some $t$ such that $F^t(x) = y$?
	\end{description}
\end{problem}

\begin{proposition}\label{prop:predtoreach}
	If we consider the family of all automata networks over some alphabet $Q$ with circuit representations, then $\PREDB$ and $\REACH$ are equivalent under \DLOG{} reductions.
\end{proposition}
\begin{proof}
	Consider first an instance ${(F,v,x,q,t)}$ of $\PREDB$ and denote by $n$ the number of nodes of $F$ and $m$ the maximum between $n$ and the number of bits of the binary representation of $t$.
	It can be transformed into an instance ${(\overline{F},x',y)}$ of $\REACH$ in \DLOG{} as follows.
	$\overline{F}$ is a map ${A^m\to A^m}$ where ${A = Q\times Q\times\{0,1\}\cup\{\alpha\}}$.
	To simplify notation, we will write any configuration of $A^m$ without occurence of $\alpha$ as ${(x,y,t)}$ and see $t$ as a number written in binary.
	Then $\overline{F}$ is the following map: 
	\[\overline{F}(x,y,t) =
	\begin{cases}
		(x,F(y),t-1) &\text{ if $t>0$},\\
		\alpha^m &\text{ if $t=0$ and $x_v=q$},\\
		(x,y,0) &\text{ else},
	\end{cases}
	\]
	and $\overline(F)$ is the identity on any other configurations,
	where ${F(y)}$ denote the application of $F$ on the first $n$ nodes and the identity on the other ones.
	It is then straightforward to check that ${F^t(x)_v\neq x_v}$ if and only if $\overline{F}$ reaches configuration ${y=\alpha^m}$ starting from configuration ${x'=(\overline{x},\overline{x},t)}$ where $\overline{x}$ is equal to configuration $x$ on the $n$ first nodes and fixed to an arbitrary constant on the remaining ones.
	
	Conversely, any instance ${(F,x,y)}$ of $\REACH$ can be transformed into an instance ${(\overline{F},v,x,q,t)}$ of $\PREDB$ in \DLOG{} as follows.
	$\overline{F}$ is a map ${A^n\to A^n}$ where $n$ is the number of nodes of $F$ and ${A = Q\times Q\times Q}$.
	Let us write any configuration of ${A^{n+1}}$ as ${(x,y,t,a)}$ where $a\in A$ is the state of the ${(n+1)}$-th node, ${(x,y,t)}$ represent the content of the $n$ first nodes on each component, and $t$ is seen as a number between ${0}$ and ${|Q|^n-1}$ written in base $Q$.
	Fix arbitrarily $a_0\neq a_1\in A$.
	Then $\overline{F}$ is the following map: 
	\[\overline{F}(x,y,t,a) =
	\begin{cases}
		(F(x),y,t-1,a) &\text{ if $t>0$ and $x\neq y$},\\
		(F(x),y,t-1,a_1) &\text{ if $t>0$ and $x= y$},\\
		(y,y,0,a) &\text{ if $t=0$ and $x\neq y$},\\
		(y,y,0,a_1) &\text{ if $t=0$ and $x=y$}.
	\end{cases}
	\]
	One can check that $F$ reaches $y$ starting from $x$ if and only if ${\overline{F}^{|Q|^n}(x,y,|Q|^n-1,a_0)}$ is of the form ${(*,*,*,a_1)}$, hence a reduction of ${\REACH_{F}}$ to ${\PREDB_{\overline{F}}}$.
	Indeed, if $F$ ever reaches $y$ from $x$, it must be at step $t$ with ${0\leq t\leq |Q|^n -1}$ and the behavior of $\overline{F}$ consists exactly in testing if at any steps equality ${x=y}$ holds and then memorize this in the ${(n+1)}$-th node. Moreover, ${\overline{F}}$ always converges in at most ${|Q|^n}$ steps to a fixed point of the form ${(y,y,0,a)}$.
\end{proof}

The above equivalence works in general, but may fail for particular families as shown below.

\begin{theorem}\label{theo:orthogonalreach}
  There exists a family with circuit representation ${(\mathcal{F},\mathcal{F}^*)}$ such that ${\PREDB_{\mathcal{F}}}$ and ${\PREDC_{\mathcal{F}}}$ are NP-hard while ${\REACH_{\mathcal{F}}}$ is solvable in polynomial time.
  Conversely, there exists a family with circuit representation ${(\mathcal{G},\mathcal{G}^*)}$ such that ${\REACH_{\mathcal{G}}}$ is PSPACE-complete while ${\PREDC_{\mathcal{G}}}$ is solvable in polynomial time.
\end{theorem}
\begin{proof}
  For each $n$ and each SAT formula $\phi$ with $n$ variables let us define $F_\phi:\{0,1\}^{n+1}\to\{0,1\}^{n+1}$
  which interprets any configuration as a pair ${(b,v)\in \{0,1\}\times\{0,1\}^n}$ where $b$ is the state of node $1$ and $v$ is both a number represented in base 2 and a valuation for $\phi$ and does the following:
  \[F_\phi(b,v) = \begin{cases}
      (0,v+1\bmod 2^n)&\text{ if $b=1$ or ($b=0$ and $\neg\phi(v)$)}\\
      (1,v)&\text{ else.}
    \end{cases}
  \]
  Note that circuits computing $F_\phi$ can be constructed in \DLOG{} from $\phi$.
  We thus have a well-defined family with circuit representation.
  A configuration $(b,v)$ is reachable (from any initial configuration) if and only if either $b=0$, or if $b=1$ and $\neg\phi(v)$.
  ${\REACH_{\mathcal{F}}}$ is therefore solvable in polynomial time in $n$.
  
  However, node $1$ will change its state in the orbit of configuration ${(0,v)}$ (for some arbitrary ${v\in\{0,1\}^n}$) if and only if $\phi$ is satisfiable.
  Besides, it holds from the definition above that $F_\phi^{2^n}(0,v)=(0,v)$ if and only if there is no configuration of the form $(1,v')$ in the orbit of $(0,v)$, \textit{i.e.} if and only if $\phi$ is not satisfiable. We deduce that $\PREDB_{\mathcal{F}}$ is NP-hard.

  The construction of family ${(\mathcal{G},\mathcal{G}^*)}$ making ${\REACH_{\mathcal{G}}}$ hard and ${\PREDC_{\mathcal{G}}}$ easy is inspired from the construction of Theorem~\ref{theo:orthogonalpredictions}.
  Taking any $F_n:Q_F^n\to Q_F^n$, we will use again the automata network $H_n:Q_H^n\to Q_H^n$ from the proof of Theorem~\ref{theo:orthogonalpredictions}, and we construct ${G_n:(Q_F\times Q_H)^n\to(Q_F\times Q_H)^n}$ defined by: 
  \[G_n(c_F,c_H) =
    \begin{cases}
      (F_n(c_F),H_n(c_H))&\text{ if $c_H$ is of the form $(0,0,0)$},\\
      (c_F,H_n(c_H))&\text{ else.}
    \end{cases}
  \]
  So $G_n$ is just $H_n$ on the $Q_H$ component.
  Moreover, it can be checked that $H_n$ starting from ${c_0=(0,0,0)}$ goes back to $c_0$ after some time $T_0$.
  Thus when $G_n$ is started from some configuration ${(x,c_0)}$ it reaches ${(F_n^t(x),c_0)}$ after ${tT_0}$ time steps (for any ${t\geq 1}$) and no other configuration in this orbit is equal to $c_0$ on the second component.
  We deduce that $F_n$ reaches configuration $y$ starting from configuration $x$ if and only if ${G_n}$ reaches ${(y,c_0)}$ from configuration ${(x,c_0)}$.
  By choosing the family of networks $F_n$ with PSPACE-hard reachability, we obtain the family ${(\mathcal{G},\mathcal{G}^*)}$ with the desired properties.
\end{proof}

As a last 'orthogonality' result, let us show that a bound on a single dynamical parameter (periods or transients) is generally not sufficient to discard maximal computational complexity, while a bound on both periods and transients is.

\begin{theorem}\label{theo:orthodyncompu}
  For any family ${(\mathcal{F},\mathcal{F}^*)}$ where all periods and all transients are polynomially bounded by the number of nodes, the problems ${\REACH_{\mathcal{F}}}$, ${\PREDC_{\mathcal{F}}}$ and ${\PREDB_{\mathcal{F}}}$ are solvable in polynomial time.
  However there exists a family with circuit representation ${(\mathcal{G},\mathcal{G}^*)}$ made only of reversible automata networks (\textit{i.e.} having only periodic orbits and no transients) such that ${\PREDB_{\mathcal{G}}}$ is PSPACE-complete.
  There exists also a family with circuit representation ${(\mathcal{H},\mathcal{H}^*)}$ made of automata networks whose only periodic orbits are fixed points such that ${\PREDB_{\mathcal{H}}}$ is PSPACE-complete.
\end{theorem}
\begin{proof}
  First, if all periods and transients are polynomially bounded, then there is a PTIME algorithm that given an initial configuration computes the (polynomial) list of configurations that are in its limit period and the first time at which they are reached. From this information it is straightforward to solve problems ${\REACH_{\mathcal{F}}}$, ${\PREDC_{\mathcal{F}}}$ and ${\PREDB_{\mathcal{F}}}$ in PTIME.

  For the second assertion of the theorem, let us recall that any deterministic Turing machine working in space S and time T no more than  exponential can be simulated by a deterministic reversible Turing machine in time polynomial in T and space polynomial in S \cite{Bennett_1989,Levine_1990}.
  By deterministic reversible Turing machine, we mean a machine whose transition graph has in-degree and out-degree at most 1, \textit{i.e.} any Turing configuration (tape, head state and position) has $0$ or $1$ successor configuration (for instance $0$ in the case of an halting state) and $0$ or $1$ predecessor. Let us therefore consider a fixed deterministic reversible machine $M$ working in space ${N(n)\in poly(n)}$, that solves the truth problem for quantified boolean formulas (QBF)  in time ${k^{N(n)}}$ where $k$ is a constant and $n$ is the size of the QBF instance.
  By using an additional counter mechanism we can delay arbitrarily the time at which the machine enters an halting state after the QBF computation has been done, precisely: we can suppose without loss of generality that the machine writes somewhere the acceptance information (truth of the QBF formula) of any well formed input in time \emph{at most} ${(k/2)^{N(n)}}$ and leaves it untouched during \emph{at least} ${(k/2)^{N(n)}}$ additional steps before entering an halting state.
  This means that from any well-formed input configuration, the machine runs for at least ${(k/2)^{N(n)}}$ steps without halting (each configuration has a successor) and at time exactly ${(k/2)^{N(n)}}$, the configuration always contains the information of the acceptance of the QBF input.
  For some large enough constant alphabet $Q$ and for each $n$, we can construct an automata network ${G_n: Q^{N(n)+1}\to Q^{N(n)+1}}$ where the $N(n)$ first nodes are used to simulate $M$ on QBF instances of size $n$, and the last node holds two Boolean informations: the simulation direction (forward or backward) and an acceptation bit. $G_n$ behaves as follows on any configuration ${x\in Q^{N(n)+1}}$:
  \begin{itemize}
  \item if $x$ is not well encoded and does not represent a valid Turing configuration (tape, head position within space bounds and head state), then let it unchanged ($x$ is a fixed point);
  \item if the ${N(n)}$ first nodes (correctly) encode a configuration with no successor (resp. predecessor) and node ${N(n)+1}$ indicates the forward (resp. backward) direction, then change the direction and let the encoded configuration unchanged;
  \item finally if the (correctly) encoded configuration has a successor (resp. predecessor) and node ${N(n)+1}$ indicates the forward (resp. backward) direction, then do one step of simulation, let the direction unchanged and update the acceptance bit according to the new configuration obtained.
  \end{itemize}
  By reversibility of $M$, $G_n$ is itself reversible. Precisely, $G_n$ has only periodic orbits, which are of three kinds: 'garbage' fixed points that do not correspond to any valid configuration of M because of bad encoding, periodic orbits corresponding to periodic orbits of $M$ (without halt), and periodic orbits that correspond to a loop of back and forth simulation of a single orbit of $M$ that starts from a configuration without predecessor and ends in a configuration without successor.
  Moreover, on any well formed input configuration $x$ representing a QBF $\phi$ of size $n$, it holds that the acceptance bit of node ${N(n)+1}$ in ${F^{(k/2)^{N(n)}}(x)}$ tells whether $\phi$ is true or not.
  We deduce that ${\PREDB_{\mathcal{G}}}$ is PSPACE-complete where ${\mathcal{G}}$ is the family of networks $G_n$ with circuit representation (a circuit representation for $G_n$ is easy to compute from $n$ since $M$ is fixed).

  For the last claim of the theorem, take any family with a PSPACE-hard $\PREDB$ problem and transform it into a new family where each automata network is simulated step by step by a new one that implements an additional $k$-ary counter layer in states that is decreased at each simulation steps, and that stops the simulation and forces a fixed point when value zero is reached.
  Such networks have only fixed point by construction, and for a suitable choice of constant $k$, the PSPACE-hardness of $\PREDB$ is preserved because on well initialized configurations, the counter mechanism still allows an exponential number of simulation steps.
\end{proof}

\subsection{Universal automata network families}

Building upon our definition of simulation, we can now define a precise notion of universality. In simple words, an universal family is one that is able to simulate every other automata network under any circuit encoding. Our definition of simulation ensures that the amount of resources needed in order to simulate is controlled so that we can deduce precise complexity results.

Consider some alphabet $Q$ and some polynomial map $P:\N\to\N$.
We denote by $\mathcal{U}_{Q,P}$ the class of all possible functions $F:Q^n \to Q^n$ for any $n\in \N$ that admits a circuit representation of size at most ${P(n)}$. We also denote $\mathcal{U}_{Q,P}^*$ the language of all possible circuit representations of size bounded by $P$ of all functions from $\mathcal{U}_{Q,P}$.
First, note than any family with standard representation ${(\mathcal{F},\mathcal{F}^\ast)}$ is actually simulated by ${(\mathcal{U}_{Q,P},\mathcal{U}_{Q,P}^\ast)}$ for some $P$ because the definition of standard representation implies that there is a $\DLOG$ algorithm to produce a circuit representation of a given automata network of family $\mathcal{F}$ from its representation.
The notion of universality is about simulations in the other direction.

Finally for any ${\Delta\geq 1}$, denote by ${\mathcal{B}_{Q,\Delta}}$ the set of automata networks on alphabet $Q$ with a communication graph of degree bounded by $\Delta$ and by ${\mathcal{B}_{Q,\Delta}^*}$ their associated bounded degree representations made of a pair (graph, local maps) as discussed above.
They form a smaller set of automata networks that may be simulated more tightly, which is the idea of the notion of strong universality.


\begin{definition}
  \label{def:universal}
  A family of automata networks $(\mathcal{F},\mathcal{F}^*)$ is :
  \begin{itemize}
  \item \emph{universal} if for any alphabet $Q$ and any polynomial
    map $P$ it can simulate $(\mathcal{U}_{Q,P},\mathcal{U}_{Q,P}^*)$
    in time $T$ and space $S$ where $T$ and $S$ are polynomial
    functions;
  \item \emph{strongly universal} if for any alphabet $Q$ and any degree ${\Delta\geq 1}$ it can simulate ${(\mathcal{B}_{Q,\Delta},\mathcal{B}_{Q,\Delta}^*)}$ in time $T$ and space $S$ where $T$ is a constant and $S$ is a linear map.
\end{itemize}
\end{definition}

\begin{remark}
  The link between the size of automata networks and the size of their representation is the key in the above definitions: a universal family must simulate any individual automata network $F$ (just take $P$ large enough so that ${F\in\mathcal{U}_{Q,P}}$), however it is not required to simulate in polynomial space and time the family of all possible networks without restriction. Actually no family admitting polynomial circuit representation could simulate the family of all networks in polynomial time and space by the Shannon effect (most $n$-ary Boolean function have super-polynomial circuit complexity). In particular the family  ${\mathcal{B}_{Q,\Delta}}$ can't.

  At this point it is clear, by transitivity of simulations, that if some ${\mathcal{B}_{Q,\Delta}}$ happens to be universal then, any strongly universal family is also universal. It turns out that ${\mathcal{B}_{\{0,1\},3}}$ is universal. We will however delay the proof until section~\ref{sec:gmonuniv} below where we prove a more precise result which will prove to be very useful to get universality result in concrete families. 
  
  Finally, observe that universality allows a polynomial spatio-temporal rescaling, while strong universality allows only constant one.
  The fact that $S$ and $T$ are polynomial maps implies that they are computable in $\DLOG$ which is coherent with the reductions presented in Lemma \ref{lemma:predsim} and Lemma \ref{lemma:predcsim}.
\end{remark}


Now, we introduce an important corollary of universality regarding complexity. Roughly speaking, a universal family exhibits all the complexity in terms of dynamical behaviour and computational complexity of prediction problems. Concerning computational complexity, we state the result for universality, but there is actually no difference between strong or standard universality since the former implies the latter as we will see latter (Corollary~\ref{coro:strongunivimpliesuniv}).

\begin{corollary}\label{cor:universality}
  Let $(\mathcal{F},\mathcal{F}^*)$ be a universal automata network family, then it is computationally complex in the following sense:
  \begin{enumerate}
  \item $\PREDU_{\mathcal{F}}$ is $\textbf{P}$-hard.
  \item $\PREDB_{\mathcal{F}}$ is $\textbf{PSPACE}$-hard.
  \item $\PREDC_{\mathcal{F}}$ is $\textbf{PSPACE}$-hard.
  \item $\REACH_{\mathcal{F}}$ is $\textbf{PSPACE}$-hard.
  \end{enumerate}
\end{corollary}

\begin{proof}
  We first show that the family ${\mathcal{B}_{Q,\Delta}}$ is computationally complex for all problems except $\REACH$ and large enough $Q$ and $\Delta$, which shows the hardness results of the three first problems by definition of universality and Lemmas~\ref{lemma:predsim} and~\ref{lemma:predcsim}. First, any Turing machine working in bounded space can be directly embedded into a cellular automaton on a periodic configuration which is a particular case of automata network on a bounded degree communication graph (for the {\PREDC} variant we can always add a witness node that changes only when the Turing machine accepts for instance). This direct embedding is such that one step of the automata network correspond to one step of the Turing machine and one node of the network corresponds to one cell of the Turing tape. However, the alphabet of the automata network depends on the tape alphabet and the state set of the Turing machine. To obtain the desired result we need to fix the target alphabet, while allowing more time and/or more space. Such simulations of any Turing machine by fixed alphabet cellular automata with linear space/time distortion are known since a long time \cite{lindgren90}, but a modern formulation would be as follows: if there exists an intrinsically universal cellular automaton \cite{bulk2} with states set $Q$ and neighborhood size $\Delta$ (whatever the dimension), then ${\mathcal{B}_{Q,\Delta}}$ is computationally complex. The 2D cellular automaton of Banks \cite{banks} is intrinsically universal \cite{surveyOllinger} and has two states and $5$ neighbors, which shows that $\mathcal{B}_{Q,\Delta}$ is computationally complex when $\Delta\geq 5$ and $Q$ is not a singleton. The 1D instrinsically universal cellular automaton of Ollinger-Richard \cite{OllingerRichard4states} has $4$ states and $3$ neighbors so $\mathcal{B}_{Q,\Delta}$ is computationally complex when $\Delta\geq 3$ and $|Q|\geq 4$.
  Finally, $\REACH_{\mathcal{F}}$ is also \textbf{PSPACE}-hard because Proposition~\ref{prop:predtoreach} shows that ${\REACH_{\mathcal{F}}}$ is as hard as ${\PREDB_{\mathcal{U}_{Q,P}}}$ for any large enough $P$, and this latter problem is \textbf{PSPACE}-hard because ${\PREDB_{\mathcal{B}_{Q,\delta}}}$ is and Lemma~\ref{lemma:predsim} applies since ${\mathcal{U}_{Q,P}}$ clearly simulates ${\mathcal{B}_{Q,\Delta}}$.
\end{proof}


We now turn to the dynamical consequences of universality. By definition simulations are particular embeddings of orbit graphs into larger ones, but the parameters of the simulation can generate some distortion and the set of orbit graphs that can be embedded have succinct descriptions by circuits.
Before stating the main theorem, let us give some definitions to clarify these aspects.

\begin{definition}
  Fix a map ${\rho:\N\to\N}$, we say that the orbit graph $G_F$ of $F$ with $n$ nodes is $\rho$-succinct if $F$ can be represented by circuits of size at most $\rho(n)$. We say that the orbit graph $G_H$ of $H$ with $m$ nodes embeds $G_F$ with distortion $\delta:\N\to\N$ if ${m\leq\delta(n)}$ and there is $T\leq\delta(n)$ such that $G_F$ is a subgraph of $G_{H^T}$.
\end{definition}

\begin{remark}
  The embedding of orbit graphs with distortion obviously modify the relation between the number of nodes of the automata netwroks and the length of paths or cycles in the orbit graph. In particular, with polynomial distortion $\delta$, if $F$ has $n$ nodes and a cyclic orbit of length ${2^n}$ (hence exponential in the number of nodes) then in $H$ it gives a cyclic orbit of size ${O(\delta(n)2^n)}$ for up to $\delta(n)$ nodes, which does not guarantee an exponential length in the number of nodes in general, but just a super-polynomial one (${n\mapsto 2^{n^\alpha}}$ for some $0<\alpha< 1$).
\end{remark}

To fix ideas, we give examples of orbit graphs of bounded degree automata networks with large components corresponding to periodic orbits or transient.

\begin{proposition}\label{prop:bounded-degree-dynamics}
  There is an alphabet $Q$ such that for any ${n\geq 1}$ there is an automata network ${F_n\in\mathcal{B}_{Q,2}}$ whose orbit graph $G_{F_n}$ has the following properties:
  \begin{itemize}
  \item it contains a cycle $C$ of length at least ${2^n}$;
  \item there is a complete binary tree $T$ with $2^n$ leaves connected to some $v_1\in C$, \textit{i.e.} for all ${v\in T}$ there is a path from $v$ to $v_1$;
  \item there is a node ${v_2\in C}$ with a directed path of length $2^{n}$ pointing towards $v_2$;
  \item it possesses at least ${2^n}$ fixed points.
  \end{itemize}
\end{proposition}

\begin{proof}
  First on a component of states ${\{0,1,2\}\subseteq Q}$ the large cycle $C$ is obtained by the following 'odometer' behavior of $F_n$: if ${x_n\in\{0,1,2\}}$ then ${F_n(x)_n = x_n+1 \bmod 3}$, and if both ${x_i,x_{i+1}\in\{0,1,2\}}$ for ${1\leq i<n}$ then 
  \[F_n(x)_i =
    \begin{cases}
      0 &\text{ if }x_i=2\\
      x_i+1 &\text{ else if }x_{i+1}=2,\\
      x_i &\text{ else.}
    \end{cases}
  \]
  $C$ is realized on ${\{0,1,2\}^n}$ as follows. For $x\in\{0,1,2\}^n$ denote by $S_i$ the sequence ${(F^t(x)_i)_{t\geq 0}}$ for any ${1\leq i\leq n}$. Clearly $S_n$ is periodic of period $012$. $S_{n-1}$ is ultimately periodic of period ${200111}$ (of length $6$) and by a straighforward induction we get that $S_1$ is ultimately periodic of period ${20^{3\cdot 2^{n-2}-1}1^{3\cdot 2^{n-2}}}$ which is of length ${3\cdot 2^{n-1}}$.


  For the tree $T$, just add states ${\{a,b\}\subseteq Q}$ with the following behavior: if $x_1\in\{a,b\}$ then ${F_n(x)_1=0}$ and if ${x_{i}\in\{a,b\}}$ and ${x_{i-1}=0}$ then ${F_n(x)_i=0}$ for ${1<i\leq n}$. In any other case, we set ${F(x)_i=x_i}$ for ${x\in\{0,1,2,a,b\}^n}$ and ${1\leq i\leq n}$.

  Using similar mechanisms as above on additional states ${c_0,c_1,c_2\in Q}$, $F_n$ runs another odometer whose behavior is isomorphic to the behavior of $F_n$ on ${\{0,1,2\}^n}$ through ${i\mapsto c_i}$, but with the following exception: when ${x_1=c_2}$ we set ${F_n(x)_1=0}$ and then state $0$ propagates from node $1$ to node $n$ as in the construction of tree $T$. We thus get a transient behavior of length more than ${3\cdot 2^{n-1}}$ which yields to configuration ${0^n}$, which itself (belongs or) yields to cycle $C$.

  Finally, the fixed points are obtained by adding two more states to the alphabet on which the automata network just acts like the identity map.
\end{proof}

We can now state that any universal family must be dynamically rich in a precise sense.

\begin{theorem}\label{them:univ-rich-dynamics}
  Let $\mathcal{F}$ be an automata networks family.
  \begin{itemize}
  \item If $\mathcal{F}$ is \emph{universal} then, for any polynomial map
    $\rho$, there is a polynomial distortion $\delta$ such that, any
    $\rho$-succinct orbit graph can be embedded into some
    ${F\in\mathcal{F}}$ with distortion $\delta$. In particular
    $\mathcal{F}$ contains networks with super-polynomial periods and
    transients, and a super-polynomial number of disjoint periodic orbits of period at most polynomial.
  \item If $\mathcal{F}$ is \emph{strongly universal} then it embeds the orbit graph of any bounded-degree automata network with linear distortion. In particular it contains networks with exponential periods and transients, and an exponential number of disjoint periodic orbits of period at most linear.
\end{itemize}
\end{theorem}
\begin{proof}
  This is a direct consequence of Lemma~\ref{lem:simu-dynamic}, Definition~\ref{def:simu-family} and Proposition~\ref{prop:bounded-degree-dynamics} above.
\end{proof}



Of course, we do not claim that computational complexity and dynamical richness as stated above are the only meaningful consequences of universality.
To conclude this subsection about universality, let us show that it allows to prove finer results linking the global dynamics with the interaction graph.

In a directed graph, we say a node $v$ belongs to a strongly connected component if there is a directed path from $v$ to $v$.

\begin{corollary}
  Any universal family $\mathcal{F}$ satisfies the following: there is a constant $\alpha$ with ${0<\alpha\leq 1}$ such that for any $m>0$ there is a network $F\in\mathcal{F}$ with $n\geq m$ nodes such that some node $v$ belonging to a strongly connected component of the interaction graph of $F$ and a periodic configuration $x$ such that the trace at $v$ of the orbit of $x$ is of period at least ${2^{n^\alpha}}$.
  \label{coro:trickyuniversal}
\end{corollary}
\begin{proof}
  Consider a Boolean network $F$ with nodes ${V=\{1,\ldots,m\}}$ that do the following on configuration $x\in \{0,1\}^V$: it interprets ${x_1,\ldots,x_m}$ as an number $k$ written in base $2$ where $x_1$ is the most significant bit and produces ${F(x)}$ which represents number ${k+1 \bmod 2^{m}}$. 

  $F$ is such that node $1$ has a trace of exponential period and belongs to a strongly connected component of the interaction graph of $F$ (because it depends on itself).
  Note that $F$ has a circuit representation which is polynomial in $m$, and take $F'\in\mathcal{F}$ of size polynomial in $m$ that simulates $F$ in polynomial time (by universality of family $\mathcal{F}$).
  Taking the notations of Definition~\ref{def:bloc-simu}, we have that each node $v\in D_i$ for each block $D_i$ is such that the map ${q\in\{0,1\}\mapsto p_{i,q}(v)}$ is either constant or bijective (because $F$ has a Boolean alphabet, see Remark~\ref{rem:bloc-simul}).
  In the last case, the value of the node $v\in D_i$ completely codes the value of the corresponding node $i$ in $F$.
  Take any $v\in B_1$ that has this coding property.
  Since node $1$ depends on itself in $F$, there must be a path from $v$ to some node $v'\in D_1$ that is also coding in the interaction graph of $F'$. 
  Then we can also find a path from $v'$ to some coding node in $D_1$.
  Iterating this reasoning we must find a cycle, and in particular we have a coding node in $D_1$ which belongs to some strongly connected component of the interaction graph of $F'$.
  Since this node is coding the values taken by node $1$ of $F$ and since the simulation is in polynomial time and space, we deduce the super-polynomial lower bound on the period of its trace for a well-chosen periodic configuration.   
\end{proof}


\subsection{Link with cellular automata and intrinsic universality}

A cellular automaton is essentially an infinite automata network which is uniform both in the communication graph and the local rule of nodes.
Let $d\geq 1$ be an integer and ${N\subseteq \Z^d}$ be a finite set.
A cellular automaton of dimension $d$ and neighborhood $N$ and state set $Q$ is defined by a local rule ${\delta : Q^N\to Q}$ which induces a global map ${F_\delta : Q^{\Z^d}\to Q^{\Z^d}}$ defined as follow: 
\[F(c)_z = \delta\bigl(z'\in N\mapsto c_{z+z'}\bigr).\]

One can naturally associate to such a cellular automaton a familly of automata networks defined over regular graphs that are $d$-dimensional tori with uniform adjacency relation defined by $N$, and such that each node has the same local rule $\delta$.
Formally, for any ${n\in\N}$, let ${G_{d,N,n}}$ be the graph of vertex set ${V_{d,n}=\Z_n^d}$ where $\Z_n$ denotes the integers modulo $n$, and such that ${(i,j)\in V_{d,n}^2}$ is an edge if and only if ${j-i\in N}$ (where the computation is done modulo $n$).
For large enough $n$, ${G_{d,N,n}}$ is a regular graph of degree ${|N|}$.
Then, on each ${G_{d,N,n}}$, we consider the automata network where the local rule of each node $i$ is $\delta$ where we identify the neighborhood of $i$ to $N$ by ${j\mapsto j-i}$ (which is one-to-one for large enough $n$, the choice for a constant number of small values of $n$ does not matter).
We denote this family ${\mathcal{F}_{d,N,Q,\delta}}$ and consider it together with its bounded-degree representation.

A well-established notion of universality in cellular automata, \emph{intrinsic universality}, is actually very close to our formalism and relies on the notion of \emph{intrinsic simulations}.
The goal of this subsection is to clarify the links between intrinsic universality of a cellular automaton, and universality of the associated automata networks family.

To be precise we consider the notion of intrinsic universality \cite[Definition 5.1]{bulk2} associated to injective simulation \cite[Definition 2.1]{bulk2}.
The definition of injective simulation between $d$-dimensional cellular automata is essentially equivalent to Definition~\ref{def:bloc-simu} with the additional constraint that the block embedding uses the same 'rectangular' shape for all blocks.
We say that a block embedding between automata networks over graphs ${G_{d,N_1,n_2}}$ and ${G_{d,N_2,n_2}}$ is a \emph{uniform rectangular block embedding} of shape ${\vec{b}=(b_1,\ldots,b_d)}$ if it is such that each block $D_i$ is of the form ${v_i+[0,\ldots,b_1]\times\cdots\times[0,\ldots,b_d]}$ for some $v_i$.
Injective intrinsic universality of cellular automata can then be defined on the associated automata network families as follows.

\begin{definition}
  A cellular automaton of dimension $d$, neighborhood $N$, state set $Q$ and local rule $\delta$ is \emph{intrinsically universal} if for any neighborhood $N'$, state set $Q'$ and local rule $\delta'$, the family $\mathcal{F}_{d,N,Q,\delta}$ simulates $\mathcal{F}_{d,N',Q',\delta'}$ in constant time and with a uniform rectangular block embedding of fixed shape $\vec{b}$.
\end{definition}

As a first result, let us show that intrinsic universality implies a very general capacity of simulating automata networks, close to strong universality but slightly weaker. 

\begin{theorem}\label{theo:iuca}
  Consider any $d$-dimensional intrinsically universal CA and denote by $(\mathcal{F},\mathcal{F}^\ast)$ its associated automata network family with bounded degree representation.
  For any $Q$ and $\Delta$, $\mathcal{F}$ can simulate $\mathcal{B}_{Q,\Delta}$ in time ${O(n\log(n))}$ and space ${O(n\log(n))}$.
\end{theorem}
\begin{proof}
  We do the proof for dimension $1$, it is straightforward to lift the construction to higher dimension by periodization in all but one dimensions.
  It is sufficient to prove that some family ${\mathcal{F}_{1,\{-1,0,1\},Q',\delta}}$ associated to a cellular automaton can simulate $\mathcal{B}_{Q,\Delta}$ in time ${O(n\log(n))}$ and space ${O(n\log(n))}$ for each fixed choice of $Q$ and $\Delta$.
  When $Q$ and $\Delta$ are fixed, there are only finitely many possible local transition rules so, by choosing $Q'$ large enough, any such rule can be encoded locally as well as the set of possible values ${Q^\Delta}$ of the neighbors of a given node.
  The only difficulty of the simulation lies in the routing of states of nodes to their corresponding neighbors according to an arbitrary graph of degree $\Delta$. The main trick is that the communication of the state of each node to each of its neighbors is done by a routing mechanism of packets turning on a ring: each holds a number coding the travel distance it has to accomplish before delivering its information (a state), and, while turning on the ring, each packet decrements its number until it is 0 and then triggers the information delivery. 
  \begin{figure}

\tikzset{every picture/.style={line width=0.75pt}} 

\begin{tikzpicture}[x=0.75pt,y=0.75pt,yscale=-1,xscale=1]

\draw  [fill={rgb, 255:red, 255; green, 255; blue, 255 }  ,fill opacity=1 ] (89.42,146.46) -- (512.94,146.46) -- (512.94,178.05) -- (89.42,178.05) -- cycle ;
\draw    (118.79,146.46) -- (118.79,272.82) ;
\draw    (151.71,146.46) -- (151.71,272.82) ;
\draw    (177.51,147.48) -- (177.51,273.84) ;
\draw    (209.54,147.48) -- (209.54,273.84) ;
\draw    (236.23,146.46) -- (236.23,272.82) ;
\draw    (269.15,146.46) -- (269.15,272.82) ;
\draw    (294.95,147.48) -- (294.95,273.84) ;
\draw    (326.98,147.48) -- (326.98,273.84) ;
\draw    (480.02,147.48) -- (480.02,273.84) ;
\draw   (509.38,102.71) .. controls (509.38,98.04) and (507.05,95.71) .. (502.38,95.71) -- (308.07,95.71) .. controls (301.4,95.71) and (298.07,93.38) .. (298.07,88.71) .. controls (298.07,93.38) and (294.74,95.71) .. (288.07,95.71)(291.07,95.71) -- (93.76,95.71) .. controls (89.09,95.71) and (86.76,98.04) .. (86.76,102.71) ;
\draw  [fill={rgb, 255:red, 255; green, 255; blue, 255 }  ,fill opacity=1 ] (89.42,209.23) -- (512.94,209.23) -- (512.94,240.82) -- (89.42,240.82) -- cycle ;
\draw  [fill={rgb, 255:red, 255; green, 255; blue, 255 }  ,fill opacity=1 ] (89.42,177.64) -- (512.94,177.64) -- (512.94,209.23) -- (89.42,209.23) -- cycle ;
\draw  [fill={rgb, 255:red, 255; green, 255; blue, 255 }  ,fill opacity=1 ] (89.42,240.82) -- (512.94,240.82) -- (512.94,272.41) -- (89.42,272.41) -- cycle ;
\draw  [fill={rgb, 255:red, 255; green, 255; blue, 255 }  ,fill opacity=1 ] (89.42,272.41) -- (512.94,272.41) -- (512.94,304) -- (89.42,304) -- cycle ;
\draw    (354.79,290.75) -- (406.07,290.75) ;
\draw [shift={(351.79,290.75)}, rotate = 0] [fill={rgb, 255:red, 0; green, 0; blue, 0 }  ][line width=0.08]  [draw opacity=0] (8.93,-4.29) -- (0,0) -- (8.93,4.29) -- cycle    ;
\draw   (118.68,143.42) .. controls (118.73,139.21) and (116.66,137.07) .. (112.45,137.02) -- (112.45,137.02) .. controls (106.43,136.93) and (103.45,134.79) .. (103.51,130.58) .. controls (103.45,134.79) and (100.41,136.85) .. (94.4,136.77)(97.11,136.81) -- (94.4,136.77) .. controls (90.19,136.71) and (88.06,138.79) .. (88,143) ;
\draw   (82.31,276.49) .. controls (78.95,276.36) and (77.21,277.98) .. (77.09,281.34) -- (77.09,281.34) .. controls (76.92,286.13) and (75.15,288.47) .. (71.79,288.35) .. controls (75.15,288.47) and (76.74,290.93) .. (76.57,295.72)(76.64,293.57) -- (76.57,295.72) .. controls (76.44,299.08) and (78.06,300.82) .. (81.42,300.94) ;
\draw   (84.09,213.3) .. controls (80.73,213.18) and (78.99,214.8) .. (78.87,218.16) -- (78.87,218.16) .. controls (78.7,222.95) and (76.93,225.29) .. (73.57,225.17) .. controls (76.93,225.29) and (78.52,227.75) .. (78.35,232.54)(78.42,230.39) -- (78.35,232.54) .. controls (78.22,235.9) and (79.84,237.64) .. (83.2,237.76) ;
\draw   (526.28,300.94) .. controls (530.95,300.91) and (533.27,298.57) .. (533.24,293.9) -- (532.9,234.47) .. controls (532.86,227.8) and (535.17,224.46) .. (539.84,224.43) .. controls (535.17,224.46) and (532.82,221.14) .. (532.78,214.47)(532.8,217.47) -- (532.44,155.04) .. controls (532.41,150.37) and (530.07,148.05) .. (525.4,148.08) ;

\draw (270.18,67.62) node [anchor=north west][inner sep=0.75pt]    {$n\ \text{blocks }$};
\draw (386.78,154.14) node [anchor=north west][inner sep=0.75pt]    {$\cdots $};
\draw (232.28,215.91) node [anchor=north west][inner sep=0.75pt]    {$\text{Adyacency layers}$};
\draw (410.67,280.11) node [anchor=north west][inner sep=0.75pt]    {$\text{Routing layers}$};
\draw (231.13,183.3) node [anchor=north west][inner sep=0.75pt]    {$\text{Firing squad component}$};
\draw (231.36,248.52) node [anchor=north west][inner sep=0.75pt]    {$\text{Transition rule component}$};
\draw (70.27,109.34) node [anchor=north west][inner sep=0.75pt]    {$\lceil \log n\rceil \text{\mbox{-}cells}$};
\draw (7.32,279.1) node [anchor=north west][inner sep=0.75pt]    {$\Delta \text{\mbox{-} layers}$};
\draw (9.1,215.91) node [anchor=north west][inner sep=0.75pt]    {$\Delta \text{\mbox{-} layers}$};
\draw (543.78,214.88) node [anchor=north west][inner sep=0.75pt]    {$2\Delta +3\text{\mbox{-} layers}$};
\draw (231.13,119.3) node [anchor=north west][inner sep=0.75pt]    {$\text{Type component}$};

\end{tikzpicture}
\caption{Representation of the coding of an arbitrary graph by a one dimensional cellular automata as it is described in Theorem \ref{theo:iuca} }
  \end{figure}
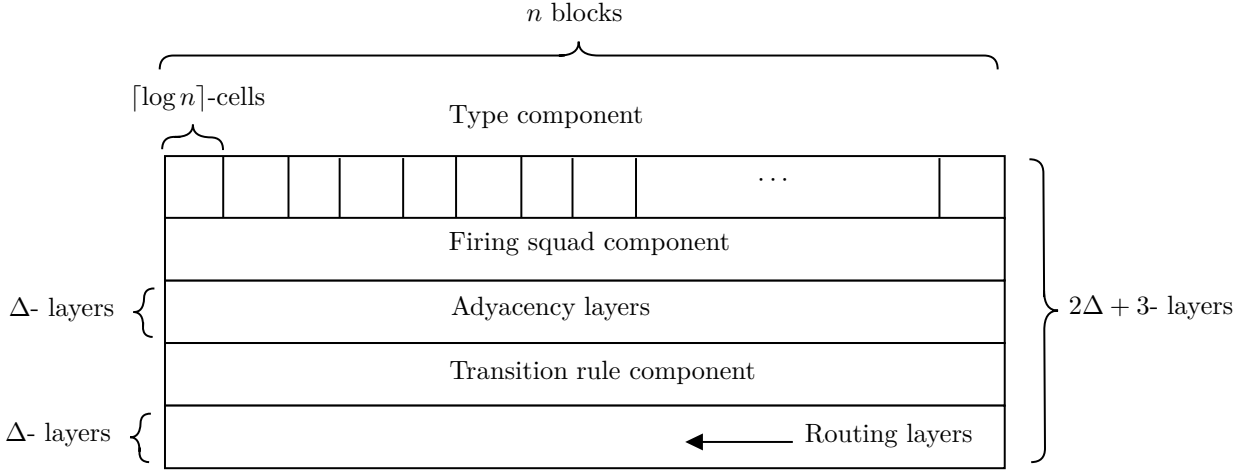
  Let us describe a cellular automaton of dimension $1$ of local rule $\delta$ that can achieve this routing task and the complete simulation. 
  Let us fix any network from ${\mathcal{B}_{Q,\Delta}}$ with $n$ nodes and communication graph $G$.
  It is simulated by the cellular automaton on the graph ${G_{1,\{-1,0,1\},n\log(n)}}$ as follows: 
  \begin{itemize}
  \item each node of $G$ is simulated by a block of ${\lceil\log(n)\rceil}$ adjacent cells of the 1D cellular automaton, the $i$th blocks corresponds to the $i$th node;
  \item the alphabet $Q'$ is structured in ${2\Delta+3}$ components:
    \begin{itemize}
    \item a \emph{type} component that serves as a marker on some cells to give them a particular behavior, or mark the limits of each block, called block skeleton; this component is invariant;
    \item a \emph{firing squad} component \cite{Umeo_2009} that serves as a
      global clock for the simulation that ``ticks'' every
      ${O(n\log(n))}$ steps by having a special ``fire'' flag appear at every cells exactly at these time steps (it is well-known that such a periodic behavior can be realized by a firing squad with a linear relation between the number of cells and the time period \cite{Umeo_2009});
    \item a \emph{transition rule} component that can hold the complete information about the transition rule of a simulated node, its state and the states of its neighbors; more precisely, this layer is empty, except for the rightmost position of each block of the skeleton where it is initialized with empty information about neighbors, just the state and the local rule of the simulated node corresponding to the block;
    \item $\Delta$ \emph{graph adjacency} components, that remain constant and that describe the adjacency relation of graph $G$: for each block $i$, each such component gives one neighbor of $i$ by coding, as a number written in binary from (less significant bit at the leftmost position), the distance in blocks starting from block $i$ to reach block $j$;
    \item $\Delta$ \emph{routing} components that are each organized as blocks of ${\lceil\log(n)\rceil}$ states, called \emph{packet}, that each hold a number (the address), that is initially aligned with the block skeleton, and initialized by the content of the graph adjacency components; moreover, the rightmost cell of the packet holds a state from $Q$ plus an index between $1$ and $\Delta$ (the data); the role of each packet is to send the state of a node to one of its neighbors; 
    \end{itemize}
  \item all packets shift ``to the left'' (from position $p$ to $p-1$ in the graph ${G_{1,\{-1,0,1\},n\log(n)}}$) synchronously cell by cell and each packet perform the following task as it travels: when the leftmost cell of a packet encounters the right boundary of block of the skeleton, it starts to decrement the number held in the packet by propagating a ${-1}$ carry that travels from left to right; if the carry reaches the rightmost position of the packet (meaning the number was $0$ and cannot be decremented), then it means that the packet has reached the block of the skeleton corresponding to its initial address: in this case the rightmost position of the packet transmits to the transition rule component of the rightmost cell of the block of the skeleton all its data exactly when it receives the carry and is aligned with the right boundary of the block; after this, the packet still travels by shifts but is deactivated so that all packets are deactivated and have correctly transmitted their data after ${n\log(n)}$ steps.
  \item the synchronous global clock achieved by the firing squad layer triggers two operations on other layers when it ticks: it forces the transition rule component to actually compute one transition and clean the local information about states of neighbors ; and it reset the routing component according to the graph adjacency component (for the address in packets) and the transition rule component (for the data of the packets). This cleaning process is done in one step since each node sees locally the special state of the firing squad and erases the layers to be cleaned locally in one step.
  \end{itemize}
  The block embedding of the simulation is very regular: each node $i$ of $G$ is represented by the $i$th block in the skeleton, with the firing squad layer initialized in the state obtained just after the ``fire'' step, the transition rule component holding only local states and transition rule of each node and no information about the neighbors, and the routing components initialized according to the graph adjacency components. By the description of the behavior above, in ${n\log(n)}$ steps, all packets have visited the entire block skeleton so they have all copied their information to the corresponding slots of the transition table component and are deactivated. After ${O(n\log(n))}$ steps the firing squads component fires, the correct transitions are computed everywhere and we are back to a well-formed block encoding representing the new configuration of nodes of graph $G$.
\end{proof}

As we will see later (section~\ref{sec:gmonuniv}), the simulation result of Theorem~\ref{theo:iuca} is enough to give universality.

\newcommand\Gmon{\GG_{m}}

\begin{corollary}
  The automata network family (with bounded degree representation) associated to any $d$-dimensional intrinsically universal CA is universal.
\end{corollary}
\begin{proof}
  Such a family simulates the family $\Gmon$ (a subfamily of $\mathcal{B}_{\{0,1\},4}$ defined in section~\ref{sec:gmonuniv}) in time ${O(n\log(n))}$ and space ${O(n\log(n))}$ by Theorem~\ref{theo:iuca}. Then, the corollary follows from Theorem~\ref{teo:gmonuniv}.
\end{proof}



Let us now prove that a family of automata networks coming from a $d$-dimensional cellular automaton cannot be strongly universal, whatever the cellular automaton considered.
The reason is that the rigidity of the network of the cellular automaton limits propagation of information and prevents a simulation with constant spatio-temporal rescaling factors of networks on arbitrary graphs.
To show this, we will use the notion of growth of balls in graph families.
Given a family of finite graphs ${(G_i)_{i\in I}}$, we say that it has a \emph{polynomial ball growth} if there is some exponent ${k\in\N}$ such that for any fixed $m$ there is a constant $C_m$ such that for any ${i\in I}$ and set $X$ of $m$ vertices of $G_i$:
\[|B(X,n)|\leq C_mn^k\]
where $B(X,n)$ denotes the set of nodes at distance $n$ from $X$ in $G_i$.
The family of d-dimensional grids has polynomial ball growth, while the family of binary trees hasn't.

Let $\mathcal{F}$ be a network family.
We say that $\mathcal{F}$ has polynomial ball growth if the family of its underlying communication graphs has polynomial ball growth.

\begin{theorem}\label{theo:ballgrowthnouniv}
  No family of polynomial ball growth can be strongly universal.
\end{theorem}
\begin{proof}
  Consider the family $\mathcal{F}_0$ of disjunctive networks on complete binary trees with self-loops, and let $\mathcal{F}$ be any family with polynomial ball growth (details about representations don't matter for the argument).
  Suppose for the sake of contradiction that $\mathcal{F}$ simulates $\mathcal{F}_0$ with constant spatio-temporal rescaling: $t$ steps of an automata network $F_0$ of size $n$ from $\mathcal{F}_0$ are simulated in time at most $\alpha t$ by an automata network $F$ of size at most $\alpha n$ from $\mathcal{F}$.
  The configuration everywhere $0$ is a fixed point of $F_0$ so its block encoding must be a fixed point of $F$.
  Let's call it $c_0$.
  In any case, there are some nodes of $F_0$ which are simulated by a block of nodes of $F$ of size at most $\alpha$.
  Choose any such node $v$ and consider the configuration of $F_0$ which is everywhere $0$ except on this node $v$ where it is $1$.
  Since the graph of $F_0$ is a complete binary tree, in $\log(n)$ steps the considered configuration becomes the configuration everywhere $1$ under the action of $F_0$.
  On the other hand, these $\log(n)$ steps must be simulated by $F$ in at most ${\alpha\log(n)}$ steps, starting from a configuration $c_1$ which differs from $c_0$ only on the block $D_v$ of size at most $\alpha$.
  In the orbit of $c_1$, after $t$ steps, only the nodes of $F$ that are at distance at most $t$ from $D_v$ can be in a different state than in $c_0$.
  By the polynomial ball growth of exponent $k$ with ${m=\alpha}$ there are at most ${C_m\log^k(n)}$ nodes that differ between ${c_0}$ and ${F^{\log(n)}(c_1)}$.
  For $n$ large enough, this is not sufficient to have $n$ blocks to change the state they represent as they should to correctly simulate the $n$ nodes of $F_0$.
\end{proof}

\begin{corollary}\label{cor:intrinsicnotstrongly}
  No automata network family coming from a $d$-dimensional cellular automaton can be strongly universal.
\end{corollary}
\begin{proof}
  It follows directly from Theorem~\ref{theo:ballgrowthnouniv} because for any $d$ and $N$, the family of graphs ${(G_{d,N,n})_{n\in\N}}$ is of polynomial ball growth.
\end{proof}

\section{Gadgets and glueing}
\label{sec:gadgetsandglueing}

In the same way as Boolean circuits are defined from Boolean gates, many automata network families can be defined by fixing a finite set of local maps $\GG$ that we can freely connect together to form a global network, called a $\GG$-network.

Such families can be strongly universal as we will see, even for very simple choices of $\GG$, which is an obvious motivation to consider them.
In this section, we introduce a general framework to prove simulation results of a $\GG$-network family by some arbitrary family that amounts to a finite set of conditions to check.
From this we will derive a framework to certify strong universality of an arbitrary family just by exhibiting a finite set of networks from the family that verify a finite set of conditions.
As already said above, our goal is to analyze automata networks with symmetric communication graph (CSAN families). Our framework is targeted towards such families.

The idea behind is that of building large automata networks from small automata networks in order to mimic the way a $\GG$-network is built from local maps in $\GG$.
The difficulty, and the main contribution of this section, is to formalize how small building blocks are glued together and what conditions on them guaranty that the large network correctly simulates the corresponding $\GG$-networks.
In particular, our formalism is perfectly suited to show that a family of \emph{undirected} networks can simulate a family of \emph{oriented} $\GG$-networks.

We will now introduce all the concepts used in this framework progressively.

\subsection{$\GG$-networks}

Let $Q$ be a fixed alphabet and $\GG$ be any set of maps of type ${g:Q^{i(g)}\rightarrow Q^{o(g)}}$ for some ${i(g),o(g)\in\N}$. We say $g$ is \emph{reducible} if it can be written as a disjoint union of two gates, and \emph{irreducible} otherwise. Said differently, if $G$ is the (bipartite) dependency graph of $g$ describing on which inputs effectively depends each output, then $g$ is irreducible if $G$ is weakly connected.

From $\GG$ we can define a natural family of networks:  a $\GG$-network is an automata network obtained by wiring outputs to inputs of a number of gates from $\GG$. To simplify some later results, we add the technical condition that no output of a gate can be wired to one of its inputs (no self-loop condition).

\begin{definition}\label{def:g-network}
  A $\GG$-network is an automata network ${F:Q^V\rightarrow Q^V}$ with set of nodes $V$ associated to a collection of gates ${g_1,\ldots, g_n\in\GG}$ with the following properties. Let 
  \begin{align*}
    I &=\{(j,k) : 1\leq j\leq n\text{ and }1\leq k\leq i(g_j)\} \text{ and}\\
    O &=\{(j,k) : 1\leq j\leq n\text{ and }1\leq k\leq o(g_j)\}
    \end{align*}
    be respectively the sets of inputs and outputs of the collection of gates ${(g_j)_{1\leq j\leq n}}$.
    We require ${|V|=|I|=|O|}$ and the existence of two bijective maps ${\alpha:I\rightarrow V}$ and ${\beta: V\rightarrow O}$ with the condition that there is no ${(j,k)\in I}$ such that $\beta(\alpha(j,k))=(j,k')$ for some $k'$ (no self-loop condition). For ${v\in V}$ with $\beta(v) = (j,k)$, let ${I_v = \{\alpha(j,1),\ldots,\alpha(j,i(g_j))\}}$ and denote by $g_v$ the map: ${x\in Q^{I_v}\mapsto g_j(\tilde x)_k}$ where ${\tilde x\in Q^{i(g_j)}}$ is defined by ${\tilde{x}_k = x_{\alpha(j,k)}}$.
    Then $F$ is defined as follows: 
  \[F(x)_v = g_v(x_{I_v}).\]
\end{definition}

\begin{figure}
  \centering
  \begin{minipage}{.3\linewidth}
    \begin{tikzpicture}[scale=.5]
      \draw (-1.5,1)--(1.5,1)--(1.5,-1)--(-1.5,-1)--cycle;
      \draw (0,0) node {$g_1$};
      \draw[->,red,very thick] (0,1)-- (0,2) node[above] {$(1,1)$};
      \draw[->,blue,very thick] (-1,-2) node[below] {$(1,1)$} -- (-1,-1);
      \draw[->,blue,very thick] (1,-2) node[below] {$(1,2)$} -- (1,-1);
      \draw (0,1.5) node[right] {$\phi_A$};
      \begin{scope}[shift={(4,0)}]
        \draw (-1.5,1)--(1.5,1)--(1.5,-1)--(-1.5,-1)--cycle;
        \draw (0,0) node {$g_2$};
        \draw[->,blue,very thick] (0,-2)  node[below] {$(2,1)$} -- (0,-1);
        \draw[->,red,very thick] (-1,1)  -- (-1,2) node[above] {$(2,1)$};
        \draw[->,red,very thick] (1,1)  -- (1,2) node[above] {$(2,2)$};
        \draw (-1,1.5) node[right] {$\phi_B$};
        \draw (1,1.5) node[right] {$\phi_C$};
      \end{scope}
    \end{tikzpicture}
  \end{minipage}
  \begin{minipage}{.25\linewidth}
    \begin{tikzpicture}[scale=.5]
      \draw (-1.5,1)--(1.5,1)--(1.5,-1)--(-1.5,-1)--cycle;
      \draw (0,0) node {$g_1$};
      \begin{scope}[shift={(0,4)}]
        \draw (-1.5,1)--(1.5,1)--(1.5,-1)--(-1.5,-1)--cycle;
        \draw (0,0) node {$g_2$};
      \end{scope}
      \draw[->,very thick] (0,1)-- node[midway,right] {$A$} (0,3);
      \draw[->,very thick] (-1,5) .. controls (-2.5,9) and (-2.5,-6) .. node[midway,left] {$B$} (-1,-1);
      \draw[->,very thick] (1,5) .. controls (2.5,9) and (2.5,-6) .. node[midway,right] {$C$} (1,-1);
    \end{tikzpicture}
  \end{minipage}
  \begin{minipage}{.3\linewidth}
    \begin{tikzpicture}[shorten >=1pt,node distance=2cm,on grid,auto]
      \tikzstyle{every state}=[fill={rgb:black,1;white,10}]
      
      \node[state]   (q_1)                    {$A$};
      \node[state] (q_2)  [left of=q_1]    {$B$};
      \node[state]           (q_3)  [right of=q_1]    {$C$};
      
      \path[<->]
      (q_1) edge (q_2);
      \path[<->]
      (q_1) edge (q_3);
      \draw (q_1)+(0,-.5) node[below] {\tiny\color{blue}(2,1) \color{red}(1,1)};
      \draw (q_2)+(0,-.5) node[below] {\tiny\color{blue}(1,1) \color{red}(2,1)};
      \draw (q_3)+(0,-.5) node[below] {\tiny\color{blue}(1,2) \color{red}(2,2)};
    \end{tikzpicture}\\
    \[
      \begin{cases}
        \phi(x)_A &= \phi_A(x_B,x_C)\\
        \phi(x)_B &= \phi_B(x_A)\\
        \phi(x)_C &= \phi_C(x_A)
      \end{cases}
    \]
  \end{minipage}
  \caption{On the left a set of maps $\GG$ over alphabet $Q$, in the middle an intuitive representation of input/output connections to make a $\GG$-network, on the right the corresponding formal $\GG$-network ${\phi : Q^3\rightarrow Q^3}$ together with the global map associated to it. The bijections $\alpha$ and $\beta$ from Definition~\ref{def:g-network} are represented in blue and red (respectively).}
  \label{fig:gnetwork}
\end{figure}
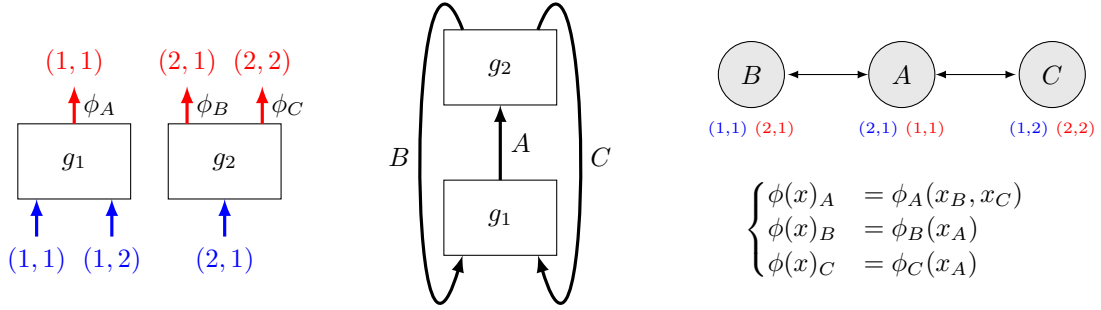

\begin{remark}\label{rem:represent-g-networks}
  Once $\GG$ is fixed, there is a bound on the degree of dependency graphs of all $\GG$-networks. Thus, it is convenient to represent $\GG$-networks by the standard representation of bounded degree automata networks (as a pair of a graph and a list of local update maps). Another representation choice following strictly Definition~\ref{def:g-network} consists in giving a list of gates ${g_1,\ldots, g_k\in\GG}$, fixing ${V=\{1,\ldots,n\}}$ and give the two bijective maps ${\alpha:I\to V}$ and ${\beta: V\to O}$ describing the connections between gates (maps are given as a simple list of pairs source/image). One can check that these two representations are \DLOG{} equivalent when the gates of $\GG$ are irreducible: we can construct the interaction graph and the local maps from the list of gates and maps $\alpha$ and $\beta$ in \DLOG{} (the incoming neighborhood of a node $v$, $I_v$, and its local map $g_v$ are easy to compute as detailed in Definition~\ref{def:g-network}); reciprocally, given the interaction graph $G$ and the list of local maps $(g_v)$, one can recover in \DLOG{} the list of gates and their connections as follows:
  \begin{itemize}
  \item for $v$ from $1$ to $n$ do:
    \begin{itemize}
    \item gather the (finite) incoming neighborhood $N^-(v)$ of $v$ then the (finite) outgoing neighborhood $N^+(N^-(v))$ and iterate this process until it converges (in finite time) to a set $I_v$ of inputs and $O_v$ of outputs with ${v\in O_v}$;
    \item check that all ${v'\in I_v\cup O_v}$ are such that ${v'\geq v}$ otherwise jump to next $v$ in the loop (this guaranties that each gate is generated only once);
    \item since the considered gates are irreducible, $I_v$ and $O_v$ actually correspond to input and output sets of a gate $g\in\GG$ that we can recover by finite checks from the local maps of nodes in $O_v$;
    \item output gate $g$ and the pairs source/image to describe $\alpha$ and $\beta$ for nodes in $I_v$ and $O_v$ respectively.
    \end{itemize}
  \end{itemize}
\end{remark}

In the sequel we denote $\Gamma(\mathcal{G})$ the family of all posible $\mathcal{G}$-networks associated to their bounded degree representation.

\subsection{Glueing of automata networks}
In this section we define an operation that allows us to 'glue' two different abstract automata networks on a common part in order to create another one which, roughly, preserve some dynamical properties in the sense that it allows to glue pseudo-orbits of each network to obtain a pseudo-orbit of the glued network.
One might find useful to think about the common part of the two networks as a dowel attaching two pieces of wood: each individual network is a piece of wood with the dowel inserted in it, and the result of the glueing is the attachment of the two pieces with a single dowel (see Figure~\ref{fig:glueingscheme}).

\begin{definition}\label{def:glueing}
  Consider ${F_1:Q^{V_1}\rightarrow Q^{V_1}}$ and ${F_2:Q^{V_2}\rightarrow Q^{V_2}}$  two automata networks with $V_1$ disjoint from $V_2$, $C$ a set disjoint from ${V_1\cup V_2}$, $\varphi_1: C \to V_1$ and $\varphi_2: C \to V_2$ two injective maps with ${\varphi_1(C)\cap\varphi_2(C)=\emptyset}$ and ${C_1,C_2}$ a partition of $C$ in two sets.
  We define \[V'= C\cup (V_1\setminus \varphi_1(C)) \cup (V_2\setminus \varphi_2(C))\] and the map $\alpha : V'\rightarrow V_1\cup V_2$ by 
  \[\alpha(v) =
    \begin{cases}
      v & \text{ if } v\not\in C\\
      \varphi_i(v) & \text{ if } v\in C_i, \text{ for }i=1,2.\\
    \end{cases}
  \]
  We then define the glueing of $F_1$ and $F_2$ over $C$ as the automata network $F' : Q^{V'}\rightarrow Q^{V'}$ where
  \[F'_v = \begin{cases}
      (F_1)_{\alpha(v)}\circ\rho_1 &\text{ if } \alpha(v) \in  V_1, \\
      (F_2)_{\alpha(v)}\circ\rho_2 &\text{ if } \alpha(v) \in  V_2,
    \end{cases}\]
  where ${\rho_i : Q^{V'}\rightarrow Q^{V_i}}$ is defined by 
  \[\rho_i(x)_v =
    \begin{cases}
      x_{\varphi_i^{-1}(v)} &\text{ if } v\in \varphi_i(C),\\
      x_v &\text{ else.}
    \end{cases}
  \]
\end{definition}

\begin{figure}
	\centering
\begin{tikzpicture}[x=0.75pt,y=0.75pt,yscale=-0.75,xscale=0.75]

\draw   (244,53) .. controls (244,41.95) and (278.14,33) .. (320.25,33) .. controls (362.36,33) and (396.5,41.95) .. (396.5,53) .. controls (396.5,64.05) and (362.36,73) .. (320.25,73) .. controls (278.14,73) and (244,64.05) .. (244,53) -- cycle ;
\draw   (115.5,98) .. controls (120.5,57) and (226.5,100) .. (209.5,124) .. controls (192.5,148) and (441.5,161) .. (235,201) .. controls (28.5,241) and (143,209) .. (79.5,180) .. controls (16,151) and (110.5,139) .. (115.5,98) -- cycle ;
\draw  [dash pattern={on 4.5pt off 4.5pt}] (158,171) .. controls (158,159.95) and (192.14,151) .. (234.25,151) .. controls (276.36,151) and (310.5,159.95) .. (310.5,171) .. controls (310.5,182.05) and (276.36,191) .. (234.25,191) .. controls (192.14,191) and (158,182.05) .. (158,171) -- cycle ;
\draw   (527.49,242.78) .. controls (522.49,283.78) and (416.49,240.78) .. (433.49,216.78) .. controls (450.49,192.78) and (201.49,179.78) .. (407.99,139.78) .. controls (614.49,99.78) and (499.99,131.78) .. (563.49,160.78) .. controls (626.99,189.78) and (532.49,201.78) .. (527.49,242.78) -- cycle ;
\draw  [dash pattern={on 4.5pt off 4.5pt}] (484.99,169.78) .. controls (484.99,180.82) and (450.85,189.78) .. (408.74,189.78) .. controls (366.63,189.78) and (332.49,180.82) .. (332.49,169.78) .. controls (332.49,158.73) and (366.63,149.78) .. (408.74,149.78) .. controls (450.85,149.78) and (484.99,158.73) .. (484.99,169.78) -- cycle ;

\draw   (193.5,285) .. controls (198.5,244) and (304.5,287) .. (287.5,311) .. controls (270.5,335) and (519.5,348) .. (313,388) .. controls (106.5,428) and (221,396) .. (157.5,367) .. controls (94,338) and (188.5,326) .. (193.5,285) -- cycle ;
\draw   (432.49,430.78) .. controls (427.49,471.78) and (321.49,428.78) .. (338.49,404.78) .. controls (355.49,380.78) and (106.49,367.78) .. (312.99,327.78) .. controls (519.49,287.78) and (404.99,319.78) .. (468.49,348.78) .. controls (531.99,377.78) and (437.49,389.78) .. (432.49,430.78) -- cycle ;
\draw   (237,358) .. controls (237,346.95) and (271.14,338) .. (313.25,338) .. controls (355.36,338) and (389.5,346.95) .. (389.5,358) .. controls (389.5,369.05) and (355.36,378) .. (313.25,378) .. controls (271.14,378) and (237,369.05) .. (237,358) -- cycle ;

\draw  [dash pattern={on 4.5pt off 4.5pt}]  (305.5,78) -- (271.01,137.41) ;
\draw [shift={(269.5,140)}, rotate = 300.14] [fill={rgb, 255:red, 0; green, 0; blue, 0 }  ][line width=0.08]  [draw opacity=0] (8.93,-4.29) -- (0,0) -- (8.93,4.29) -- cycle    ;
\draw  [dash pattern={on 4.5pt off 4.5pt}]  (334.5,79) -- (367.99,136.41) ;
\draw [shift={(369.5,139)}, rotate = 239.74] [fill={rgb, 255:red, 0; green, 0; blue, 0 }  ][line width=0.08]  [draw opacity=0] (8.93,-4.29) -- (0,0) -- (8.93,4.29) -- cycle    ;
\draw    (317.5,187) -- (317.5,254) ;
\draw [shift={(317.5,257)}, rotate = 270] [fill={rgb, 255:red, 0; green, 0; blue, 0 }  ][line width=0.08]  [draw opacity=0] (8.93,-4.29) -- (0,0) -- (8.93,4.29) -- cycle    ;

\draw (312,42.4) node [anchor=north west][inner sep=0.75pt]    {$C$};
\draw (228,160.4) node [anchor=north west][inner sep=0.75pt]    {$\varphi _{1}( C)$};
\draw (383,159.4) node [anchor=north west][inner sep=0.75pt]    {$\varphi _{2}( C)$};
\draw (114,132.4) node [anchor=north west][inner sep=0.75pt]    {$V_{1}$};
\draw (518,169.4) node [anchor=north west][inner sep=0.75pt]    {$V_{2}$};
\draw (244,87.4) node [anchor=north west][inner sep=0.75pt]    {$\varphi _{1}$};
\draw (373,88.4) node [anchor=north west][inner sep=0.75pt]    {$\varphi _{2}$};
\draw (305,347.4) node [anchor=north west][inner sep=0.75pt]    {$C$};
\draw (311,284.4) node [anchor=north west][inner sep=0.75pt]    {$V'$};

\end{tikzpicture}
\caption{Scheme of a glueing}
\label{fig:glueingscheme}

\end{figure}
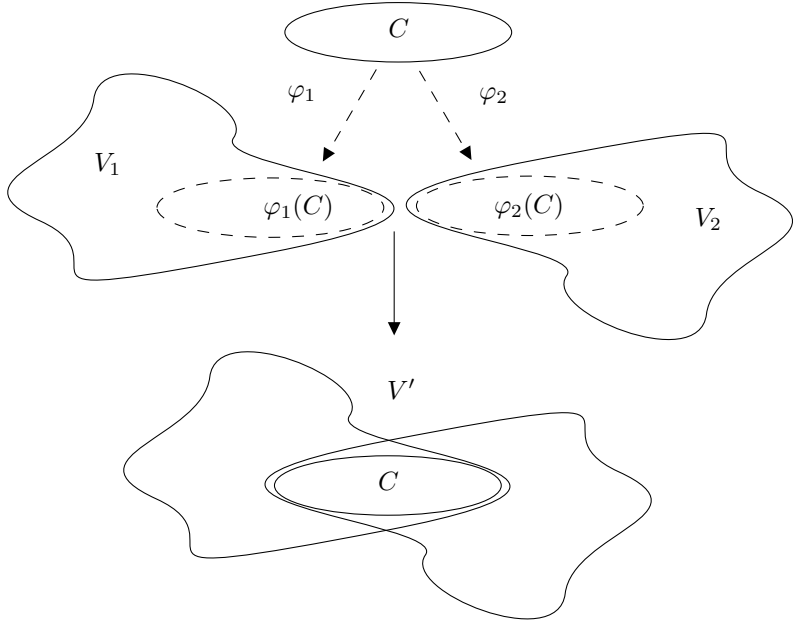

\newcommand\glueing[6]{#1\,{}_{#3}^{#5}\!\!\oplus_{#4}^{#6}\,#2}
When necessary, we will use the notation ${F' = \glueing{F_1}{F_2}{C_1}{C_2}{\phi_1}{\phi_2}}$ to underline the dependency of the glueing operation on its parameters.

Given an automata network ${F:Q^V\rightarrow Q^V}$ and a set ${X\subseteq V}$, we say that a sequence ${(x^t)_{0\leq t\leq T}}$ of configurations from $Q^V$ is a $X$-pseudo-orbit if it respects $F$ as in a normal orbit, except on $X$ where it can be arbitrary, formally: ${x^{t+1}_v = F(x^t)_v}$ for all ${v\in V\setminus X}$ and all ${0\leq t<T}$. The motivation for Definition~\ref{def:glueing} comes from the following lemma.

\begin{lemma}[Pseudo-orbits glueing]\label{lem:pseudo-orbit-glueing}
  Taking the notations of Definition~\ref{def:glueing}, let ${X\subseteq V_1\setminus\varphi_1(C)}$ and ${Y\subseteq V_2\setminus\varphi_2(C)}$ be two (possibly empty) sets. If ${(x^t)_{0\leq t\leq T}}$ is a ${X\cup\varphi_1(C_2)}$-pseudo-orbit for $F_1$ and if ${(y^t)_{0\leq t\leq T}}$ is a ${Y\cup\varphi_2(C_1)}$-pseudo-orbit for $F_2$ and if they verify for all ${0\leq t\leq T}$
  
  \begin{equation}
    \forall v\in C, x^t_{\varphi_1(v)}=y^t_{\varphi_2(v)},
    \label{eq:sametrace}
  \end{equation} 
  then the sequence ${(z^t)_{0\leq t\leq T}}$ of configurations of $Q^{V'}$ is a ${X\cup Y}$-pseudo-orbit of $F'$, where 
  \[z^t_v =
    \begin{cases}
      x^t_{\alpha(v)} &\text{ if } \alpha(v)\in V_1, \\
      y^t_{\alpha(v)} &\text{ if } \alpha(v)\in V_2.
    \end{cases}
  \]
\end{lemma}

\begin{proof}
  Take any ${v\in V'\setminus(X\cup Y)}$. Suppose first that ${\alpha(v)\in V_1}$. By definition of $F'$, we have ${F'(z^{t})_v = (F_1)_{\alpha(v)}\circ\rho_1 (z^t)}$ but ${\rho_1(z^t) = x^t}$ (using the Equation~\ref{eq:sametrace} in the hypothesis) so ${F'(z^t)_v=F_1(x^t)_{\alpha(v)}}$. Since ${(x^t)}$ is a ${X\cup\varphi_1(C_2)}$-pseudo-orbit and since ${\alpha(v)\not\in X\cup\varphi_1(C_2)}$, we have \[F_1(x^t)_{\alpha(v)}=x^{t+1}_{\alpha(v)}=z^{t+1}_v.\] We conclude that ${z^{t+1}_v=F'(z^{t})_v}$. 
  By a similar reasoning, we obtain the same conclusion if ${\alpha(v)\in V_2}$.
  We deduce that ${(z^t)}$ is a ${X\cup Y}$-pseudo-orbit of $F'$. 
\end{proof}
In order to illustrate the latter lemma, we show an example of glueing considering classical life-like automata network, given by the \emph{Game of life}. In this case, we have that $B=\{3\}$ and $S=\{2,3\}$ meaning that a dead cell can update its state to alive if it has exactly three alive neighbors and it will survive only it has exactly $2$ or $3$ alive neighbors.

First, we introduce, In Figure \ref{fig:clockgl}, the dynamics of a clock network (roughly, a network exhibiting a dynamics consisting in a periodic sequence of patterns that move from left to right).  This network will be very important to the construction we will show in the next section and for the example of pseudo-orbit glueing that we are going to introduce hereunder. Now, observe that the communication graph of the clock network is composed by six layers of three independent nodes connected to all the nodes in the next layer. The layers in gray boxes (the first and the last) are connected. In addition, all the layers are connected to an auxiliary node. This node will allow the layer to return to state $0$ when the next layer is in state $1$ since it will provide an additional node in state $1$ so all the nodes in the layer will have exactly four nodes in state $1$ (note that each node in a layer has three neighbors in the adjacent layers and the auxiliary node). Observe that in each time step two layers are in state one and the rest o the layers are in state zero. This pattern is shifted in each time step and thus it takes $6$ time steps in order to go from the first layer to the last one.

Now, let us consider a slightly different network, which is essentially the same as in Figure \ref{fig:clockgl} but the first and the last layer are not connected. We call this a wire network and we show its communication graph in Figure \ref{fig:wiregl}. We are going to glue two of the wire networks, using the previous lemma, in order to create a larger wire. In Figure \ref{fig:cablegluegl}, it is represented the glueing operation between two of these wire networks. The glueing parts are highlighted inside a box in both gadgets. Dashed boxes indicate the zones that are not ruled by the dynamics (the X, Y , $C_1$ and $C_2$ in the latter lemma). Observe that the pseudo-orbit which is induces the pattern that goes through the layers from left to right in both gadgets is preserved in the new network. This pseudo-orbit is shown in Figure \ref{fig:wiregl}. More precisely, the sequence of configurations $(w^i_1)_{i=0,\hdots,4}$ and $w_0$ are both $X$ and $Y$ pseudo-orbits in each of the wire networks that we are glueing. Observe that both of this sequences satisfy the conditions asked by the lemma and thus, they induce a $X\cup Y$-pseudo orbit on the glued network. 

\begin{figure}
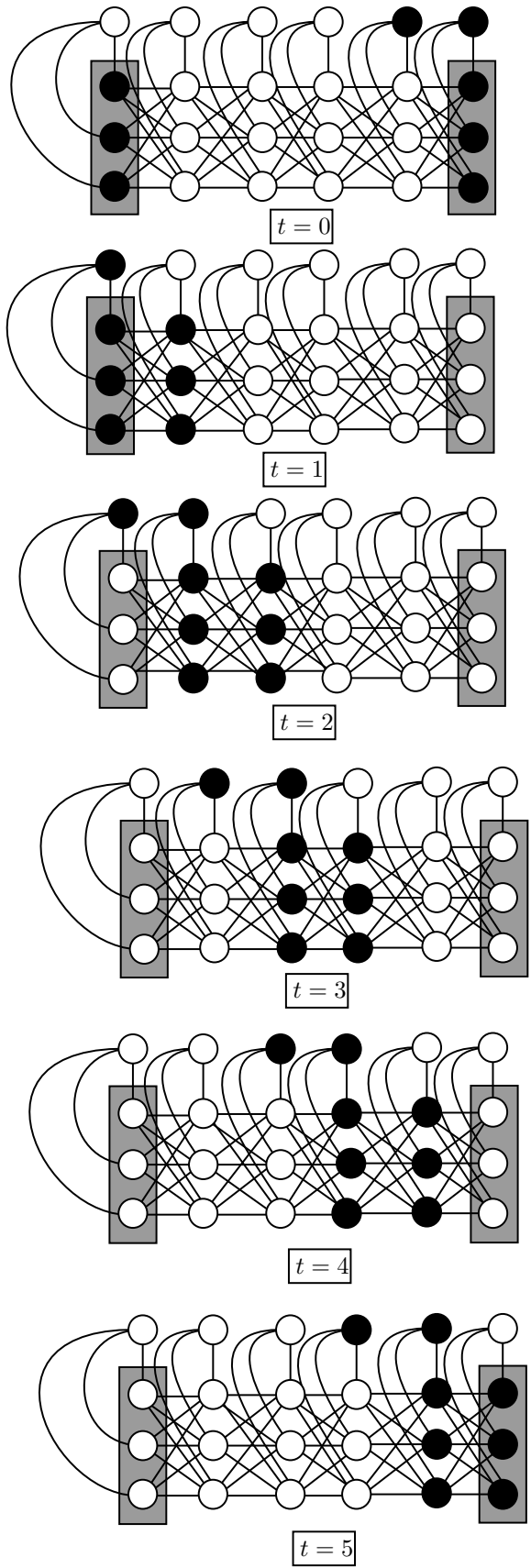


\centering
\tikzset{every picture/.style={line width=0.75pt}} 


\caption{An orbit of a clock network. Gray-shaded nodes are connected. Nodes in state $1$ are represented by black circles and those in state $0$ are represented by white circles. The represented orbit is periodic of period $6$.} 
\label{fig:clockgl}
\end{figure}
\begin{figure}
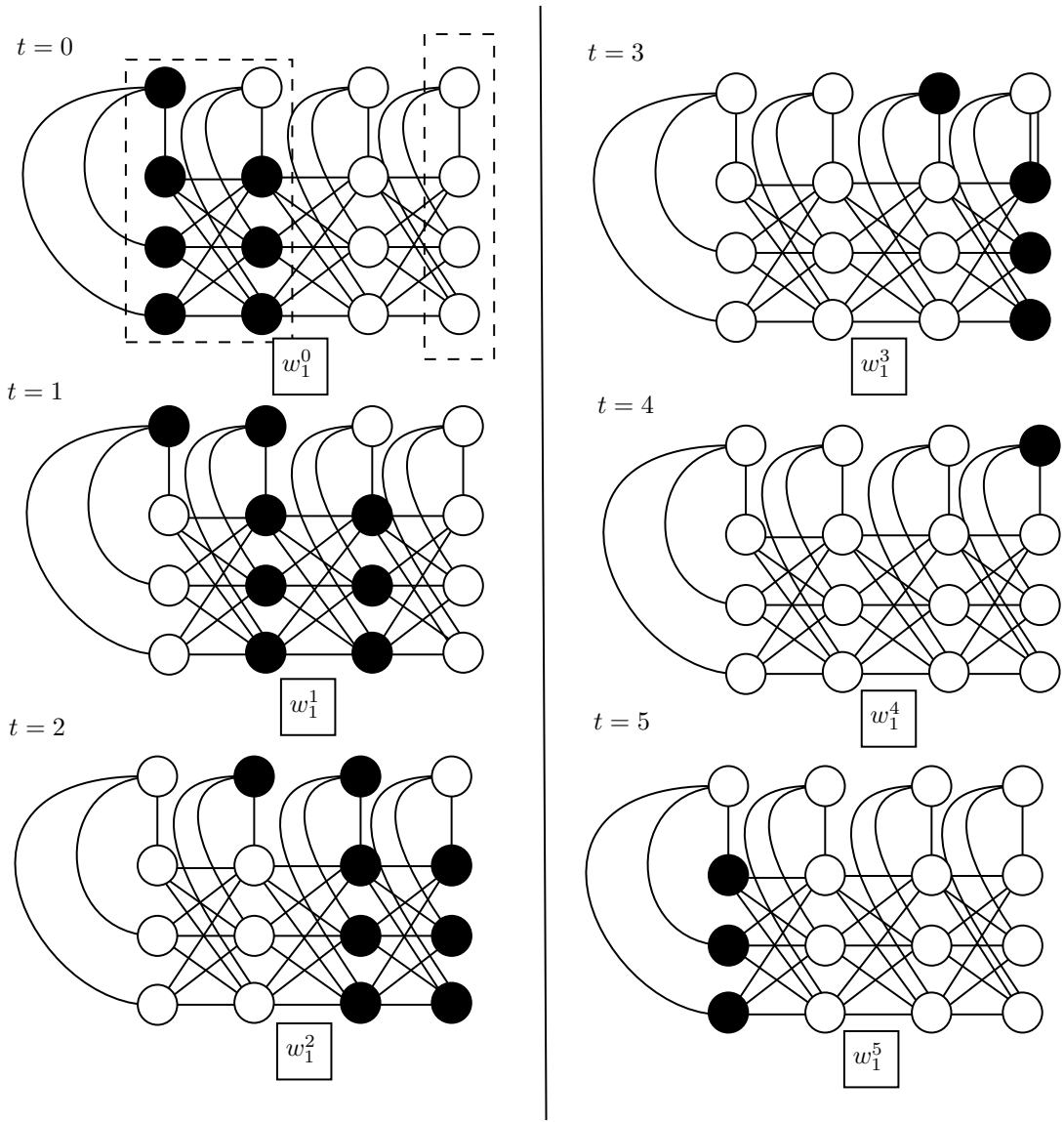


\centering

\tikzset{every picture/.style={line width=0.75pt}} 

%

\caption{A pseudo-orbit representing a signal on a wire network. Black nodes are in state $1$ and white nodes are in state $0$. The configuration at each time step is represented by the notation $w^t_i, t=0,1,2,3,4,5 i=0,1.$ This notation stands for the configuration in time $t$ that codes a signal representing the bit $i.$ Thus, the pseudo-orbit represented is $(w^0_1,w^1_1,w^2_1,w^3_1,w^4_1,,w^5_1).$  In time $t=0$, it is marked in dotted boxes the parts of the networks that do not depend on the local rule of the network.} 
\label{fig:wiregl}
\end{figure}

\begin{figure}
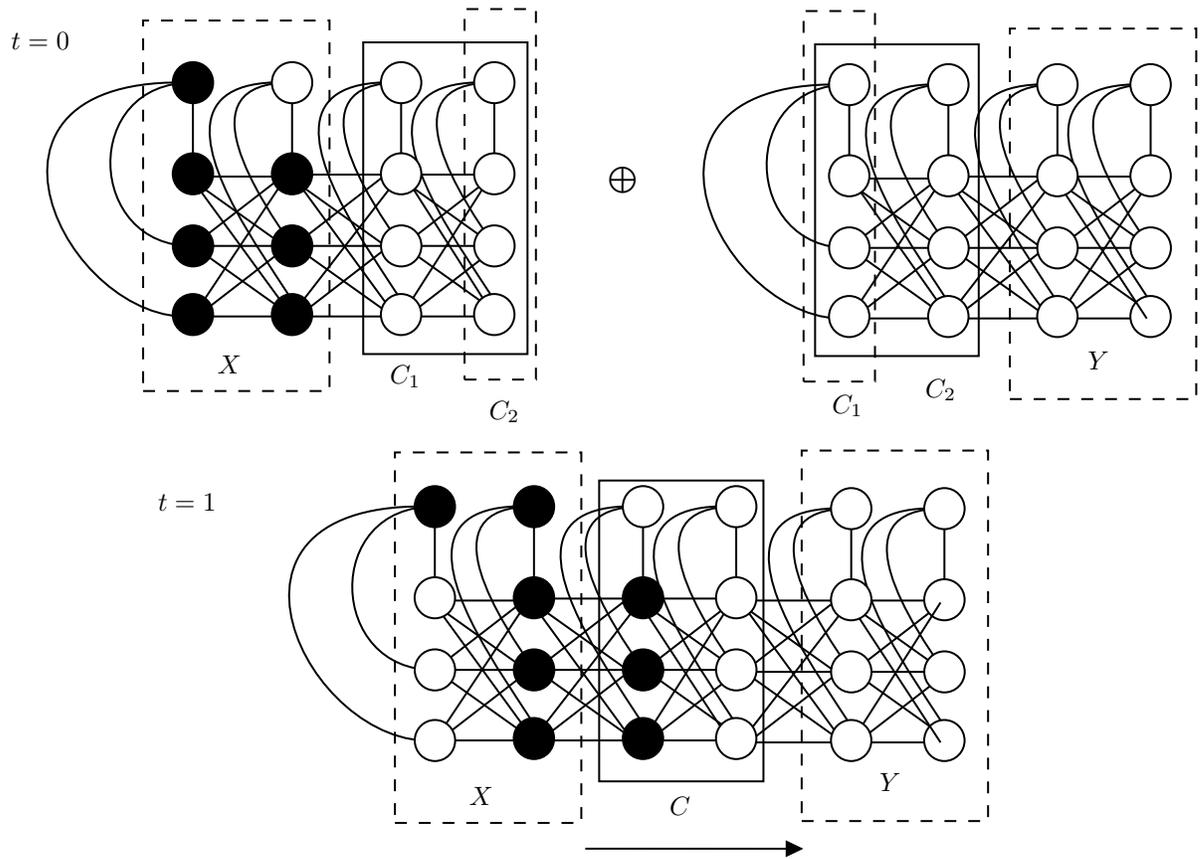

\centering

\tikzset{every picture/.style={line width=0.75pt}} 


\caption{Example of two wire gadgets transmitting a $1$ signal represented as a particular pseudo-orbit. Nodes in black are in state $1$ and nodes in white are in state $0$. Nodes inside dashed boxes represent the zones that are arbitrary for the pseudo-orbit. At time $t=0$ it is shown a $X\cup C_2$-pseudo-orbit for the left gadget and in the right a $Y \cup C_1$-pseudo-orbit.  At time $t=1$ it is shown the result of the glueing of the two gadgets and the resulting $X\cup Y$-pseudo-orbit given by the theorem.}
\label{fig:cablegluegl}
\end{figure}

Observe that in the latter example, the glueing is straightforward not only because the structure of the graph is uniform but due to the fact that the local rule is CSAN. However, even in the case of a CSAN family where the transition rules are determined by a labeled non-directed graph, the result of a glueing operation has no reason to belong to the family because the symmetry of the interaction graph might be broken (see Figure~\ref{fig:nonsymglueing}).
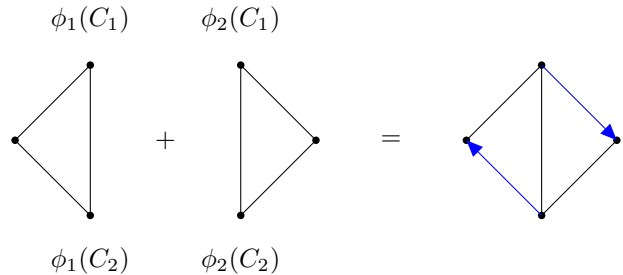
\begin{figure}
  \centering
  \begin{tikzpicture}
    \draw (2,0) node {$+$};
    \draw (5,0) node {$=$};
    \tikzstyle{every node}=[fill,shape=circle,inner sep=1pt]
    \coordinate[label = above:$\phi_1(C_1)$] (A) at (1,1);
    \coordinate[label = below:$\phi_1(C_2)$] (B) at (1,-1);
    \draw (0,0) node {} -- (A) node {} -- (B) node {} -- cycle;
    \coordinate[label = above:$\phi_2(C_1)$] (C) at (3,1);
    \coordinate[label = below:$\phi_2(C_2)$] (D) at (3,-1);
    \draw (4,0) node {} -- (C) node {} -- (D) node {} -- cycle;
    \draw (6,0) node {} -- (7,1) node {} -- (7,-1) node {} -- (8,0) node {};
    \draw[blue,-triangle 45] (7,1) -- (8,0);
    \draw[blue,-triangle 45] (7,-1) -- (6,0);
  \end{tikzpicture}
  \caption{Symmetry breaking in interaction graph after a glueing operation. Arrows indicate influence of a node (source) on another (target), edges without arrow indicates bi-directional influence. Here $C$ consists in two nodes only.}
  \label{fig:nonsymglueing}
\end{figure}
The following lemma gives a sufficient condition in graph theoretical terms for glueing within a concrete family of automata networks. Intuitively, it consists in asking that, in each graph, all the connections of one half of the dowel to the rest of the graph goes through the other half of the dowel. Here the wooden dowel metaphor is particularly relevant: when considering a single piece of wood with the dowel inserted inside, one half of the dowel is 'inside' (touches the piece of wood), the other half is 'outside' (not touching the piece of wood); then, when the two pieces are attached, each position in the wood assembly is locally either like in one piece of wood with the dowel inserted or like in the other one with the dowel inserted.

\begin{lemma}[Glueing for CSAN]
  \label{lem:concreteglueing}
  Let $(G_1,\lambda_1,\rho_1)$ and $(G_2,\lambda_2,\rho_2)$ be two CSAN from the same CSAN family $\mathcal{F}$ where $G_1$ and $G_2$ are disjoint and $F_1$ and $F_2$ are the associated global maps. Taking again the notations of Definition~\ref{def:glueing}, if the following conditions hold
  \begin{itemize}
  \item the labeled graphs induced by $\varphi_1(C)$ and $\varphi_2(C)$ in $G_1$ and $G_2$ are the same (using the identification ${\varphi_1(v)=\varphi_2(v)}$)
  \item ${N_{G_1}(\varphi_1(C_2))\subseteq \varphi_1(C)}$ 
  \item ${N_{G_2}(\varphi_2(C_1))\subseteq \varphi_2(C)}$
  \end{itemize}
  then the glueing ${F' = \glueing{F_1}{F_2}{C_1}{C_2}{\phi_1}{\phi_2}}$ can be defined as the CSAN on graph $G'=(V',E')$ where $V'$ is as in Definition~\ref{def:glueing}
  and each node ${v\in V'}$ has the same label and same labeled neighborhood as either a node of ${(G_1,\lambda_1,\rho_1)}$ or a node of ${(G_2,\lambda_2,\rho_2)}$. In particular $F'$ belongs to ${\mathcal{F}}$.
\end{lemma} 
\begin{proof} 
  Let us define ${\beta_i : V_i\to V'}$ by
  \[\beta_i(v) =
    \begin{cases}
      \phi_i^{-1}(v) &\text{ if }v\in\phi_i(C),\\
      v&\text{ else.}
    \end{cases}
  \]
  Fix ${i=1}$ or $2$. According to Definition~\ref{def:glueing}, if ${v\in V'}$ is such that ${\alpha(v) \in  V_i}$ then ${F'_v=(F_i)_{\alpha(v)}\circ\rho_i}$. By definition of CSAN, this means that for any ${x\in Q^{V'}}$ we have ${F'_v(x) = \psi_{i,\alpha(v)}(x_{|\beta(N_{G_i}(\alpha(v)))})}$ where $\psi_{i,\alpha(v)}$ is a map depending only on the labeled neighborhood of $\alpha(v)$ in $G_i$ as in Definition~\ref{def:csan}. So the dependencies of $v$ in $F'$ are in one-to-one correspondence through $\beta$ with the neighborhood of $\alpha(v)$ in $G_i$. They key observation is that the symmetry of dependencies is preserved, formally for any ${v'\in \beta(N_{G_i}(\alpha(v)))}$:
  \begin{itemize}
  \item either $\alpha(v')\in V_i$ in which case the dependency of $v'$ on $v$ (in map $\psi_{i,\alpha(v')}$) is the same as the dependency of $v$ on $v'$ (in map $\psi_{i,\alpha(v)}$), and both are determined by the undirected labeled edge ${\{\alpha(v),\alpha(v')\}}$ of $G_i$;
  \item or $\alpha(v')\not\in V_{i}$ and in this case necessarily ${v\in C_i}$ and ${v'\in C_{3-i}}$ (because ${N_{G_i}(V_i\setminus\varphi_i(C))\cap\varphi_i(C)\subseteq \varphi_i(C_i)}$ from the hypothesis), so the dependency of $v'$ on $v$ is the same as the dependency of $v$ on $v'$ because the labeled graphs induced by $\phi_1(C)$ and $\phi_2(C)$ in $G_1$ and $G_2$ are the same.
  \end{itemize}

  Concretely, $F'$ is a CSAN that can be defined on graph ${G'=(V',E_1'\cup E_2'\cup E(C))}$ with
  \[E_i' = E(V_i \setminus \varphi_i(C)) \cup \{(u,v_i): u \in V(C_i), v_i \in (V_i \setminus \varphi_i(C)),  (\varphi_i(u),v_i) \in  E_i\},\]
  and labels as follows:
  \begin{itemize}
  \item on $E(C)$ as in both $G_1$ and $G_2$ (which agree through maps $\phi_1$ and $\phi_2$ on $C$), 
  \item on ${E(V_i \setminus \varphi_i(C))}$ as in $G_i$,
  \item for each $u \in V(C_i), v_i \in (V_i \setminus \varphi_i(C))$ such that  $(\varphi_i(u),v_i) \in  E_i$, edge $(u,v_i)$ has same label as $(\varphi_i(u),v_i)$.
  \end{itemize}
  Since any CSAN families (Definition~\ref{def:csan}) is entirely based on local constraints on labels (vertex label plus set of labels of the incident edges), we deduce that $F'$ is in $\mathcal{F}$. 
\end{proof}

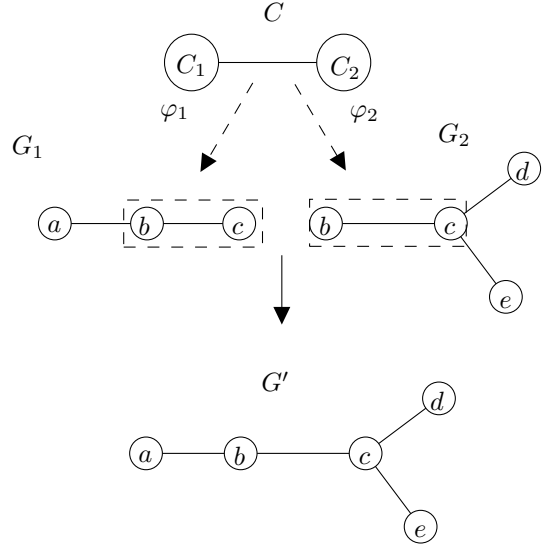
\begin{figure}\label{fig:concreteglueing}
\centering
\begin{tikzpicture}[x=0.75pt,y=0.75pt,yscale=-1,xscale=1]

\draw    (110,239) -- (156.5,239) ;
\draw    (300.75,95.25) -- (263.75,123) ;
\draw    (263.75,123) -- (291.5,160) ;
\draw    (209,123) -- (255.5,123) ;
\draw  [fill={rgb, 255:red, 255; green, 255; blue, 255 }  ,fill opacity=1 ] (192.5,123) .. controls (192.5,118.44) and (196.19,114.75) .. (200.75,114.75) .. controls (205.31,114.75) and (209,118.44) .. (209,123) .. controls (209,127.56) and (205.31,131.25) .. (200.75,131.25) .. controls (196.19,131.25) and (192.5,127.56) .. (192.5,123) -- cycle ;
\draw  [fill={rgb, 255:red, 255; green, 255; blue, 255 }  ,fill opacity=1 ] (255.5,123) .. controls (255.5,118.44) and (259.19,114.75) .. (263.75,114.75) .. controls (268.31,114.75) and (272,118.44) .. (272,123) .. controls (272,127.56) and (268.31,131.25) .. (263.75,131.25) .. controls (259.19,131.25) and (255.5,127.56) .. (255.5,123) -- cycle ;
\draw  [fill={rgb, 255:red, 255; green, 255; blue, 255 }  ,fill opacity=1 ] (283.25,160) .. controls (283.25,155.44) and (286.94,151.75) .. (291.5,151.75) .. controls (296.06,151.75) and (299.75,155.44) .. (299.75,160) .. controls (299.75,164.56) and (296.06,168.25) .. (291.5,168.25) .. controls (286.94,168.25) and (283.25,164.56) .. (283.25,160) -- cycle ;
\draw  [fill={rgb, 255:red, 255; green, 255; blue, 255 }  ,fill opacity=1 ] (292.5,95.25) .. controls (292.5,90.69) and (296.19,87) .. (300.75,87) .. controls (305.31,87) and (309,90.69) .. (309,95.25) .. controls (309,99.81) and (305.31,103.5) .. (300.75,103.5) .. controls (296.19,103.5) and (292.5,99.81) .. (292.5,95.25) -- cycle ;

\draw    (64,123) -- (110.5,123) ;
\draw  [fill={rgb, 255:red, 255; green, 255; blue, 255 }  ,fill opacity=1 ] (55.75,123) .. controls (55.75,118.44) and (59.44,114.75) .. (64,114.75) .. controls (68.56,114.75) and (72.25,118.44) .. (72.25,123) .. controls (72.25,127.56) and (68.56,131.25) .. (64,131.25) .. controls (59.44,131.25) and (55.75,127.56) .. (55.75,123) -- cycle ;
\draw    (110.5,123) -- (157,123) ;
\draw  [fill={rgb, 255:red, 255; green, 255; blue, 255 }  ,fill opacity=1 ] (102.25,123) .. controls (102.25,118.44) and (105.94,114.75) .. (110.5,114.75) .. controls (115.06,114.75) and (118.75,118.44) .. (118.75,123) .. controls (118.75,127.56) and (115.06,131.25) .. (110.5,131.25) .. controls (105.94,131.25) and (102.25,127.56) .. (102.25,123) -- cycle ;
\draw  [fill={rgb, 255:red, 255; green, 255; blue, 255 }  ,fill opacity=1 ] (148.75,123) .. controls (148.75,118.44) and (152.44,114.75) .. (157,114.75) .. controls (161.56,114.75) and (165.25,118.44) .. (165.25,123) .. controls (165.25,127.56) and (161.56,131.25) .. (157,131.25) .. controls (152.44,131.25) and (148.75,127.56) .. (148.75,123) -- cycle ;

\draw    (132.95,41.71) -- (209.72,41.71) ;
\draw  [fill={rgb, 255:red, 255; green, 255; blue, 255 }  ,fill opacity=1 ] (119.33,41.71) .. controls (119.33,33.82) and (125.43,27.42) .. (132.95,27.42) .. controls (140.48,27.42) and (146.57,33.82) .. (146.57,41.71) .. controls (146.57,49.6) and (140.48,56) .. (132.95,56) .. controls (125.43,56) and (119.33,49.6) .. (119.33,41.71) -- cycle ;
\draw  [fill={rgb, 255:red, 255; green, 255; blue, 255 }  ,fill opacity=1 ] (196.1,41.71) .. controls (196.1,33.82) and (202.2,27.42) .. (209.72,27.42) .. controls (217.24,27.42) and (223.34,33.82) .. (223.34,41.71) .. controls (223.34,49.6) and (217.24,56) .. (209.72,56) .. controls (202.2,56) and (196.1,49.6) .. (196.1,41.71) -- cycle ;
\draw    (257.75,211.25) -- (220.75,239) ;
\draw    (220.75,239) -- (248.5,276) ;
\draw    (166,239) -- (212.5,239) ;
\draw  [fill={rgb, 255:red, 255; green, 255; blue, 255 }  ,fill opacity=1 ] (149.5,239) .. controls (149.5,234.44) and (153.19,230.75) .. (157.75,230.75) .. controls (162.31,230.75) and (166,234.44) .. (166,239) .. controls (166,243.56) and (162.31,247.25) .. (157.75,247.25) .. controls (153.19,247.25) and (149.5,243.56) .. (149.5,239) -- cycle ;
\draw  [fill={rgb, 255:red, 255; green, 255; blue, 255 }  ,fill opacity=1 ] (212.5,239) .. controls (212.5,234.44) and (216.19,230.75) .. (220.75,230.75) .. controls (225.31,230.75) and (229,234.44) .. (229,239) .. controls (229,243.56) and (225.31,247.25) .. (220.75,247.25) .. controls (216.19,247.25) and (212.5,243.56) .. (212.5,239) -- cycle ;
\draw  [fill={rgb, 255:red, 255; green, 255; blue, 255 }  ,fill opacity=1 ] (240.25,276) .. controls (240.25,271.44) and (243.94,267.75) .. (248.5,267.75) .. controls (253.06,267.75) and (256.75,271.44) .. (256.75,276) .. controls (256.75,280.56) and (253.06,284.25) .. (248.5,284.25) .. controls (243.94,284.25) and (240.25,280.56) .. (240.25,276) -- cycle ;
\draw  [fill={rgb, 255:red, 255; green, 255; blue, 255 }  ,fill opacity=1 ] (249.5,211.25) .. controls (249.5,206.69) and (253.19,203) .. (257.75,203) .. controls (262.31,203) and (266,206.69) .. (266,211.25) .. controls (266,215.81) and (262.31,219.5) .. (257.75,219.5) .. controls (253.19,219.5) and (249.5,215.81) .. (249.5,211.25) -- cycle ;

\draw  [fill={rgb, 255:red, 255; green, 255; blue, 255 }  ,fill opacity=1 ] (101.75,239) .. controls (101.75,234.44) and (105.44,230.75) .. (110,230.75) .. controls (114.56,230.75) and (118.25,234.44) .. (118.25,239) .. controls (118.25,243.56) and (114.56,247.25) .. (110,247.25) .. controls (105.44,247.25) and (101.75,243.56) .. (101.75,239) -- cycle ;
\draw  [dash pattern={on 4.5pt off 4.5pt}] (98.75,111.13) -- (168.75,111.13) -- (168.75,134.88) -- (98.75,134.88) -- cycle ;
\draw  [dash pattern={on 4.5pt off 4.5pt}] (192.5,110.75) -- (271.5,110.75) -- (271.5,134.5) -- (192.5,134.5) -- cycle ;
\draw  [dash pattern={on 4.5pt off 4.5pt}]  (163.85,52) -- (138.96,94.41) ;
\draw [shift={(137.44,97)}, rotate = 300.40999999999997] [fill={rgb, 255:red, 0; green, 0; blue, 0 }  ][line width=0.08]  [draw opacity=0] (8.93,-4.29) -- (0,0) -- (8.93,4.29) -- cycle    ;
\draw  [dash pattern={on 4.5pt off 4.5pt}]  (185.12,52.73) -- (209.27,93.69) ;
\draw [shift={(210.79,96.27)}, rotate = 239.48] [fill={rgb, 255:red, 0; green, 0; blue, 0 }  ][line width=0.08]  [draw opacity=0] (8.93,-4.29) -- (0,0) -- (8.93,4.29) -- cycle    ;
\draw    (178.5,139) -- (178.5,171) ;
\draw [shift={(178.5,174)}, rotate = 270] [fill={rgb, 255:red, 0; green, 0; blue, 0 }  ][line width=0.08]  [draw opacity=0] (8.93,-4.29) -- (0,0) -- (8.93,4.29) -- cycle    ;

\draw (59,112.4+8) node [anchor=north west][inner sep=0.75pt]    {$a$};
\draw (105,113.4+5) node [anchor=north west][inner sep=0.75pt]    {$b$};
\draw (152,113.4+8) node [anchor=north west][inner sep=0.75pt]    {$c$};
\draw (196,113.4+5) node [anchor=north west][inner sep=0.75pt]    {$b$};
\draw (259,113.4+8) node [anchor=north west][inner sep=0.75pt]    {$c$};
\draw (295,85.4+5) node [anchor=north west][inner sep=0.75pt]    {$d$};
\draw (287,150.4+8) node [anchor=north west][inner sep=0.75pt]    {$e$};
\draw (244,266.4+8) node [anchor=north west][inner sep=0.75pt]    {$e$};
\draw (252,201.4+5) node [anchor=north west][inner sep=0.75pt]    {$d$};
\draw (216,229.4+8) node [anchor=north west][inner sep=0.75pt]    {$c$};
\draw (153,229.4+5) node [anchor=north west][inner sep=0.75pt]    {$b$};
\draw (124.01,28.42+8) node [anchor=north west][inner sep=0.75pt]    {$C_{1}$};
\draw (201.61,29.65+8) node [anchor=north west][inner sep=0.75pt]    {$C_{2}$};
\draw (105,228.4+8) node [anchor=north west][inner sep=0.75pt]    {$a$};
\draw (41,77.4) node [anchor=north west][inner sep=0.75pt]    {$G_{1}$};
\draw (256,71.4) node [anchor=north west][inner sep=0.75pt]    {$G_{2}$};
\draw (167,196.4) node [anchor=north west][inner sep=0.75pt]    {$G'$};
\draw (168,10.4) node [anchor=north west][inner sep=0.75pt]    {$C$};
\draw (115.94,61.46) node [anchor=north west][inner sep=0.75pt]    {$\varphi _{1}$};
\draw (211.56,62.19) node [anchor=north west][inner sep=0.75pt]    {$\varphi _{2}$};
\end{tikzpicture}

      
        

\caption{Example of a glueing of two compatibles CSAN. The labeling in nodes of $G_1$, $G_2$ and $G'$ shows equalities between local $\lambda$ maps of these three CSAN.}
\end{figure}

\subsection{$\GG$-gadgets, gadget glueing and simulation of $\GG$-networks}

We now give a precise meaning to the intuitively simple fact that, if a family of automata networks can coherently simulate a set of small building blocks (gates from $\GG$), it should be able to simulate any automata network that can be built out of them ($\GG$-networks).

The key idea here is that gates from $\GG$ will be represented by networks of the family called $\GG$-gadgets, and the wiring between gates to obtain a 
$\GG$-network will translate into glueing between $\GG$-gadgets. Following this idea there are two main conditions for the family to simulate any $\GG$-network:
\begin{itemize}
\item the glueing of gadgets should be freely composable inside the family to allow the building of any $\GG$-network;
\item the gadgets corresponding to gates from $\GG$ should correctly and coherently simulate the functional relation between inputs and outputs given by their corresponding gate.
\end{itemize}

For clarity, we separate these conditions in two definitions.


We start by developing a definition for gadget glueing. Recall first that Definition~\ref{def:glueing} relies on the identification of a common dowel in the two networks to be glued. Here, as we want to mimic the wiring of gates which connects inputs to outputs, several copies of a fixed network called \textit{glueing interface} will be identified in each gadget, some of them corresponding to input, and the other ones to outputs. In this context,  the only glueing operations we will use are those where some output copies of the interface in a gadget $A$ are glued on input copies of the interface in a gadget $B$ and some input copies of the interface in $A$ are glued on output copies of the interface in $B$. Then, the global dowel used to formally apply Definition~\ref{def:glueing} is a disjoint union of the selected input/output copies of the interface. Figure~\ref{fig:gadgetglue} illustrates with the notations of the following Definition.

\newcommand\QF{Q_{\mathcal F}}

\begin{definition}[Glueing interface and gadgets]\label{def:gadgets-glueing}
  Let ${C=C_i\cup C_o}$ be a fixed set partitioned into two sets.
  A \emph{gadget} with \textit{glueing interface} ${C=C_i\cup C_o}$ is an automata network ${F:Q^{V_F}\rightarrow Q^{V_F}}$ together with two collections of injective maps ${\phi_{F,k}^i:C\rightarrow V_F}$ for ${k\in I(F)}$ and ${\phi_{F,k}^o:C\rightarrow V_F}$ for ${k\in O(F)}$ whose images in $V_F$ are pairwise disjoint and where $I(F)$ and $O(F)$ are disjoint sets which might be empty.

  Given two disjoint gadgets ${(F,(\phi_{F,k}^i),(\phi_{F,k}^o))}$ and ${(G,(\phi_{G,k}^i),(\phi_{G,k}^o))}$ with same alphabet and interface ${C=C_i\cup C_o}$, a \emph{gadget glueing} is a glueing of the form ${H = \glueing{F}{G}{C_F}{C_G}{\phi_F}{\phi_G}}$ defined as follows: 
  \begin{itemize}
  \item a choice of a set $A$ of inputs from $F$ and outputs from $G$ given by injective maps ${\sigma_F : A \to I(F)}$ and ${\sigma_G: A\to O(G)}$,
  \item a choice of a set $B$ of outputs from $F$ and inputs from $G$ given by injective maps ${\tau_F : B \to O(F)}$ and ${\tau_G: B\to I(G)}$ (the set $B$ is disjoint from $A$),
  \item ${C_F}$ is a disjoint union of $|A|$ copies of $C_i$, and $|B|$ copies of $C_o$: ${C_F = A\times C_i\cup B\times C_o}$,
  \item ${C_G}$ is a disjoint union of $|A|$ copies of $C_o$, and $|B|$ copies of $C_i$: ${C_G = A\times C_o\cup B\times C_i}$,
  \item ${\phi_F:C_F\cup C_G\to V_F}$ is such that ${\phi_F(a,c) = \phi_{F,\sigma_F(a)}^i(c)}$ for ${a\in A}$ and ${c\in C}$, and ${\phi_F(b,c) = \phi_{F,\tau_F(b)}^o(c)}$ for ${b\in B}$ and ${c\in C}$,
  \item ${\phi_G:C_F\cup C_G\to V_G}$ is such that ${\phi_G(a,c) = \phi_{G,\sigma_G(a)}^o(c)}$ for ${a\in A}$ and ${c\in C}$, and ${\phi_G(b,c) = \phi_{G,\tau_G(b)}^i(c)}$ for ${b\in B}$ and ${c\in C}$. 
  \end{itemize}
  The resulting network $H$ is a gadget with same alphabet and same interface with ${I(H) = I(F)\setminus\sigma_F(A) \cup I(G)\setminus\tau_G(B)}$ and ${O(H) = O(F)\setminus\tau_F(B)\cup O(F)\setminus\sigma_G(A)}$ and ${\phi_{H,k}^i}$ is $\phi_{F,k}^i$ when $k\in I(F)$ and ${\phi_{G,k}^i}$ when ${k\in I(G)}$, and ${\phi_{H,k}^o}$ is ${\phi_{F,k}^o}$ when ${k\in O(F)}$ and ${\phi_{G,k}^o}$ when ${k\in O(G)}$.

  Given a set of gadgets $X$ with same alphabet and interface, its closure by gadget glueing is the closure of $X$ by the following operations:
  \begin{itemize}
  \item add a disjoint copy of some gadget from the current set,
  \item add the disjoint union of two gadgets from the current set,
  \item add a gadget glueing of two gadgets from the current set.
  \end{itemize}
\end{definition}

\begin{remark}\label{rem:gadgets-glueing}
  The representation of the result of a gadget glueing can be easily computed from the two gadgets $F$ and $G$ and the choices of inputs/outputs given by maps $\sigma_F$, $\sigma_G$, $\tau_F$ and $\tau_G$: precisely, by definition of glueing (Definition~\ref{def:glueing}) the local map of each node of the result automata network is either a local map of $F$ (when in $V_F\setminus \phi_F(C_G)$ or in $C_F$) or a local map of $G$ (when in $V_G\setminus \phi_G(C_F)$ or in $C_G$). Note also that the closure by gadget glueing of a finite set of gadgets $X$ is always a set of automata networks of bounded degree.
\end{remark}

\begin{figure}
  \begin{center}
    \begin{tikzpicture}[
      mnode/.style={circle,fill=white,draw=black,minimum size=.7cm},
      minode/.style={circle,fill=red!20!white,draw=black,minimum size=.7cm},
      mvnode/.style={circle,fill=green!20!white,draw=black,minimum size=.7cm},
      monode/.style={circle,fill=blue!20!white,draw=black,minimum size=.7cm}
      ]
      \draw[dashed] (1.5,2.5)--(1.5,-5);
      \draw (.5,2) node {inputs};
      \draw (2.5,2) node {outputs};
      \draw (1,0) -- (2,-1) -- (2,1);
      \draw (-1,0) node {$F$};
      \draw (.5,.7) node {$\sigma_F(A)$};
      \draw (2.5,-1.7) node {$\tau_F(B)$};
      \draw (0,0) node[minode] {$\alpha$} --  (1,0) node[monode] {$\beta$} -- (2,1) node[minode] {$\gamma$} -- (3,1) node[monode] {$\delta$};
      \draw (2,-1) node[minode] {$\epsilon$} -- (3,-1) node[monode] {$\zeta$};
      \begin{scope}[shift={(0,-2)}]
        \draw (1,-2) -- (2,-1);
      \draw (.5,-1.3) node {$\sigma_G(B)$};
      \draw (2.5,-1.7) node {$\tau_G(A)$};
        \draw (-1,-1) node {$G$};
        \draw (0,0) node[minode] {$a$} -- (1,0) node[monode] {$b$} -- (2,-1) node[minode] {$c$} -- (3,-1) node[monode] {$d$};
        \draw (0,-2) node[minode] {$e$} -- (1,-2) node[monode] {$f$};
      \end{scope}
      \begin{scope}[shift={(7,.5)}]
        \draw (-1.5,1) node {G};
        \draw (0,0) -- (0,-1);
        \draw (-1,0) -- (0,0);
        \draw[dotted] (0,0) -- (.5,1);
        \draw[dotted] (1,0) -- (1.5,1);
        \draw[dotted] (0,-1) -- (.5,-2);
        \draw[dotted] (1,-1) -- (1.5,-2);
        \draw (0,0) node[mvnode] {$c$} -- (1,0) node[mnode] {$d$};
        \draw (0,-1) node[mvnode] {$f$} -- (1,-1) node[mnode] {$e$};
        \draw (-2,0) node[mnode] {$a$} -- (-1,0) node[mnode] {$b$};
      \end{scope}
      \begin{scope}[shift={(7.5,.5)}]
        \draw (2.5,1) node {F};
        \draw (1,1) -- (1,-2);
        \draw (1,-2) -- (2,0);
        \draw (0,1) node[mnode] {$\alpha$} -- (1,1) node[mvnode] {$\beta$} -- (2,0) node[mnode] {$\gamma$} -- (3,0) node[mnode] {$\delta$};
        \draw (0,-2) node[mnode] {$\zeta$} -- (1,-2) node[mvnode] {$\epsilon$};
      \end{scope}
      \draw[very thick] (4,-2.5) -- (11,-2.5);
      \draw[very thick] (4,2.5) -- (4,-5.5);
      \begin{scope}[shift={(7,-4)}]
        \draw (-1.5,1) node {input};
        \draw (2.5,1) node {output};
        \draw (0,0) -- (0,-1);
        \draw (1,0) -- (1,-1);
        \draw (1,-1) -- (2,0);
        \draw (-1,0) -- (0,0);
        \draw (0,0) node[mnode] {$c$} -- (1,0) node[mnode] {$\beta$} -- (2,0) node[minode] {$\gamma$} -- (3,0) node[monode] {$\delta$};
        \draw (0,-1) node[mnode] {$f$} -- (1,-1) node[mnode] {$\epsilon$};
        \draw (-2,0) node[minode] {$a$} -- (-1,0) node[monode] {$b$};
      \end{scope}
    \end{tikzpicture}
  \end{center}
  \caption{\label{fig:gadgetglue}Gadget glueing as in Definition~\ref{def:gadgets-glueing}. On the left, two gadgets with interface ${C=C_i\cup C_o}$ where $C_i$ part in each copy of the interface dowel is in red and $C_o$ part in blue. The gadget glueing is done with input ${\sigma_F(A)}$ on output ${\tau_G(A)}$ (here $A$ is a singleton) and output ${\tau_F(B)}$ on input ${\tau_G(B)}$ ($B$ is also a singleton). On the upper right, a representation of the global glueing process where nodes in green are those in the copy of $C_F$ in $F$ or in the copy $C_G$ in $G$; dotted links show the bijection between the embeddings of ${C=C_F\cup C_G}$ into $V_F$ and $V_G$ via maps $\phi_F$ and $\phi_G$. On the lower right the resulting gadget with the same interface ${C=C_i\cup C_o}$ as the two initial gadgets.}
\end{figure}
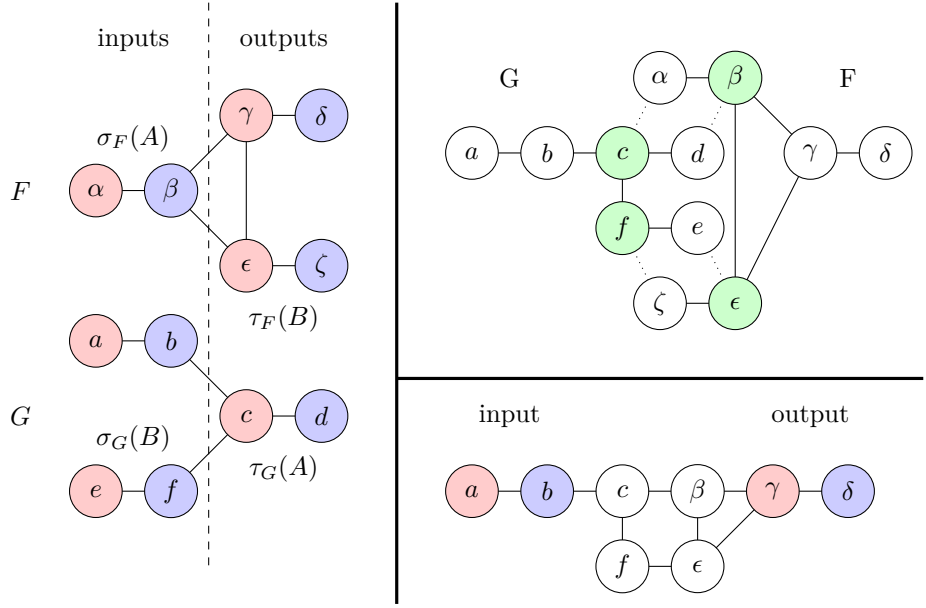

Lemma~\ref{lem:concreteglueing} gives sufficient conditions on a set of gadgets to have its closure by gadget glueing contained in a CSAN family.

\begin{lemma}\label{lem:csan-gadgets-glueing-closure}
  Fix some alphabet $Q$ and some glueing interface ${C=C_i\cup C_o}$ and some CSAN family $\mathcal{F}$.
  Let ${(G_n,\lambda_n,\rho_n)}$ for ${n\in S}$ be a set of CSAN belonging to $\mathcal{F}$ with associated global maps $F_n$. Let ${\phi_{F_n,k}^i}$ for ${k\in I(F_n)}$ and ${\phi_{F_n,k}^o}$ for ${k\in O(F_n)}$ be maps as in Definition~\ref{def:gadgets-glueing} so that ${(F_n,(\phi_{F_n,k}^i),(\phi_{F_n,k}^o))}$ is a gadget with interface ${C=C_i\cup C_o}$. Denote by $X$ the set of such gadgets. If the following conditions hold:
  \begin{itemize}
  \item the labeled graphs induced by ${\phi_{F_n,k}^i(C)}$ and by ${\phi_{F_n,k}^o(C)}$ in $G_n$ are all the same for all $n$ and $k$ with the identification of vertices given by the $\phi_{\ast,\ast}^\ast$ maps,
  \item ${N_{G_n}(\phi_{F_n,k}^i(C_o))\subseteq \phi_{F_n,k}^i(C)}$ for all ${n\in S}$ and all ${k\in I(F_n)}$,
  \item ${N_{G_n}(\phi_{F_n,k}^o(C_i))\subseteq \phi_{F_n,k}^o(C)}$ for all ${n\in S}$ and all ${k\in O(F_n)}$,
  \end{itemize}
  then the closure by gadget glueing of $X$ is included in $\mathcal{F}$.
\end{lemma}
\begin{proof}
  Consider first the gadget glueing $H$ of two gadgets $F_{n}$ and $F_{n'}$ from $X$. Following Definition~\ref{def:gadgets-glueing}, the global dowel ${C_{F_n}\cup C_{F_{n'}}}$ used in such a glueing is a disjoint union of copies of $C$, and its embedding $\phi_{F_n}$ in $G_n$ (resp. $\phi_{F_{n'}}$ in $G_{n'}$) is a disjoint union of maps $\phi_{F_n,\ast}^\ast$ (resp. ${\phi_{F_{n'},\ast}^\ast}$). Therefore the three conditions of Lemma~\ref{lem:concreteglueing} follow from the three conditions of the hypothesis on gadgets from $X$ and we deduce that $H$ belongs to family $\mathcal{F}$. Moreover, it is clear that gadget $H$ then also verifies the three conditions from the hypothesis, and adding a copy of any gadget to the set also verifies the conditions. We deduce that the closure by gadget glueing of $X$ is included in ${\mathcal{F}}$.
   
\end{proof}

The second key aspect to have a coherent set $X$ of $\GG$-gadgets is of dynamical nature: there must exists a collection of pseudo-orbits on each gadget satisfying suitable conditions to permit application of Lemma~\ref{lem:pseudo-orbit-glueing} for any gadget glueing in the closure of $X$; moreover, these pseudo-orbits must simulate via an appropriate coding the input/output relations of each gate ${g\in\GG}$ in the corresponding gadget. To obtain this, we rely on a standard set of traces on the glueing interface that must be respected on any copy of it in any gadget.

\begin{definition}[Coherent $\GG$-gadgets]\label{def:coherent-gadgets}
  Let $\GG$ be any set of finite maps over alphabet $Q$ and let ${\mathcal F}$ be any set of abstract automata networks over alphabet $\QF$.
  We say  ${\mathcal F}$ has \emph{coherent $\GG$-gadgets} if there exists:
    \begin{itemize}
    \item a unique glueing interface ${C=C_i\cup C_o}$, 
    \item a set $X$ of gadgets ${(F_g,(\phi_{g,k}^i)_{1\leq k\leq i(g)},(\phi_{g,k}^o)_{1\leq k\leq o(g)})}$ for each $g\in\GG$ where ${F_g:\QF^{V_g}\rightarrow\QF^{V_g}\in\mathcal{F}}$ and sets $V_g$ and $C$ are pairwise disjoint, and the closure of $X$ by gadget glueing is contained in $\mathcal F$,
    \item a \emph{state configuration} ${s_q\in \QF^{C}}$ for each ${q\in Q}$ such that ${q\mapsto s_q}$ is an injective map,
    \item a \emph{context configuration} ${c_g\in \QF^{{\hat V}_g}}$ for each ${g\in\GG}$ where ${{\hat V}_g=V_g\setminus\bigl(\cup_k\phi_{g,k}^i(C)\cup_k \phi_{g,k}^o(C)\bigr)}$,
    \item a time constant $T$,
    \item a \emph{standard trace} $\tau_{q,q'}\in (\QF^C)^{\{0,\ldots,T\}}$ for each pair ${q,q'\in Q}$ such that ${\tau_{q,q'}(0)=s_{q}}$ and ${\tau_{q,q'}(T)=s_{q'}}$,
    \item for each ${g\in\GG}$ and for any uples of states ${q_{i,1},\ldots,q_{i,i(g)}\in Q}$ and ${q_{o,1},\ldots,q_{o,o(g)}\in Q}$ and ${q_{i,1}',\ldots,q_{i,i(g)}'\in Q}$ and ${q_{o,1}',\ldots,q_{o,o(g)}'\in Q}$ such that ${g(q_{i,1},\ldots,q_{i,i(g)}) = (q'_{o,1},\ldots,q'_{o,o(g)})}$, a $P_g$-pseudo-orbit ${(x^t)_{0\leq t\leq T}}$ of $F_g$ with ${P_g = \bigcup_{1\leq k\leq i(g)}\phi_{g,k}^i(C_o)\cup \bigcup_{1\leq k\leq o(g)}\phi_{g,k}^o(C_i)}$ and with
      \begin{itemize}
      \item for each ${1\leq k\leq i(g)}$, the trace ${t\mapsto x^t_{\phi_{g,k}^i(C)}}$ is exactly $\tau_{q_{i,k},q_{i,k}'}$,
      \item for each ${1\leq k\leq o(g)}$, the trace ${t\mapsto x^t_{\phi_{g,k}^o(C)}}$ is exactly $\tau_{q_{o,k},q_{o,k}'}$,
      \item ${x^0_{{\hat V}_g} = x^T_{{\hat V}_g} = c_g}$.
      \end{itemize}
    \end{itemize}
\end{definition}

We can now state the key lemma of our framework: having coherent $\GG$-gadgets is sufficient to simulate the whole family of $\GG$-networks.

\begin{lemma}\label{lem:from-gadgets-to-networks}
  Let $\GG$ be a set of irreducible gates.
  If an abstract automata network family $\mathcal{F}$ has coherent $\GG$-gadgets then it contains a subfamily of bounded degree networks with the canonical bounded degree representation ${(\mathcal{F}_0,\mathcal{F}_0^*)}$ that simulates $\Gamma(\mathcal{G})$ in time $T$ and space $S$ where $T$ is a constant map and $S$ is bounded by a linear map.
\end{lemma}
\begin{proof}
  We take the notations of Definition~\ref{def:coherent-gadgets}.
  To any $\GG$-network $F$ with set of nodes $V$ given as in Definition~\ref{def:g-network} by a list of gates ${g_1,\ldots,g_k\in\GG}$ and maps $\alpha$ and $\beta$ (see Remark~\ref{rem:represent-g-networks}) we associate an automata network from $\mathcal{F}$ as follows. First, let $(F_{g_i})_{1\leq i\leq k}$ be the gadgets corresponding to gates $g_i$ and suppose they are all disjoint (by taking disjoint copies when necessary). Then, start from the gadget $F_1=F_{g_1}$ and for any ${1\leq i< k}$ we define ${F_{i+1}}$ as the gadget glueing of $F_i$ and $F_{g_{i+1}}$ on the input/outputs as prescribed by maps $\alpha$ and $\beta$. More precisely, the gadget glueing select the set of inputs ${(j,k)}$ with ${1\leq j\leq i}$ and ${1\leq k\leq i(g_j)}$ such that ${\beta(\alpha(j,k))=(i+1,k')}$ for some ${1\leq k'\leq o(g_{i+1})}$ and glue them on their corresponding output ${(i+1,k')}$ of $g_{i+1}$ (precisely, through maps $\sigma_{F_i}$ and $\sigma_{F_{g_{i+1}}}$ of domain $A_{i+1}$ playing the role of maps $\sigma_F$ and $\sigma_G$ of Definition~\ref{def:gadgets-glueing}), and, symmetrically, selects the inputs ${(i+1,k)}$ with ${1\leq k\leq i(g_{i+1})}$ such that ${\beta(\alpha(i+1,k))=(j,k')}$ for some ${1\leq j\leq i}$ and ${1\leq k'\leq o(g_{j})}$ and glue their corresponding output ${(j,k')}$ (precisely, through maps $\tau_{F_{g_{i+1}}}$ and $\tau_{F_i}$ of domain $B_{i+1}$ playing the role of maps $\tau_G$ and $\tau_F$ from Definition~\ref{def:gadgets-glueing}). If both of these sets of inputs/outputs are empty, the gadget glueing is replaced by a simple disjoint union.

  The final gadget $F_k$ has no input and no output, and a representation of it as a pair graph and local maps can be constructed in \DLOG{}, because the local map of each of its nodes is independent of the glueing sequence above and completely determined by the gadget $F_{g_j}$ it belongs to and whether the node is inside some input or some output dowel or not (see Reamrk~\ref{rem:gadgets-glueing}).

  It now remains to show that the automata network $F_k$ simulates $F$. To fix notations, let $V_k$ be the set of nodes of $F_k$. For each ${v\in V}$, define ${D_v\subseteq V_k}$ as the copy of the dowel that correspond to node $v$ of $F$, \textit{i.e.} that was produced in the gadget glueing of ${F_i}$ with $F_{g_{i+1}}$ for $i$ such that ${\beta(v)=(i+1,k')}$ for some ${1\leq k'\leq o(g_{i+1})}$ (or symmetrically $\alpha(i+1,k)=v$ for some ${1\leq k\leq i(g_j)}$). More precisely, if $a\in A_{i+1}$ is such that ${\sigma_{F_{g_{i+1}}}(a)=(i+1,k')}$ then $D_v = \{a\}\times C$ (symmetrically if $b\in B_{i+1}$ is such that ${\tau_{F_{g_{i+1}}}(b)=(i+1,k)}$ then $D_v = \{b\}\times C$). Also denote by $\rho_v : D_v\to C$ the map such that ${\rho_v(a,c)=c}$ for all ${c\in C}$ (symmetrically, ${\rho_v(b,c)=c}$). With these notations, we have 
  \[V_k = \bigcup_{v\in V} D_v\cup\bigcup_{1\leq i\leq k}\hat{V}_{g_i}\]
  Let us define the block embedding $\phi : Q^V\to Q_{\mathcal{F}}^{V_k}$ as follows 
  \[\phi(x)(v') =
    \begin{cases}
      s_{x_v}(\rho_v(v'))&\text{ if }v'\in D_v,\\
      c_{g_i}(v')&\text{ if }v'\in\hat{V}_{g_i}.\\
    \end{cases}
  \]
  for any ${x\in Q^V}$ and any ${v'\in V_k}$, where $s_q$ for ${q\in Q}$ are the state configurations and $c_g$ for ${g\in\GG}$ are the context configurations granted by Definition~\ref{def:coherent-gadgets}. Note that $\phi$ is injective because the map ${q\mapsto s_q}$ is injective. By inductive applications of Lemma~\ref{lem:pseudo-orbit-glueing}, the $P_{g_i}$-pseudo-orbits of each ${F_{g_i}}$ from Definition~\ref{def:coherent-gadgets} can be glued together to form valid orbits of $F_k$ that start from any configuration ${\phi(x)}$ with ${x\in Q^V}$ and ends after $T$ steps in a configuration ${\phi(y)}$ for some ${y\in Q^V}$ which verifies ${y=F(x)}$. Said differently, we have the following equality on $Q^V$: 
  \[\phi\circ F = F_k^T\circ\phi.\]
  Note that $T$ is a constant and that the size of $V_k$ is at most linear in the size of $V$. The lemma follows. 
\end{proof}

\begin{remark}
  Note that in Lemma~\ref{lem:from-gadgets-to-networks} above, the block embedding that is constructed can be viewed as a collection of blocks of bounded size that encode all the information plus a context (see Remark~\ref{rem:bloc-simul}).
\end{remark}

In the case of CSAN families and using Remark~\ref{rem:csan-bounded-degree} we have a simpler formulation of the Lemma.

\begin{corollary}\label{cor:gadgets-mon-csan}
  If $\GG$ is a set of irreducible gates and $\mathcal{F}$ a CSAN family which has coherent $\GG$-gadgets then $\mathcal{F}$ simulates $\Gamma(\mathcal{G})$ in time $T$ and space $S$ where $T$ is a constant map and $S$ is bounded by a linear map.  
\end{corollary}

\subsubsection{Game of life has coherent gadgets}

\newcommand\Gnor{\GG_{\text{NOR}}}

In this section, we are going to show that \emph{Game of life} has coherent $\Gnor$ gadgets, where $\Gnor = \{\text{NOR}(x,y) = (\overline{x\vee y},\overline{x\vee y})\}.$ In order to do that, we are going to show that we are able to simulate this two gates by combining wire networks and a clock network. 

First, we show, in Figure \ref{fig:NORgadg}, the structure of the communication graph of the NOR gadget. Observe that it is composed by $2$ copies of the wire network and $2$ copies of the clock network (see Figure \ref{fig:gadgetsym}). Now, we present the main result of this subsection:
\begin{lemma}
Game of Life automata networks admits coherent $\Gnor$ gadgets. \label{lemma:GOCcoherent}
\end{lemma}
\begin{proof}
We are going to show that the NOR gadget satisfies the conditions of the Definition~\ref{def:coherent-gadgets}. In fact we have that:
\begin{itemize}
\item The NOR gadget has a unique glueing interface $C$ which is shown in Figure \ref{fig:NORgadg}. The functions $\phi^i, \phi^o$ are also represented in the same figure.
\item Observe that the map $q\in \{0,1\} \mapsto s_q = (w^q_0)|_{C}$ (observe that $w^i_0 = w_0$ for each $i$ since all the nodes are in state $0$ in this configuration), is injective. 
\item The context configuration is given by:
\begin{itemize}
\item For the clock network we use the initial condition of the dynamics of the clock network (see Figure \ref{fig:clockgl}). 
\item For the copies of the wire network: the nodes are in state $0$ with the exception of the nodes in the copies of $C$ (see Figure \ref{fig:wiregl}).
\end{itemize} 
\item The time constant is $T= 6$
\item The configurations $w^0_q, \hdots, w^5_q$ (see Figure \ref{fig:wiregl}) and $w^t_0$ (or simply $w_0$) where $w^t_0$ is the configuration in which each node is in state $0$ for every $t = 0,\hdots,T$ define a standard trace for each pair $(0,0),(0,1),(1,0),(1,1)$ as follows: $\tau_{q,q'} = ((w^0_q)|_C,(w^1_q)|_C,(w^3_q)|_C,(w^4_q)|_C,(w^5_{q'})|_C,(w^6_{q')}|_C = (w^0_{q')}|_C  ))$. 
\item The pseudo orbits are shown in Table \ref{table:pseudoNOR} and Figure \ref{fig:gadgetsym}. In the table, the detail of the local computation produced by the central part of the gadget is given. Figure \ref{fig:gadgetsym} shows how the nodes are labeled in the previous table. In addition, the latter figure shows a general picture on how the signals are transmitted and computed by the gadget. In particular, it is posible to verify that all the computation is produced in $T= 6$.
\item After $T=6$ time steps, the nodes that are not part of any copy of $C$ will return to the state given by the context configuration.
\end{itemize}
We conclude that the NOR gadget satisfies the conditions given by Definition~\ref{def:coherent-gadgets}. Thus, the lemma holds.
\end{proof}

\begin{table}[]
\tiny
\centering
\begin{tabular}{|l|l|l|l|l|l|l|l|l|l|l|l|l|l|l|l|l|l|}
\hline
Nodes/Time & $l_1$ & $l_2$ & $l_3$ & $l'_1$ & $l'_2$ & $l'_3$ & $c_1$ & $c_2$ & $c_3$ & $a$                  & $v$                  & $r_1$                & $r_2$                & $r_3$                & $r'_1$               & $r'_2$               & $r'_3$               \\ \hline
0          & 0     & 0     & 0     & 0      & 0      & 0      & 0     & 0     & 0     & 0                    & 0                    & 0                    & 0                    & 0                    & 0                    & 0                    & 0                    \\ \hline
1          & 0     & 0     & 0     & 0      & 0      & 0      & 0     & 0     & 0     & 0                    & 0                    & 0                    & 0                    & 0                    & 0                    & 0                    & 0                    \\ \hline
2          & $x$   & $x$   & $x$   & $y$    & $y$    & $y$    & 1     & 1     & 1     & 0                    & 0                    & 0                    & 0                    & 0                    & 0                    & 0                    & 0                    \\ \hline
3          & $x$   & $x$   & $x$   & $y$    & $y$    & $y$    & 1     & 1     & 1     & 0                    & $\overline{x\vee y}$ & 0                    & 0                    & 0                    & 0                    & 0                    & 0                    \\ \hline
4          & 0     & 0     & 0     & 0      & 0      & 0      & 0     & 0     & 0     & $\overline{x\vee y}$ & $\overline{x\vee y}$ & $\overline{x\vee y}$ & $\overline{x\vee y}$ & $\overline{x\vee y}$ & $\overline{x\vee y}$ & $\overline{x\vee y}$ & $\overline{x\vee y}$ \\ \hline
5          & 0     & 0     & 0     & 0      & 0      & 0      & 0     & 0     & 0     & $\overline{x\vee y}$ & 0                    & $\overline{x\vee y}$ & $\overline{x\vee y}$ & $\overline{x\vee y}$ & $\overline{x\vee y}$ & $\overline{x\vee y}$ & $\overline{x\vee y}$ \\
\hline
\end{tabular}
\caption{Pseudo-orbit associated to the NOR gadget implemented using Game of Life rules. The labelling is defined according to Figure \ref{fig:gadconnect}}
\label{table:pseudoNOR}
\end{table}

\begin{figure}
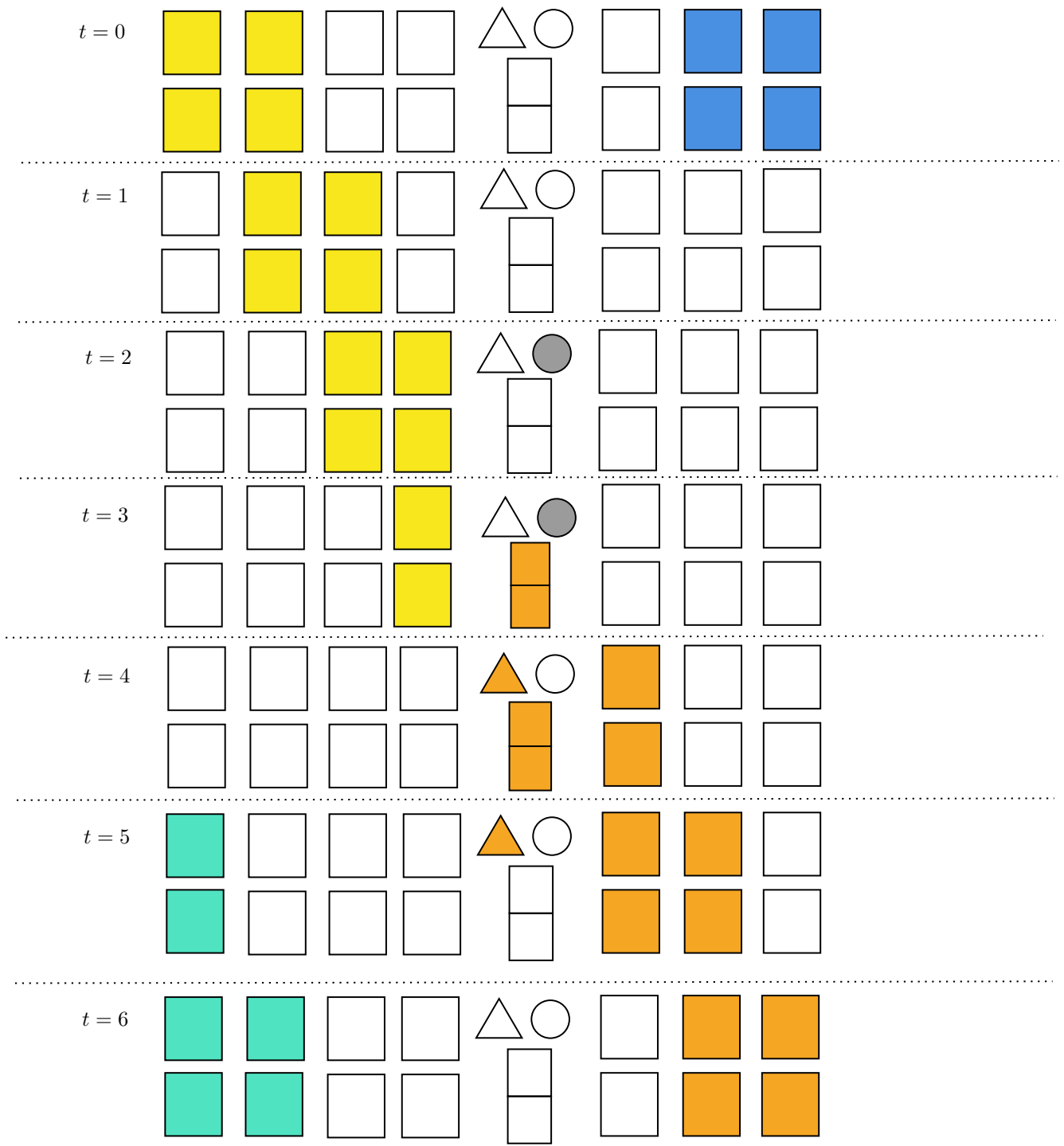

\centering

\tikzset{every picture/.style={line width=0.75pt}} 

%

\caption{Representation of the pseudo-orbits used for the computation of the NOR gadget implemented using Game of Life automata networks.}
\label{fig:gadgetsym}
\end{figure}
\begin{figure}
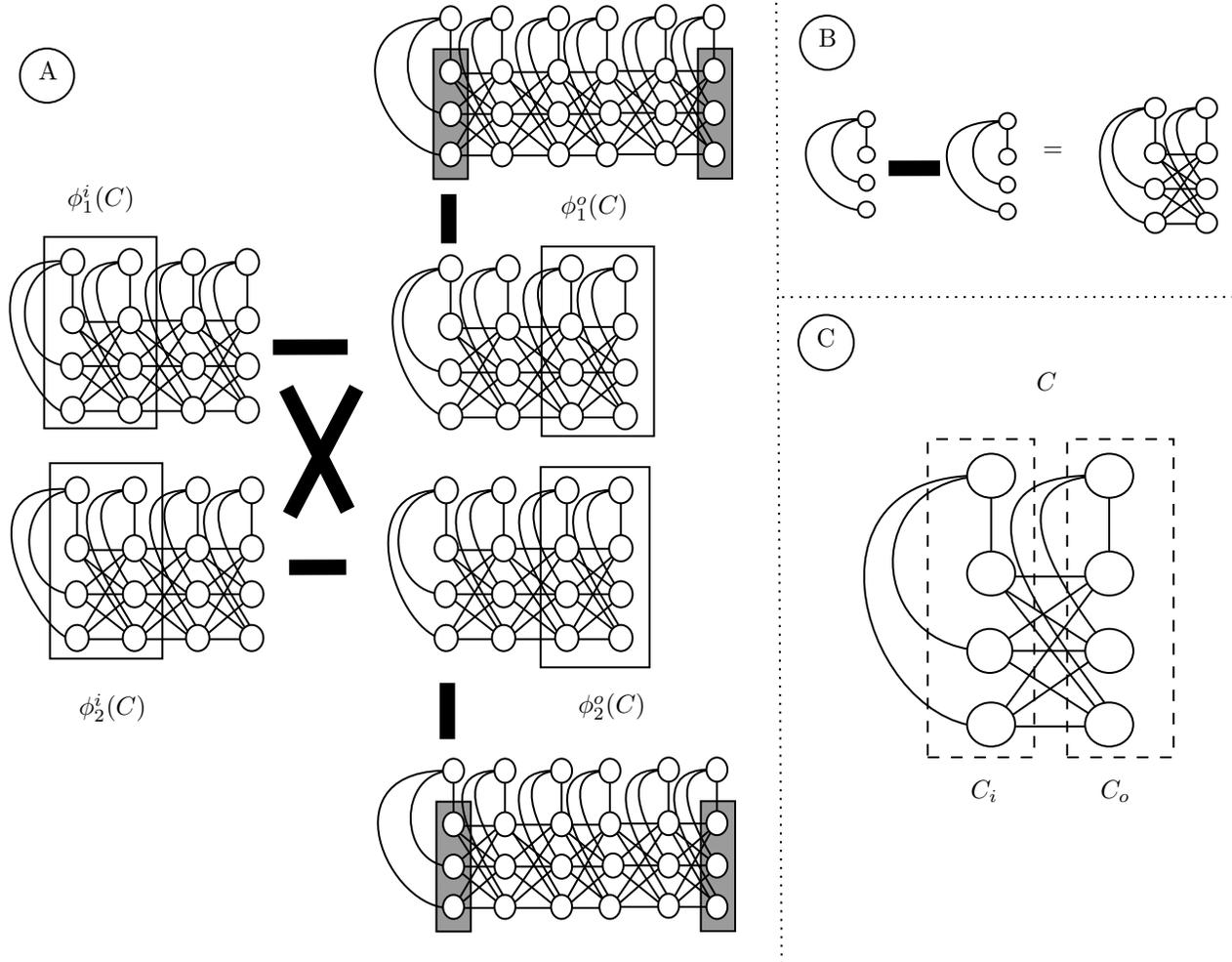

\centering

\tikzset{every picture/.style={line width=0.75pt}} 

%

\caption{(A) Representation of the computation part of the NOR gadget implemented using Game of Life rules. (B) Thick lines represent complete connections between the different layers, i.e. each node in the central layer is connected with any node in the left layer. (C) Representation of the glueing interface. } 
\label{fig:NORgadg}
\end{figure}

\begin{figure}
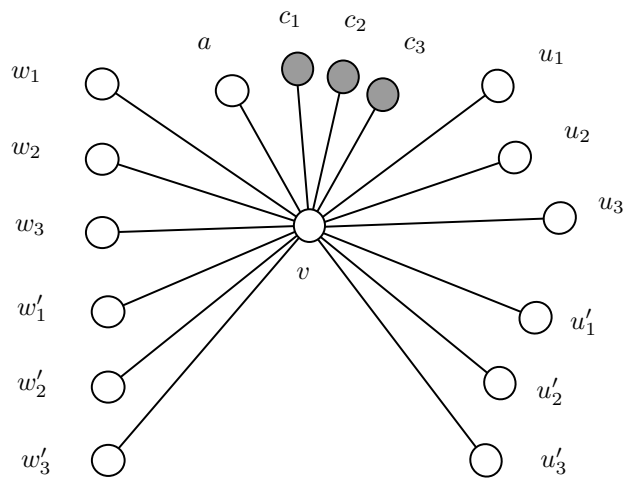


\centering

\tikzset{every picture/.style={line width=0.75pt}} 

%

\caption{Structure of the connections of the central nodes in the NOR gadget implemented using Game of Life rules. Nodes labeled as $l$ (resp. $r$) are nodes inside the last part (resp. first) of a wire network, nodes labeled as $c$ are nodes inside one layer of a clock network (see Figure \ref{fig:wiregl}).}
\label{fig:gadconnect}
\end{figure}

\newcommand\Gmontwo{\GG_{m,2}}
\subsection{$\Gmon$-networks and $\Gmontwo$-networks as standard universal families}
\label{sec:gmonuniv}
\newcommand\gateOR{\mathrm{OR}}
\newcommand\gateAND{\mathrm{AND}}
\newcommand\gateCOPY{\mathrm{COPY}}
Let $i,o \in \{1,2\}$ be two numbers. We define the functions $\gateOR, \gateAND: \{0,1\}^i \to \{0,1\}^o$ where $\gateOR(x) = \max(x)$ and $\gateAND(x) = \min(x)$. Note that in the case in which $i = o = 1$ we have $\gateAND(x) = \gateOR(x) = \text{Id}(x) = x$ and also in the case $i = 1$ and $o=2$ we have that $\gateAND(x)=\gateOR(x) = (x, x).$ We define the set $\Gmon = \{\gateAND, \gateOR\}.$ Observe that in this case $o$ and $i$ may take different values. In addition, we define the set $\Gmontwo$ in which we fix $i=o=2.$

It is folklore knowledge that monotone Boolean networks (with AND/OR local maps) can simulate any other network.
Here we make this statement precise within our formalism: $\Gmon$-networks are strongly universal.
Note that there is more work than the classical circuit transformations involving monotone gates because we need to obtain a simulation of any automata network via block embedding. In particular we need to build monotone circuitry that is synchronized and reusable (\textit{i.e.} that can be reinitialized to a standard state before starting a computation on a new input). Moreover, our definitions requires a production of $\Gmon$-networks in \DLOG{}.
The main ingredient for establishing universality of $\Gmon$-networks is an efficient circuit transformation due to Greenlaw, Hoover and Ruzzo in \cite[Theorems 6.2.3 to 6.2.5]{Greenlaw_1995}.
Let us start by proving that this family is strongly universal, which is slightly simpler to prove.

\begin{theorem}\label{theo:gmon-univ}
  The family $\Gamma(\mathcal{G}_m)$ of all $\Gmon$-networks is strongly universal.
\end{theorem}
\begin{proof}
Let $Q$ an arbitrary alphabet and $F:Q^{n} \to Q^{n}$ an arbitrary automata network on alphabet $Q$ such that the communication graph of $F$ has maximum degree $\Delta.$  Let $C:\{0,1\}^{n} \to \{0,1\}^{n}$ be a constant depth circuit representing $F$. Let us assume that $C$  has only  OR, AND and NOT gates. We can also assume that $C$ is synchronous because, as its depth does not depend on the size of the circuit, one can always add fanin one and fanout one OR gates in order to modify layer structure. We are going to use a very similar transformation to the one proposed in \cite[Theorem 6.2.3]{Greenlaw_1995} in order to efficiently construct an automata network in $\Gamma(\mathcal{G}_m)$. In fact, we are going to duplicate the original circuit by considering the coding $x \in \{0,1\} \to (x, 1-x) \in \{0,1\}^{2}.$ Roughly, each gate will have a positive part (which is essentially a copy) and a negative part which is produces the negation of the original output by using De Morgan's laws.  More precisely, we are going to replace each gate in the network by the gadgets shown in Figure \ref{fig:ANDORgatesgad}.  The main idea is that one can represent the function $x \wedge y$ by the coding: $(x \wedge y, \overline{x} \vee \overline{y})$ and $x \vee y$ by the coding: $(x \vee y, \overline{x} \wedge \overline{y}).$ In addition, each time there is a NOT gate, we replace it by a fan in $1$ fan out $1$ OR gadget and we connect positive outputs to negative inputs in the next layer and negative outputs to positive inputs as it is shown in Figure \ref{fig:NOTgatesgad}.  
We are going to call $C^*$ to the circuit constructed by latter transformations. Observe that $C^*$ is such that it holds on $\{0,1\}^i$: 
  \[\phi\circ C = C^*\circ\phi\]
  where $\phi : \{0,1\}^n\to\{0,1\}^{2n}$ is defined for any $n$ by ${\phi(x)_{2j}=x_j}$ and ${\phi(x)_{2j+1}=\neg x_j}$.

  Now consider the coding map ${m_Q:Q\to\{0,1\}^k}$ and let ${n=k|V|}$. Build from $C^*$ the $\Gmontwo$-network $F^*:\{0,1\}^{V^{+}} \to \{0,1\}^{V^{+}}$ that correspond to it (gate by gate) and where the output $j$ is wired to input $j$ for all ${1\leq j\leq 2n}$. Define a block embedding of $Q^V$ into ${\{0,1\}^{V^+}}$ as follows (see Remark~\ref{rem:bloc-simul}):
  \begin{itemize}
  \item for each $v\in V$ let $D_v$ be the set of input nodes in $F^*$ that code $v$ (via $m_Q$ and then double railed logic),
  \item let $C=V^+\setminus \bigcup_v D_v$ be the remaining context block,
  \item let ${p_{v,q}\in \{0,1\}^{D_v}}$ be the pattern coding node $v$ in state $q$,
  \item let $p_C = 0^C$ be the context pattern,
  \item let $\phi: Q^V\to \{0,1\}^{V^+}$ be the associated block embedding map.
  \end{itemize}
  We claim that $F^*$ simulates $F$ via block embedding $\phi$ with time constant equal to the depth of $C^*$ plus $1$. Indeed, $F^* $ can be seen as a directed cycle of $N$ layers where layer $L_{i+1\bmod N}$ only depends on layer $i$. The block embedding is such that for any configuration ${x\in Q^V}$, ${\phi(x)}$ is $0$ on each layer except the layer containing the inputs. On configurations where a single layer $L_i$ is non-zero, $F^*$ will produce a configuration where the only non-zero layer is $L_{i+1\bmod N}$. From there, it follows by construction of $F^*$ that ${\phi\circ F (x) = (F^*)^N\circ\phi(x)}$ for all ${x\in Q^V}$.

The fact that construction is obtainable in $\DLOG$ follows from the same reasoning used to show in \cite[Theorem 6.2.3]{Greenlaw_1995}. In fact, authors show that reduction is actually better as they show it is $\NC^{1}.$  
\end{proof}
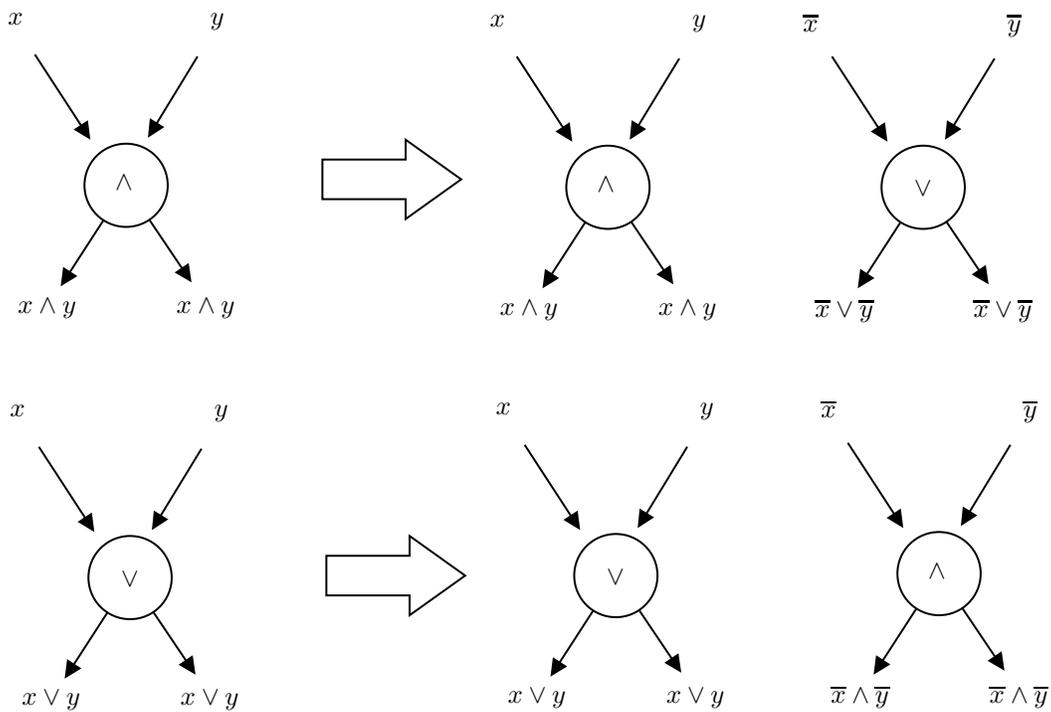
\begin{figure}

\centering

\tikzset{every picture/.style={line width=0.75pt}} 

\begin{tikzpicture}[x=0.75pt,y=0.75pt,yscale=-1,xscale=1]


\draw    (38,34) -- (64.36,74.49) ;
\draw [shift={(66,77)}, rotate = 236.93] [fill={rgb, 255:red, 0; green, 0; blue, 0 }  ][line width=0.08]  [draw opacity=0] (8.93,-4.29) -- (0,0) -- (8.93,4.29) -- cycle    ;
\draw    (120,35) -- (96.56,73.44) ;
\draw [shift={(95,76)}, rotate = 301.37] [fill={rgb, 255:red, 0; green, 0; blue, 0 }  ][line width=0.08]  [draw opacity=0] (8.93,-4.29) -- (0,0) -- (8.93,4.29) -- cycle    ;
\draw    (84,100) -- (115.32,146.51) ;
\draw [shift={(117,149)}, rotate = 236.04] [fill={rgb, 255:red, 0; green, 0; blue, 0 }  ][line width=0.08]  [draw opacity=0] (8.93,-4.29) -- (0,0) -- (8.93,4.29) -- cycle    ;
\draw    (84,100) -- (52.63,148.48) ;
\draw [shift={(51,151)}, rotate = 302.90999999999997] [fill={rgb, 255:red, 0; green, 0; blue, 0 }  ][line width=0.08]  [draw opacity=0] (8.93,-4.29) -- (0,0) -- (8.93,4.29) -- cycle    ;
\draw  [fill={rgb, 255:red, 255; green, 255; blue, 255 }  ,fill opacity=1 ] (63,100) .. controls (63,88.4) and (72.4,79) .. (84,79) .. controls (95.6,79) and (105,88.4) .. (105,100) .. controls (105,111.6) and (95.6,121) .. (84,121) .. controls (72.4,121) and (63,111.6) .. (63,100) -- cycle ;

\draw    (40,232) -- (66.36,272.49) ;
\draw [shift={(68,275)}, rotate = 236.93] [fill={rgb, 255:red, 0; green, 0; blue, 0 }  ][line width=0.08]  [draw opacity=0] (8.93,-4.29) -- (0,0) -- (8.93,4.29) -- cycle    ;
\draw    (122,233) -- (98.56,271.44) ;
\draw [shift={(97,274)}, rotate = 301.37] [fill={rgb, 255:red, 0; green, 0; blue, 0 }  ][line width=0.08]  [draw opacity=0] (8.93,-4.29) -- (0,0) -- (8.93,4.29) -- cycle    ;
\draw    (86,298) -- (117.32,344.51) ;
\draw [shift={(119,347)}, rotate = 236.04] [fill={rgb, 255:red, 0; green, 0; blue, 0 }  ][line width=0.08]  [draw opacity=0] (8.93,-4.29) -- (0,0) -- (8.93,4.29) -- cycle    ;
\draw    (86,298) -- (54.63,346.48) ;
\draw [shift={(53,349)}, rotate = 302.90999999999997] [fill={rgb, 255:red, 0; green, 0; blue, 0 }  ][line width=0.08]  [draw opacity=0] (8.93,-4.29) -- (0,0) -- (8.93,4.29) -- cycle    ;
\draw  [fill={rgb, 255:red, 255; green, 255; blue, 255 }  ,fill opacity=1 ] (65,298) .. controls (65,286.4) and (74.4,277) .. (86,277) .. controls (97.6,277) and (107,286.4) .. (107,298) .. controls (107,309.6) and (97.6,319) .. (86,319) .. controls (74.4,319) and (65,309.6) .. (65,298) -- cycle ;
\draw   (183,87) -- (225,87) -- (225,77) -- (253,97) -- (225,117) -- (225,107) -- (183,107) -- cycle ;
\draw    (281,35) -- (307.36,75.49) ;
\draw [shift={(309,78)}, rotate = 236.93] [fill={rgb, 255:red, 0; green, 0; blue, 0 }  ][line width=0.08]  [draw opacity=0] (8.93,-4.29) -- (0,0) -- (8.93,4.29) -- cycle    ;
\draw    (363,36) -- (339.56,74.44) ;
\draw [shift={(338,77)}, rotate = 301.37] [fill={rgb, 255:red, 0; green, 0; blue, 0 }  ][line width=0.08]  [draw opacity=0] (8.93,-4.29) -- (0,0) -- (8.93,4.29) -- cycle    ;
\draw    (327,101) -- (358.32,147.51) ;
\draw [shift={(360,150)}, rotate = 236.04] [fill={rgb, 255:red, 0; green, 0; blue, 0 }  ][line width=0.08]  [draw opacity=0] (8.93,-4.29) -- (0,0) -- (8.93,4.29) -- cycle    ;
\draw    (327,101) -- (295.63,149.48) ;
\draw [shift={(294,152)}, rotate = 302.90999999999997] [fill={rgb, 255:red, 0; green, 0; blue, 0 }  ][line width=0.08]  [draw opacity=0] (8.93,-4.29) -- (0,0) -- (8.93,4.29) -- cycle    ;
\draw  [fill={rgb, 255:red, 255; green, 255; blue, 255 }  ,fill opacity=1 ] (306,101) .. controls (306,89.4) and (315.4,80) .. (327,80) .. controls (338.6,80) and (348,89.4) .. (348,101) .. controls (348,112.6) and (338.6,122) .. (327,122) .. controls (315.4,122) and (306,112.6) .. (306,101) -- cycle ;

\draw    (440,35) -- (466.36,75.49) ;
\draw [shift={(468,78)}, rotate = 236.93] [fill={rgb, 255:red, 0; green, 0; blue, 0 }  ][line width=0.08]  [draw opacity=0] (8.93,-4.29) -- (0,0) -- (8.93,4.29) -- cycle    ;
\draw    (522,36) -- (498.56,74.44) ;
\draw [shift={(497,77)}, rotate = 301.37] [fill={rgb, 255:red, 0; green, 0; blue, 0 }  ][line width=0.08]  [draw opacity=0] (8.93,-4.29) -- (0,0) -- (8.93,4.29) -- cycle    ;
\draw    (486,101) -- (517.32,147.51) ;
\draw [shift={(519,150)}, rotate = 236.04] [fill={rgb, 255:red, 0; green, 0; blue, 0 }  ][line width=0.08]  [draw opacity=0] (8.93,-4.29) -- (0,0) -- (8.93,4.29) -- cycle    ;
\draw    (486,101) -- (454.63,149.48) ;
\draw [shift={(453,152)}, rotate = 302.90999999999997] [fill={rgb, 255:red, 0; green, 0; blue, 0 }  ][line width=0.08]  [draw opacity=0] (8.93,-4.29) -- (0,0) -- (8.93,4.29) -- cycle    ;
\draw  [fill={rgb, 255:red, 255; green, 255; blue, 255 }  ,fill opacity=1 ] (465,101) .. controls (465,89.4) and (474.4,80) .. (486,80) .. controls (497.6,80) and (507,89.4) .. (507,101) .. controls (507,112.6) and (497.6,122) .. (486,122) .. controls (474.4,122) and (465,112.6) .. (465,101) -- cycle ;

\draw    (285,231) -- (311.36,271.49) ;
\draw [shift={(313,274)}, rotate = 236.93] [fill={rgb, 255:red, 0; green, 0; blue, 0 }  ][line width=0.08]  [draw opacity=0] (8.93,-4.29) -- (0,0) -- (8.93,4.29) -- cycle    ;
\draw    (367,232) -- (343.56,270.44) ;
\draw [shift={(342,273)}, rotate = 301.37] [fill={rgb, 255:red, 0; green, 0; blue, 0 }  ][line width=0.08]  [draw opacity=0] (8.93,-4.29) -- (0,0) -- (8.93,4.29) -- cycle    ;
\draw    (331,297) -- (362.32,343.51) ;
\draw [shift={(364,346)}, rotate = 236.04] [fill={rgb, 255:red, 0; green, 0; blue, 0 }  ][line width=0.08]  [draw opacity=0] (8.93,-4.29) -- (0,0) -- (8.93,4.29) -- cycle    ;
\draw    (331,297) -- (299.63,345.48) ;
\draw [shift={(298,348)}, rotate = 302.90999999999997] [fill={rgb, 255:red, 0; green, 0; blue, 0 }  ][line width=0.08]  [draw opacity=0] (8.93,-4.29) -- (0,0) -- (8.93,4.29) -- cycle    ;
\draw  [fill={rgb, 255:red, 255; green, 255; blue, 255 }  ,fill opacity=1 ] (310,297) .. controls (310,285.4) and (319.4,276) .. (331,276) .. controls (342.6,276) and (352,285.4) .. (352,297) .. controls (352,308.6) and (342.6,318) .. (331,318) .. controls (319.4,318) and (310,308.6) .. (310,297) -- cycle ;

\draw    (448,230) -- (474.36,270.49) ;
\draw [shift={(476,273)}, rotate = 236.93] [fill={rgb, 255:red, 0; green, 0; blue, 0 }  ][line width=0.08]  [draw opacity=0] (8.93,-4.29) -- (0,0) -- (8.93,4.29) -- cycle    ;
\draw    (530,231) -- (506.56,269.44) ;
\draw [shift={(505,272)}, rotate = 301.37] [fill={rgb, 255:red, 0; green, 0; blue, 0 }  ][line width=0.08]  [draw opacity=0] (8.93,-4.29) -- (0,0) -- (8.93,4.29) -- cycle    ;
\draw    (494,296) -- (525.32,342.51) ;
\draw [shift={(527,345)}, rotate = 236.04] [fill={rgb, 255:red, 0; green, 0; blue, 0 }  ][line width=0.08]  [draw opacity=0] (8.93,-4.29) -- (0,0) -- (8.93,4.29) -- cycle    ;
\draw    (494,296) -- (462.63,344.48) ;
\draw [shift={(461,347)}, rotate = 302.90999999999997] [fill={rgb, 255:red, 0; green, 0; blue, 0 }  ][line width=0.08]  [draw opacity=0] (8.93,-4.29) -- (0,0) -- (8.93,4.29) -- cycle    ;
\draw  [fill={rgb, 255:red, 255; green, 255; blue, 255 }  ,fill opacity=1 ] (473,296) .. controls (473,284.4) and (482.4,275) .. (494,275) .. controls (505.6,275) and (515,284.4) .. (515,296) .. controls (515,307.6) and (505.6,317) .. (494,317) .. controls (482.4,317) and (473,307.6) .. (473,296) -- cycle ;

\draw   (185,287) -- (227,287) -- (227,277) -- (255,297) -- (227,317) -- (227,307) -- (185,307) -- cycle ;

\draw (72+5,88.4+5) node [anchor=north west][inner sep=0.75pt]    {$\land $};
\draw (23,11.4) node [anchor=north west][inner sep=0.75pt]    {$x$};
\draw (125,11.4) node [anchor=north west][inner sep=0.75pt]    {$y$};
\draw (28,155.4) node [anchor=north west][inner sep=0.75pt]    {$x\land y$};
\draw (108,155.4) node [anchor=north west][inner sep=0.75pt]    {$x\land y$};
\draw (75+5,288.4+5) node [anchor=north west][inner sep=0.75pt]    {$\lor $};
\draw (24,209.4) node [anchor=north west][inner sep=0.75pt]    {$x$};
\draw (127,209.4) node [anchor=north west][inner sep=0.75pt]    {$y$};
\draw (30,353.4) node [anchor=north west][inner sep=0.75pt]    {$x\lor y$};
\draw (110,353.4) node [anchor=north west][inner sep=0.75pt]    {$x\lor y$};
\draw (351,156.4) node [anchor=north west][inner sep=0.75pt]    {$x\land y$};
\draw (271,156.4) node [anchor=north west][inner sep=0.75pt]    {$x\land y$};
\draw (368,12.4) node [anchor=north west][inner sep=0.75pt]    {$y$};
\draw (266,12.4) node [anchor=north west][inner sep=0.75pt]    {$x$};
\draw (315+5,89.4+5) node [anchor=north west][inner sep=0.75pt]    {$\land $};
\draw (510,156.4) node [anchor=north west][inner sep=0.75pt]    {$\overline{x} \lor \overline{y}$};
\draw (430,156.4) node [anchor=north west][inner sep=0.75pt]    {$\overline{x} \lor \overline{y}$};
\draw (527,12.4) node [anchor=north west][inner sep=0.75pt]    {$\overline{y}$};
\draw (424,12.4) node [anchor=north west][inner sep=0.75pt]    {$\overline{x}$};
\draw (475+5,91.4+5) node [anchor=north west][inner sep=0.75pt]    {$\lor $};
\draw (355,352.4) node [anchor=north west][inner sep=0.75pt]    {$x\lor y$};
\draw (275,352.4) node [anchor=north west][inner sep=0.75pt]    {$x\lor y$};
\draw (372,208.4) node [anchor=north west][inner sep=0.75pt]    {$y$};
\draw (269,208.4) node [anchor=north west][inner sep=0.75pt]    {$x$};
\draw (320+5,287.4+5) node [anchor=north west][inner sep=0.75pt]    {$\lor $};
\draw (518,351.4) node [anchor=north west][inner sep=0.75pt]    {$\overline{x} \land \overline{y}$};
\draw (438,351.4) node [anchor=north west][inner sep=0.75pt]    {$\overline{x} \land \overline{y}$};
\draw (535,207.4) node [anchor=north west][inner sep=0.75pt]    {$\overline{y}$};
\draw (433,207.4) node [anchor=north west][inner sep=0.75pt]    {$\overline{x}$};
\draw (482+5,284.4+5) node [anchor=north west][inner sep=0.75pt]    {$\land $};

\end{tikzpicture}
\caption{AND and OR gadgets for simulating AND/OR gates with fanin and fanout 2. For other values of fanin and fanout gadgets are the same but considering different number of inputs/outputs}
\label{fig:ANDORgatesgad}
\end{figure}
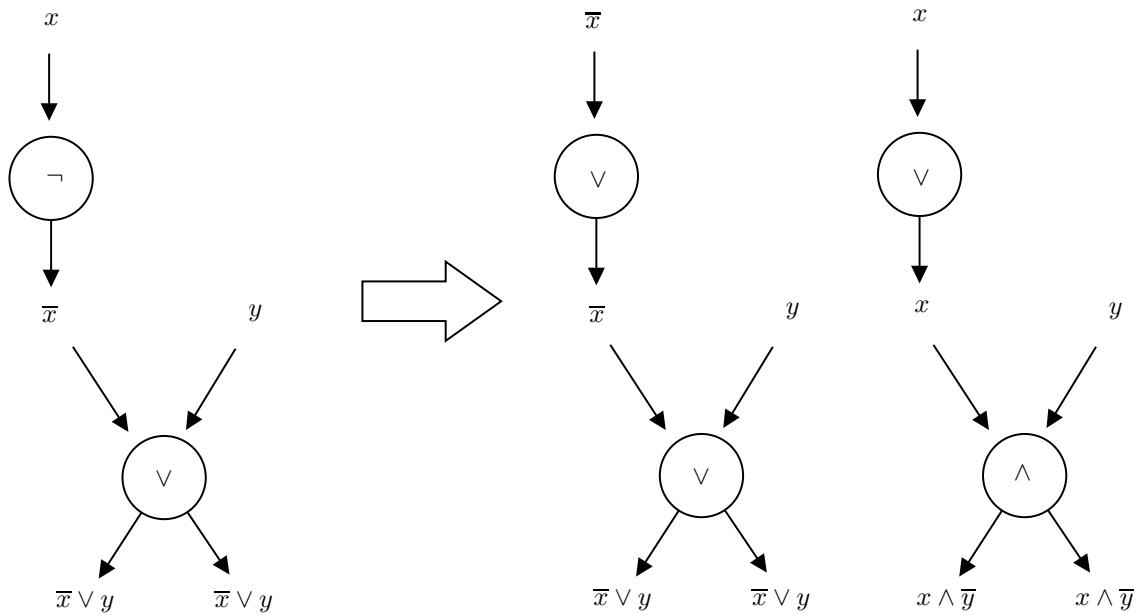
\begin{figure}
\centering

\tikzset{every picture/.style={line width=0.75pt}} 

\begin{tikzpicture}[x=0.75pt,y=0.75pt,yscale=-1,xscale=1]

\draw   (202,181) -- (244,181) -- (244,171) -- (272,191) -- (244,211) -- (244,201) -- (202,201) -- cycle ;
\draw    (44,66) -- (44,97) ;
\draw [shift={(44,100)}, rotate = 270] [fill={rgb, 255:red, 0; green, 0; blue, 0 }  ][line width=0.08]  [draw opacity=0] (8.93,-4.29) -- (0,0) -- (8.93,4.29) -- cycle    ;
\draw    (45,130) -- (45,181) ;
\draw [shift={(45,184)}, rotate = 270] [fill={rgb, 255:red, 0; green, 0; blue, 0 }  ][line width=0.08]  [draw opacity=0] (8.93,-4.29) -- (0,0) -- (8.93,4.29) -- cycle    ;
\draw  [fill={rgb, 255:red, 255; green, 255; blue, 255 }  ,fill opacity=1 ] (24,129) .. controls (24,117.4) and (33.4,108) .. (45,108) .. controls (56.6,108) and (66,117.4) .. (66,129) .. controls (66,140.6) and (56.6,150) .. (45,150) .. controls (33.4,150) and (24,140.6) .. (24,129) -- cycle ;

\draw    (482,64) -- (482,95) ;
\draw [shift={(482,98)}, rotate = 270] [fill={rgb, 255:red, 0; green, 0; blue, 0 }  ][line width=0.08]  [draw opacity=0] (8.93,-4.29) -- (0,0) -- (8.93,4.29) -- cycle    ;
\draw    (483,128) -- (483,179) ;
\draw [shift={(483,182)}, rotate = 270] [fill={rgb, 255:red, 0; green, 0; blue, 0 }  ][line width=0.08]  [draw opacity=0] (8.93,-4.29) -- (0,0) -- (8.93,4.29) -- cycle    ;
\draw  [fill={rgb, 255:red, 255; green, 255; blue, 255 }  ,fill opacity=1 ] (462,127) .. controls (462,115.4) and (471.4,106) .. (483,106) .. controls (494.6,106) and (504,115.4) .. (504,127) .. controls (504,138.6) and (494.6,148) .. (483,148) .. controls (471.4,148) and (462,138.6) .. (462,127) -- cycle ;
\draw    (319,65) -- (319,96) ;
\draw [shift={(319,99)}, rotate = 270] [fill={rgb, 255:red, 0; green, 0; blue, 0 }  ][line width=0.08]  [draw opacity=0] (8.93,-4.29) -- (0,0) -- (8.93,4.29) -- cycle    ;
\draw    (320,129) -- (320,180) ;
\draw [shift={(320,183)}, rotate = 270] [fill={rgb, 255:red, 0; green, 0; blue, 0 }  ][line width=0.08]  [draw opacity=0] (8.93,-4.29) -- (0,0) -- (8.93,4.29) -- cycle    ;
\draw  [fill={rgb, 255:red, 255; green, 255; blue, 255 }  ,fill opacity=1 ] (299,128) .. controls (299,116.4) and (308.4,107) .. (320,107) .. controls (331.6,107) and (341,116.4) .. (341,128) .. controls (341,139.6) and (331.6,149) .. (320,149) .. controls (308.4,149) and (299,139.6) .. (299,128) -- cycle ;
\draw    (327,213) -- (353.36,253.49) ;
\draw [shift={(355,256)}, rotate = 236.93] [fill={rgb, 255:red, 0; green, 0; blue, 0 }  ][line width=0.08]  [draw opacity=0] (8.93,-4.29) -- (0,0) -- (8.93,4.29) -- cycle    ;
\draw    (409,214) -- (385.56,252.44) ;
\draw [shift={(384,255)}, rotate = 301.37] [fill={rgb, 255:red, 0; green, 0; blue, 0 }  ][line width=0.08]  [draw opacity=0] (8.93,-4.29) -- (0,0) -- (8.93,4.29) -- cycle    ;
\draw    (373,279) -- (404.32,325.51) ;
\draw [shift={(406,328)}, rotate = 236.04] [fill={rgb, 255:red, 0; green, 0; blue, 0 }  ][line width=0.08]  [draw opacity=0] (8.93,-4.29) -- (0,0) -- (8.93,4.29) -- cycle    ;
\draw    (373,279) -- (341.63,327.48) ;
\draw [shift={(340,330)}, rotate = 302.90999999999997] [fill={rgb, 255:red, 0; green, 0; blue, 0 }  ][line width=0.08]  [draw opacity=0] (8.93,-4.29) -- (0,0) -- (8.93,4.29) -- cycle    ;
\draw  [fill={rgb, 255:red, 255; green, 255; blue, 255 }  ,fill opacity=1 ] (352,279) .. controls (352,267.4) and (361.4,258) .. (373,258) .. controls (384.6,258) and (394,267.4) .. (394,279) .. controls (394,290.6) and (384.6,300) .. (373,300) .. controls (361.4,300) and (352,290.6) .. (352,279) -- cycle ;
\draw    (490,213) -- (516.36,253.49) ;
\draw [shift={(518,256)}, rotate = 236.93] [fill={rgb, 255:red, 0; green, 0; blue, 0 }  ][line width=0.08]  [draw opacity=0] (8.93,-4.29) -- (0,0) -- (8.93,4.29) -- cycle    ;
\draw    (572,214) -- (548.56,252.44) ;
\draw [shift={(547,255)}, rotate = 301.37] [fill={rgb, 255:red, 0; green, 0; blue, 0 }  ][line width=0.08]  [draw opacity=0] (8.93,-4.29) -- (0,0) -- (8.93,4.29) -- cycle    ;
\draw    (536,279) -- (567.32,325.51) ;
\draw [shift={(569,328)}, rotate = 236.04] [fill={rgb, 255:red, 0; green, 0; blue, 0 }  ][line width=0.08]  [draw opacity=0] (8.93,-4.29) -- (0,0) -- (8.93,4.29) -- cycle    ;
\draw    (536,279) -- (504.63,327.48) ;
\draw [shift={(503,330)}, rotate = 302.90999999999997] [fill={rgb, 255:red, 0; green, 0; blue, 0 }  ][line width=0.08]  [draw opacity=0] (8.93,-4.29) -- (0,0) -- (8.93,4.29) -- cycle    ;
\draw  [fill={rgb, 255:red, 255; green, 255; blue, 255 }  ,fill opacity=1 ] (515,279) .. controls (515,267.4) and (524.4,258) .. (536,258) .. controls (547.6,258) and (557,267.4) .. (557,279) .. controls (557,290.6) and (547.6,300) .. (536,300) .. controls (524.4,300) and (515,290.6) .. (515,279) -- cycle ;

\draw    (56,214) -- (82.36,254.49) ;
\draw [shift={(84,257)}, rotate = 236.93] [fill={rgb, 255:red, 0; green, 0; blue, 0 }  ][line width=0.08]  [draw opacity=0] (8.93,-4.29) -- (0,0) -- (8.93,4.29) -- cycle    ;
\draw    (138,215) -- (114.56,253.44) ;
\draw [shift={(113,256)}, rotate = 301.37] [fill={rgb, 255:red, 0; green, 0; blue, 0 }  ][line width=0.08]  [draw opacity=0] (8.93,-4.29) -- (0,0) -- (8.93,4.29) -- cycle    ;
\draw    (102,280) -- (133.32,326.51) ;
\draw [shift={(135,329)}, rotate = 236.04] [fill={rgb, 255:red, 0; green, 0; blue, 0 }  ][line width=0.08]  [draw opacity=0] (8.93,-4.29) -- (0,0) -- (8.93,4.29) -- cycle    ;
\draw    (102,280) -- (70.63,328.48) ;
\draw [shift={(69,331)}, rotate = 302.90999999999997] [fill={rgb, 255:red, 0; green, 0; blue, 0 }  ][line width=0.08]  [draw opacity=0] (8.93,-4.29) -- (0,0) -- (8.93,4.29) -- cycle    ;
\draw  [fill={rgb, 255:red, 255; green, 255; blue, 255 }  ,fill opacity=1 ] (81,280) .. controls (81,268.4) and (90.4,259) .. (102,259) .. controls (113.6,259) and (123,268.4) .. (123,280) .. controls (123,291.6) and (113.6,301) .. (102,301) .. controls (90.4,301) and (81,291.6) .. (81,280) -- cycle ;

\draw (91+5,270.4+5) node [anchor=north west][inner sep=0.75pt]    {$\lor $};
\draw (143,191.4) node [anchor=north west][inner sep=0.75pt]    {$y$};
\draw (46,335.4) node [anchor=north west][inner sep=0.75pt]    {$\overline{x} \lor y$};
\draw (126,335.4) node [anchor=north west][inner sep=0.75pt]    {$\overline{x} \lor y$};
\draw (313,43.4) node [anchor=north west][inner sep=0.75pt]    {$\overline{x}$};
\draw (362+5,269.4+5) node [anchor=north west][inner sep=0.75pt]    {$\lor $};
\draw (414,190.4) node [anchor=north west][inner sep=0.75pt]    {$y$};
\draw (317,334.4) node [anchor=north west][inner sep=0.75pt]    {$\overline{x} \lor y$};
\draw (397,334.4) node [anchor=north west][inner sep=0.75pt]    {$\overline{x} \lor y$};
\draw (524+5,267.4+5) node [anchor=north west][inner sep=0.75pt]    {$\land $};
\draw (577,190.4) node [anchor=north west][inner sep=0.75pt]    {$y$};
\draw (480,334.4) node [anchor=north west][inner sep=0.75pt]    {$x\land \overline{y}$};
\draw (560,334.4) node [anchor=north west][inner sep=0.75pt]    {$x\land \overline{y}$};
\draw (310+5,119.4+5) node [anchor=north west][inner sep=0.75pt]    {$\lor $};
\draw (315,191.4) node [anchor=north west][inner sep=0.75pt]    {$\overline{x}$};
\draw (478,42.4) node [anchor=north west][inner sep=0.75pt]    {$x$};
\draw (473+5,118.4+5) node [anchor=north west][inner sep=0.75pt]    {$\lor $};
\draw (479,189.4) node [anchor=north west][inner sep=0.75pt]    {$x$};
\draw (38+3,119.4+5) node [anchor=north west][inner sep=0.75pt]    {$\neg $};
\draw (40,44.4) node [anchor=north west][inner sep=0.75pt]    {$x$};
\draw (39,191.4) node [anchor=north west][inner sep=0.75pt]    {$\overline{x}$};

\end{tikzpicture}

\caption{NOT gadget wiring for circuit simulation using gates from $\Gmon.$ In this case a NOT gate is connected to an OR gate in the original circuit. Copies of the NOT gate in the circuit performing simulation are connected to the copies of the OR gate switched: positive part is connected to negative part of the OR gate and viceversa.}
\label{fig:NOTgatesgad}

\end{figure}

\begin{theorem}
The family $\Gamma(\Gmontwo)$ simulates in constant time and linear space the family $\Gamma(\Gmon)$, i.e. there exists a constant function $T: \N \to \N$ and a linear function $S: \N \to \N$ such that $\Gamma(\Gmon) \preccurlyeq^{T}_{S}\Gamma(\Gmontwo)$
\label{theo:gmontwosimgmon}
\end{theorem}

\begin{proof}
Let $F:Q^{V} \to Q^{V}$ be an arbitrary $\Gmon$-network coded by its standard representation defined by a list of gates $g_{1}, \hdots, g_{n}$ and two functions $\alpha$ an $\beta$ mapping inputs to nodes in $F$ and nodes in $F$ to outputs respectively. We are going to construct in $\DLOG$ a $\Gmontwo$-network $G$ that simulates $F$ in time $T = \mathcal{O}(1)$ and space $S = \mathcal{O}(|V|)$ where $H$ is the communication graph of $F$. In order to do that, we are going to replace each gate $g_{k}$ by a small gadget. More precisely, we are going to introduce the following coding function: $x \in \{0,1\} \to (x,x,0) \in \{0,1\}^{3}$. We are going to define gadgets for each gate. Let us take $k \in \{1,\hdots,n\}$ and call $g^{*}_{k}$ the corresponding gadget associated to $g_{k}$. Let us that suppose $g_{k}$ is an OR gate and that it has fanin $2$ and fanout $1$ then, we define $g^{*}: \{0,1\}^{6} \to \{0,1\}^{6}$ as a function that for each input of the form $(x,x,y,y,0,0)$ produces the output $g^{*}((x,x,y,y,0,0)) = (x \vee y, x \vee y, 0,0,0,0)$. The case fanin $1$ and fanout $1$ is given by $g^{*}((x,x,0,0,0,0)) = (x, x, 0,0,0,0)$,  the case fanin 2 and fanout 2 is given by $g^{*}((x,x,y,y,0,0)) = (x \vee y, x \vee y, x \vee y, x \vee y,0,0)$ and finally case fanin $1$ and fanout $2$ is given by the same latter function but on input (x,x,0,0,0,0). The AND case is completely analogous. We are going to implement the previous functions as small (constant depth) synchronized circuits that we call block gadgets. More precisely, we are going to identify functions $g^{*}$ with its correspondent block gadget. The detail on the construction of these circuits that define latter functions are provided in Figures \ref{fig:simgmongmon21}, \ref{fig:simgmongmon22} and \ref{fig:simgmongmon23}.

Once we have defined the structure of block gadgets, we have to manage connections between them and  also manage the fixed $0$ inputs that we have added in addition  to the zeros that are produced by the blocks as outputs. In order to do that, let us assume that gates $g_{i} $ and $g_{j}$ are connected. Note from the discussion on coding above that AND/OR gadgets have between $2$ and $4$ inputs and outputs fixed to $0$. In particular, as it is shown in Figures \ref{fig:simgmongmon21}, \ref{fig:simgmongmon22} and \ref{fig:simgmongmon23}, all the block gadgets have the same amount of zeros in the input and in the output with the exception of the fanin $1$ fanout $2$ gates and the fanin $2$ fanout $1$ gates. However, as $\GG$-networks are closed systems (the amount of inputs must be the same that the amount of outputs) we have that, for each  fanin $1$ fanout $2$  gate,  it must be a fanin $2$ fanout $1$ gate and vice versa (otherwise there would be more input than outputs or more outputs than inputs). In other words, there is a bijection between the set of fanin $1$ fanout $2$  gates and the set of fanin $2$ fanout $1$. Observe that fanin $2$ fanout $1$ gates consume $2$ zeros in input but produce $4$ zeros in output while fanin $1$ fanout $2$ gates consume $4$ in input and produce $2$ zeros in output (see Figures \ref{fig:simgmongmon22} and \ref{fig:simgmongmon23}). So, between $g^{*}_{i}$ and $g^{*}_{j}$ we have to distinct two cases: a) if both gates have the same number of inputs and outputs, connections are managed in the obvious way i.e., outputs corresponding to the computation performed by original gate are assigned between $g^{*}_{i}$ and $g^{*}_{j}$  and each gate uses the same zeros they produce to feed its inputs. b) if $g^{*}_{i}$ or $g^{*}_{j}$ have more inputs than outputs or vice versa, we have to manage the extra zeros (needed or produced). Without lost of generality, we assume that $g^{*}_{i}$ is fanin $2$ fanout $1$. Then, by latter observation it must exists another gate $g_{k}$ and thus, a gadget block $g^{*}_{k}$ with fanin $1$ and fanout $2$. We simply connect extra zeros produced by  $g^{*}_{i}$ to block  $g^{*}_{k}$ and we do the same we did in previous case in order to manage connections.

Note that $F^{*}$  is constructible in $\DLOG$ as it suffices to read the standard representation of $F$ and produce the associated block gadgets which have constant size. In addition we have that previous encoding $g \to g^*$  induce a block map $\phi: \{0,1\}^{V} \to \{0,1\}^{V^{+}}$  where $|V^{+}| =  \mathcal{O} ( |V| )$ and that $\phi \circ F = F^{*T} \circ \phi$ where $T=6$ is the size of each gadget block in $F^{*}$. We conclude that $F^{*} \in \Gamma(\Gmontwo)$ simulates  $F$ in space $|V^{+}| = \mathcal{O} ( |V| )$ and time $T =6$ and thus,  $ \Gamma(\Gmon)  \preccurlyeq^T_{S} \Gamma(\Gmontwo)$ where $T$ is constant and $S$ is a linear function.  .

 \end{proof}
\begin{figure}
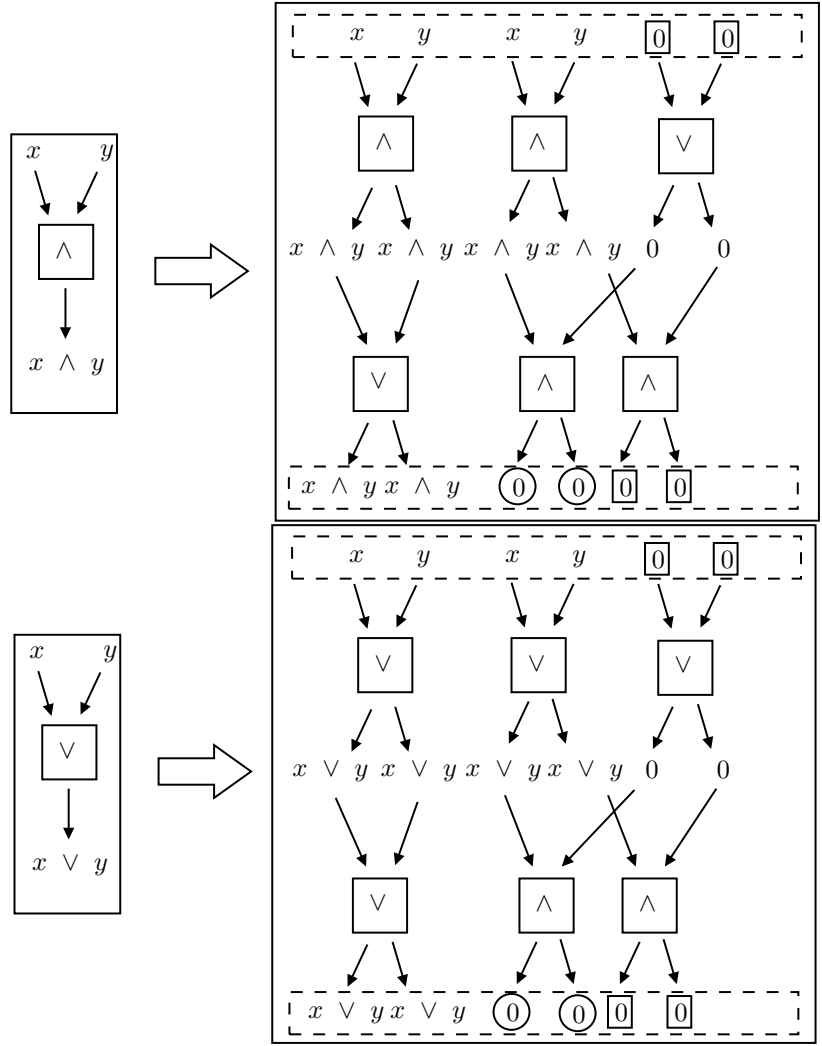

\centering

\tikzset{every picture/.style={line width=0.75pt}} 



\caption{Block gadgets for simulating Fanin $2$ Fanout $1$ AND/OR gates using only gates in $\Gmontwo$. Squared zeros represent the amount of zeros that can be used as inputs for the same block. Circled zeros correspond to extra zeros that need to be assigned to a Fanin $1$ Fanout $2$ gate.}
\label{fig:simgmongmon21}
\end{figure}

\begin{figure}
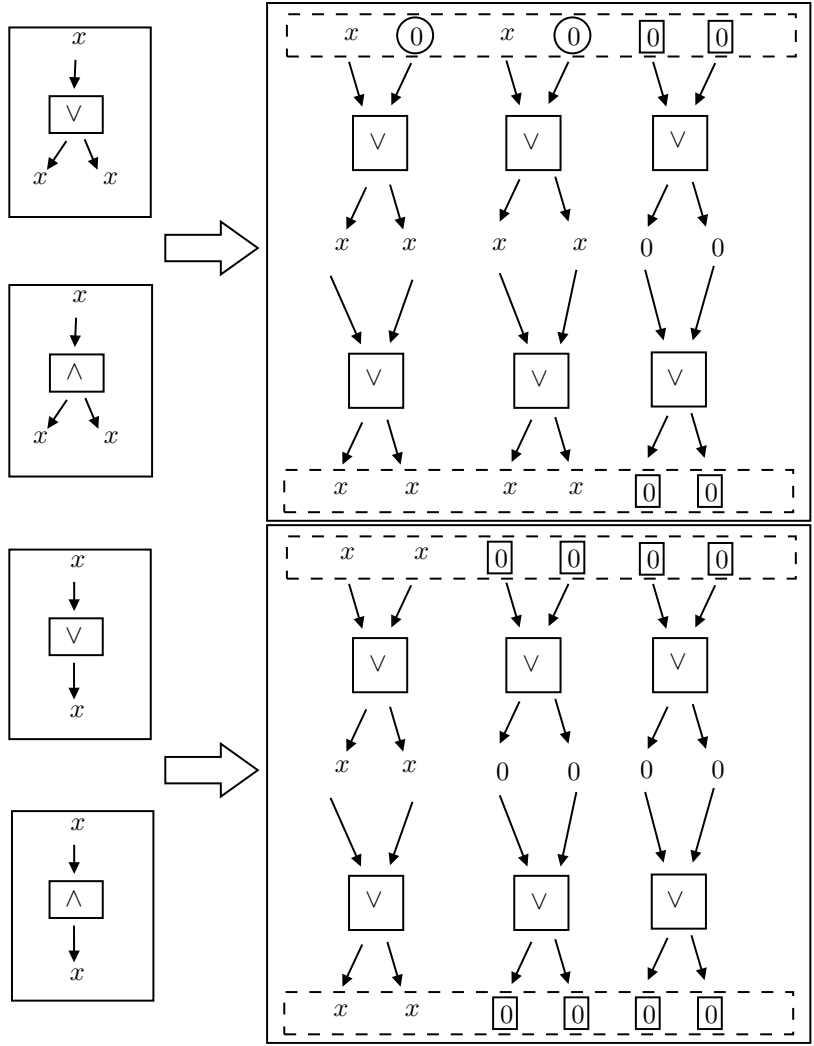


\centering
\tikzset{every picture/.style={line width=0.75pt}} 

%

\caption{Block gadgets for simulating  Fanin $1$ Fanout $2$ and Fanin $1$ Fanout $1$ AND/OR gates using only gates in $\Gmontwo$. Squared zeros represent the amount of zeros that can be used as inputs for the same block. Circled zeros correspond to extra zeros that need to be received from a Fanin $2$ Fanout $1$ gate.}
\label{fig:simgmongmon22}
\end{figure}

\begin{figure}
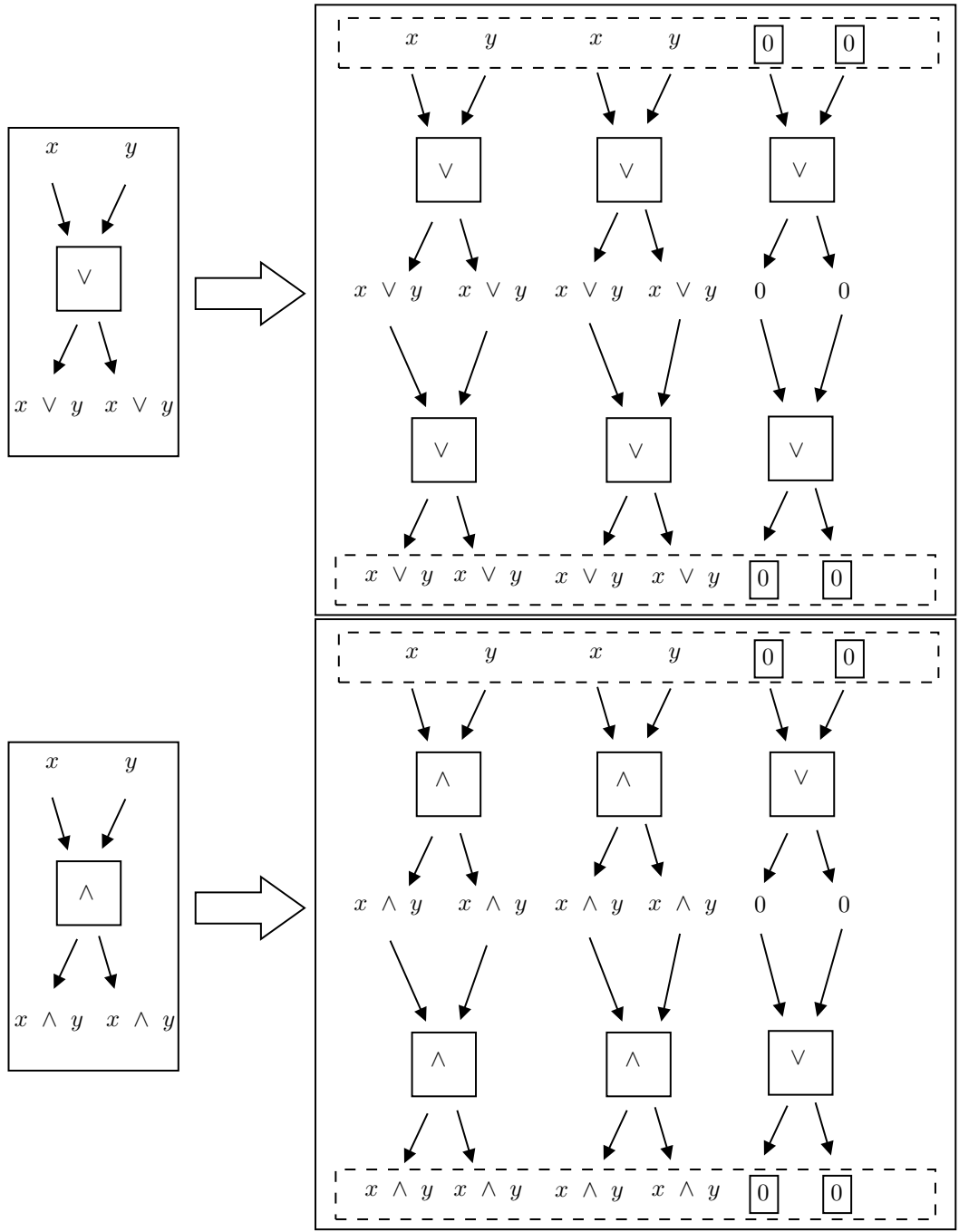

\centering

\tikzset{every picture/.style={line width=0.75pt}} 

%

\caption{Block gadgets for simulating  fanin $2$ fanout $2$ AND/OR gates using only gates in $\Gmontwo$. Squared zeros represent the amount of zeros that can be used as inputs for the same block. }
\label{fig:simgmongmon23}
\end{figure}
\begin{corollary}
The family $\Gamma(\Gmontwo)$ is strongly universal.
\end{corollary}

\begin{proof}
Result is direct from Theorem \ref{theo:gmon-univ} ($\Gamma(\Gmon)$ is strongly universal)  and Theorem \ref{theo:gmontwosimgmon} ($ \Gamma(\Gmon)  \preccurlyeq^T_{S} \Gamma(\Gmontwo)$ where $T$ is constant and $S$ is a linear function).
\end{proof}


Now we show the universality of $\Gamma(\Gmon)$. The proof is essentially a consequence of \cite[Theorem 6.2.5]{Greenlaw_1995}.
Roughly, latter result starts with alternated monotone circuit which has only fanin $2$ and fanout $2$ gates (previous results in the same reference show that one can always reduce to this case starting from an arbitrary circuit) and  gives an $\textbf{NC}^{1}$ construction of a synchronous circuit preserving latter properties.
We need additional care here because we want a reusable circuit whose output is fed back to its input.
Note also that the construction uses quadratic space in the number of gates of the circuit given in input, so we cannot show strong universality this way but only universality.
\begin{theorem}\label{teo:gmonuniv}
The family $\Gamma(\Gmon)$ of all $\Gmon$-networks is universal
\end{theorem}
\begin{proof}
Let $F:Q^{k} \to Q^{k}$ be some arbitrary network with a circuit representation $C:\{0,1\}^{n} \to \{0,1\}^{n}$ such that $n = k^{\mathcal{O}(1)}.$ By \cite[Theorem 6.2.5]{Greenlaw_1995} we can assume that there exists a circuit  $C':\{0,1\}^{n'} \to \{0,1\}^{n'}$ where $n' = \mathcal{O}(n^{2})$ such that $C'$ is synchronous alternated and monotone. In addition, every gate in $C'$ has fanin and fanout $2$. We remark that latter reference do not only provides the standard encoding of $C'$ but also give us a $\DLOG$ algorithm (it is actually $\textbf{NC}^{1}$) which takes the standard representation of $C:\{0,1\}^{n} \to \{0,1\}^{n}$ and produces $C'$. We are going to slightly modify latter algorithm in order to construct not only a circuit but a $\Gmon$-network. In fact, the only critical point is to manage the identification between outputs and inputs. This is not direct from the result by Ruzzo et al. as their algorithm involves duplication of inputs and also adding constant inputs. In order to manage this, it suffices to simply modify their construction in order to mark original, copies and constant inputs. Then, as $\Gmon$ includes COPY gates and also AND/OR gates with fanout $1$, one can always produce copies of certain input if we need more, or erase extra copies by adding and small tree of $\mathcal{O}(\log(n))$ depth. Same goes for constant inputs. Formally, at the end of the algorithm, the $\DLOG$ can read extra information regarding copies and constant inputs, and then can construct $\mathcal{O}(\log(n))$ depth circuit that produces a coherent encoding for inputs and outputs. This latter construction defines a $\Gmon$-network $G:\{0,1\}^{n''} \to \{0,1\}^{n''}$ and an encoding $\phi:Q^{n} \to \{0,1\}^{n''}$ where $n'' = \mathcal{O}(n^{2})$  such that $\phi \circ F = G^{T} \circ \phi$ where $T = \mathcal{O}(\text{depth(C')} + \log(n))$. Thus, $\Gmon$ is universal.
\end{proof}

We can now state the following direct corollary.
\begin{corollary}\label{coro:strongunivimpliesuniv}
Let $\mathcal{F}$ be a strongly universal automata network family. Then, $\mathcal{F}$ is universal.
\end{corollary}
\begin{proof}
In order to show the result, it suffices to exhibit a $\GG$-network family ($\GG$-networks are bounded degree networks) which is strongly universal and universal at the same time. By Theorem \ref{teo:gmonuniv} we take $\GG = \Gmon$ and thus, corollary holds.
\end{proof}

\begin{corollary}\label{cor:univfrommon}
  Let $\GG$ be either $\Gmon$ or $\Gmontwo$.
  Any family $\mathcal{F}$ that has coherent $\GG$-gadgets contains a subfamily of bound degree networks with bounded degree representation which is  (strongly) universal. Any CSAN family with coherent $\GG$-gadgets is  (strongly) universal.
\end{corollary}

\subsection{Closure and synchronous closure}

Although monotone gates are sometimes easier to realize in concrete dynamical system which make the above results useful, there is nothing special about them to achieve universality: any set of gates that are expressive enough for Boolean functions yields the same universality result. Given a set of maps $\GG$ over alphabet $Q$, we define its \emph{closure} ${\overline{\GG}}$ as the set of maps that are computed by circuits that can be built using only gates from $\GG$. More precisely, $\overline{\GG}$ is the closure of $\GG$ by composition, \textit{i.e.} forming from maps ${g_1:Q^{I_1}\to Q^{O_1}}$ and ${g_2:Q^{I_2}\to Q^{O_2}}$ (with ${I_1, I_2, O_1, O_2}$ disjoint) a composition $g$ by plugging a subset of outputs $O\subseteq O_2$ of $g_1$ into a subset of inputs $I\subseteq I_2$ of $g_2$, thus obtaining ${g: Q^{I_1\cup I_2\setminus I}\to Q^{O_1\setminus O\cup O_2}}$ with 
\[g(x)_o =
  \begin{cases}
    g_1(x_{I_1})_o &\text{ if }o\in O_1\setminus O\\
    g_2(y)_o &\text{ if }o\in O_2
  \end{cases}
\]
where ${y_j = x_j}$ for ${j\in I_2\setminus I}$ and ${y_j= g_1(x_{I_1})_{\pi(j)}}$ where ${\pi:I\to O}$ is the chosen bijection between $I$ and $O$ (the wiring of outputs of $g_1$ to inputs of $g_2$).
A composition is \emph{synchronous} if either ${I=\emptyset}$ or ${I=I_2}$. We then define the synchronous closure ${\overline{\GG}^2}$ as the closure by synchronous composition. The synchronous composition correspond to synchronous circuits with gates in $\GG$. A $\GG$-circuit is a sequence of compositions starting from elements of $\GG$. It is synchronous if the compositions are synchronous. The depth of a $\GG$-circuit is the maximal length of a path from an input to an output. In the case of a synchronous circuits, all such path are of equal length.

\begin{figure}
  \centering
  \begin{minipage}{.2\linewidth}
    \begin{tikzpicture}[scale=.5]
      \draw (-1.5,1)--(2.5,1)--(2.5,-1)--(-1.5,-1)--cycle;
      \draw[->,very thick] (-1,-2)  -- (-1,-1);
      \draw[->,very thick] (0,-2) -- (0,-1);
      \draw[->,very thick] (1,-2) -- (1,-1);
      \draw[->,very thick] (2,-2) -- (2,-1);
      \draw[->,very thick] (3,-2) -- (3,2);
      \draw[->,very thick] (-1,1)  -- (-1,2);
      \draw[->,very thick] (0,1) -- (0,2);
      \draw[->,very thick] (1,1) -- (1,5);
      \draw[->,very thick] (2,1) -- (2,2);
      \begin{scope}[shift={(0,3)}]
        \draw (-1.5,1)--(.5,1)--(.5,-1)--(-1.5,-1)--cycle;
        \draw[->,very thick] (-1,1)  -- (-1,2);
        \draw[->,very thick] (0,1) -- (0,2);
      \end{scope}
      \begin{scope}[shift={(3,3)}]
        \draw (-1.5,1)--(.5,1)--(.5,-1)--(-1.5,-1)--cycle;
        \draw[->,very thick] (-1,1)  -- (-1,2);
        \draw[->,very thick] (0,1) -- (0,2);
      \end{scope}
    \end{tikzpicture}
  \end{minipage}
  \begin{minipage}{.2\linewidth}
    \begin{tikzpicture}[scale=.5]
      \draw (-1.5,1)--(2.5,1)--(2.5,-1)--(-1.5,-1)--cycle;
      \draw[->,very thick] (-1,-2)  -- (-1,-1);
      \draw[->,very thick] (0,-2) -- (0,-1);
      \draw[->,very thick] (1,-2) -- (1,-1);
      \draw[->,very thick] (2,-2) -- (2,-1);
      \draw[->,very thick] (-1,1)  -- (-1,2);
      \draw[->,very thick] (0,1) -- (0,2);
      \draw[->,very thick] (1,1) -- (1,2);
      \draw[->,very thick] (2,1) -- (2,2);
      \begin{scope}[shift={(0,3)}]
        \draw (-1.5,1)--(.5,1)--(.5,-1)--(-1.5,-1)--cycle;
        \draw[->,very thick] (-1,1)  -- (-1,2);
        \draw[->,very thick] (0,1) -- (0,2);
      \end{scope}
      \begin{scope}[shift={(2,3)}]
        \draw (-1.5,1)--(.5,1)--(.5,-1)--(-1.5,-1)--cycle;
        \draw[->,very thick] (-1,1)  -- (-1,2);
        \draw[->,very thick] (0,1) -- (0,2);
      \end{scope}
    \end{tikzpicture}
  \end{minipage}
  \caption{Non-synchronous composition (on the left) and synchronous composition (on the right).}
  \label{fig:sync-nonsync-composition}
\end{figure}
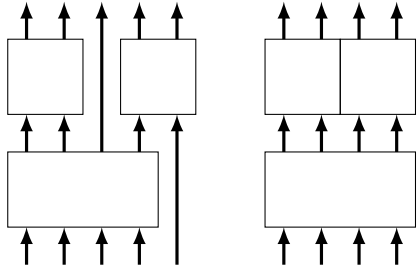

\begin{remark}\label{rem:clone}
  The above definitions are very close to the classical notion of clones \cite{post_lattice}. However, we stress that, in our case, projections maps ${Q^k\to Q}$ are generally not available, nor duplication maps ${x\mapsto (x,x)}$ allowing to use the same variable several times. This is important because in a given dynamical systems, erasing or duplicating information might be impossible (think about reversible systems) and hiding it into some non-coding part might be complicated.
\end{remark}

\begin{proposition}\label{prop:closuregadgets}
  Fix some alphabet $Q$ and consider two finite sets of maps $\GG$ and $\GG'$ over alphabet $Q$ such that:
  \begin{itemize}
  \item either contains the identity map $Q\to Q$ and is such that $\overline{\GG}$ contains $\GG'$,
  \item or there is an integer $k$ such that ${\overline{\GG}^2_k}$, the set of elements of $\overline{\GG}^2$ that can be realized by a circuit of depth $k$, contains $\GG'$.
  \end{itemize}
  Then, any family $\mathcal{F}$ that has coherent $\GG$-gadgets has coherent $\GG'$-gadgets.
\end{proposition}

\begin{proof}
  Suppose first that the first item holds.  Since ${\overline{\GG}}$ contains $\GG'$ there must exist a circuit made of gates from $\GG$ that produces any given element $g\in\GG'$. One then wants to apply gadget glueing on gadgets from $\GG$ to mimic the composition and thus obtain a gadget corresponding to $g$. However this doesn't work as simply because propagation delay is a priori not respected at each gate in the circuit composition yielding $g$ and there is a risk that information arrives distinct delays at different outputs. However, since $\GG$ contains the identity map, there is a corresponding gadget in the family that actually implements a delay line. This additionnal gadget solves the problem: it is straightforward to transform by padding with identity gates all circuit with gates in $\GG$ into synchronous ones. Moreover, by padding again, we can assume that the finite set of such circuits computing elements of $\Gmon$ are all of same depth. It is then straightforward to translate this set of circuits into coherent $\Gmon$-gadgets by iterating gadget glueing and using Lemma~\ref{lem:pseudo-orbit-glueing}.

  If the second item holds the situation is actually simpler because the synchronous closure contains only synchronous circuits of gates from $\Gmontwo$ so we can directly translate the circuits producing the maps of $\Gmontwo$ into gadgets via gadget glueing by Lemma~\ref{lem:pseudo-orbit-glueing} as in the previous case. Moreover, the hypothesis is that all elements of $\GG'$ are realized by circuit of same depth so we get gadgets that share the same time constant.
   
\end{proof}

\newcommand\Gnand{\GG_{\text{NAND}}}

As a direct corollary of Proposition~\ref{prop:closuregadgets}, we can extend the results about strong universality of $\Gmontwo$ to other families of $\GG$-networks associated to elementary Boolean gates, like $\Gnor$ and ${\Gnand = \{\text{NAND}(x,y) = (\overline{x\wedge y},\overline{x\wedge y})\}}$.
Note however that classical results on Boolean gates and clone theory cannot be applied immediately (see Remark~\ref{rem:clone}) and the expected universality result requires a little bit of care.

\begin{corollary}
  The families ${\Gamma(\Gnor)}$ and $\Gamma(\Gnand)$ are strongly universal. \label{cor:NORNAND}
\end{corollary}
\begin{proof}
  First, since NAND and NOR gates are conjugated by negation, it is clear that families $\Gamma(\Gnor)$ and $\Gamma(\Gnand)$ simulate each other with time constant $1$ via a block embedding that just apply ${x\mapsto\overline{x}}$ at each node.
  It is thus sufficient to prove that $\Gamma(\Gnor)$ is strongly universal.
  Consider the two maps ${\alpha : \{0,1\}^4\to\{0,1\}^4}$ and ${\omega : \{0,1\}^4\to \{0,1\}^4}$ that are synchronous $\Gnor$-circuits of depth $2$ defined by :
  \begin{center}
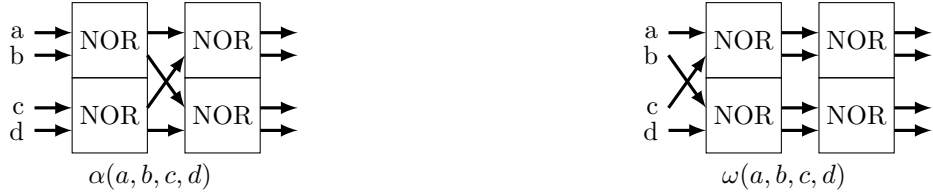
    \begin{minipage}{.45\linewidth}
      \begin{center}
        \begin{tikzpicture}
          \draw (0,0) rectangle +(1,1) node[pos=.5] {NOR};
          \draw (0,1) rectangle +(1,1) node[pos=.5] {NOR};
          \draw (1.5,0) rectangle +(1,1) node[pos=.5] {NOR};
          \draw (1.5,1) rectangle +(1,1) node[pos=.5] {NOR};
          \draw[->,very thick] (-.5,.3) node[left] {d} -- +(.5,0);
          \draw[->,very thick] (-.5,.6) node[left] {c} -- +(.5,0);
          \draw[->,very thick] (-.5,1.3) node[left] {b} -- +(.5,0);
          \draw[->,very thick] (-.5,1.6) node[left] {a} -- +(.5,0);
          \draw[->,very thick] (1,.3) -- +(.5,0);
          \draw[->,very thick] (2.5,.3) -- +(.5,0);
          \draw[->,very thick] (1,.6) -- +(.5,.7);
          \draw[->,very thick] (2.5,.6) -- +(.5,0);
          \draw[->,very thick] (1,1.3) -- +(.5,-.7);
          \draw[->,very thick] (2.5,1.3) -- +(.5,0);
          \draw[->,very thick] (1,1.6) -- +(.5,0);
          \draw[->,very thick] (2.5,1.6) -- +(.5,0);
        \end{tikzpicture}\\
        ${\alpha(a,b,c,d)}$
      \end{center}
    \end{minipage}\hskip 1cm
    \begin{minipage}{.45\linewidth}
      \begin{center}
        \begin{tikzpicture}
          \draw (0,0) rectangle +(1,1) node[pos=.5] {NOR};
          \draw (0,1) rectangle +(1,1) node[pos=.5] {NOR};
          \draw (1.5,0) rectangle +(1,1) node[pos=.5] {NOR};
          \draw (1.5,1) rectangle +(1,1) node[pos=.5] {NOR};
          \draw[->,very thick] (-.5,.3) node[left] {d} -- +(.5,0);
          \draw[->,very thick] (-.5,.6) node[left] {c} -- +(.5,.7);
          \draw[->,very thick] (-.5,1.3) node[left] {b} -- +(.5,-.7);
          \draw[->,very thick] (-.5,1.6) node[left] {a} -- +(.5,0);
          \draw[->,very thick] (1,.3) -- +(.5,0);
          \draw[->,very thick] (2.5,.3) -- +(.5,0);
          \draw[->,very thick] (1,.6) -- +(.5,0);
          \draw[->,very thick] (2.5,.6) -- +(.5,0);
          \draw[->,very thick] (1,1.3) -- +(.5,0);
          \draw[->,very thick] (2.5,1.3) -- +(.5,0);
          \draw[->,very thick] (1,1.6) -- +(.5,0);
          \draw[->,very thick] (2.5,1.6) -- +(.5,0);
        \end{tikzpicture}\\
        ${\omega(a,b,c,d)}$
      \end{center}
    \end{minipage}
  \end{center}
  By Proposition~\ref{prop:closuregadgets} (second item), the family $\Gamma(\Gnor)$ has coherent $\GG$-gadgets where ${\GG=\{\alpha,\omega\}}$ and therefore simulates the family $\Gamma(\GG)$ with constant spatio-temporal rescaling factors by Lemma~\ref{lem:from-gadgets-to-networks}.
  Now observe that for any ${x,y\in\{0,1\}}$ it holds that ${\alpha(x,x,y,y) = (a,a,a,a)}$ with ${a=x\wedge y}$ and ${\omega(x,x,y,y)=(o,o,o,o)}$ with ${o=x\vee y}$.
  This implies that family $\Gamma(\GG)$ simulates $\Gamma(\Gmontwo)$ with spatial rescaling factor 2 and temporal rescaling factor 1, simply by doubling each node because AND and OR gates of type ${\{0,1\}^2\to\{0,1\}^2}$ in $\Gmontwo$ are such that ${AND(x,y)=(a,a)}$ and ${OR(x,y)=(o,o)}$.
  We deduce that $\Gamma(\GG)$ and therefore $\Gamma(\Gnor)$ are strongly universal.
\end{proof}

\subsubsection{Game of life is strongly universal}
\begin{theorem}\label{theo:goluniv}
The family of outer-totalistic CSAN networks with $B={3}$ and $S={2,3}$ i.e. \emph{Game of life} automata networks, is strongly universal.
\end{theorem}

\begin{proof}
The result holds as a direct consequence of the Lemma \ref{lemma:GOCcoherent} which tell us that Game of life automata networks have coherent $\mathcal{G}_{\text{NOR}}$ gadgets, the Corollary \ref{cor:gadgets-mon-csan} which tell us that the family of Game of life automata networks simulates $\Gamma(\mathcal{G}_{\text{NOR}})$ in constant time and linear space and finally, the Corollary \ref{cor:NORNAND} which tell us that the familly  $\Gamma(\mathcal{G}_{\text{NOR}})$ is strongly universal and thus, the family of  Game of life automata networks are strongly universal. 
\end{proof}

\begin{remark}
We would like to remark three things about the latter result:
\begin{enumerate}
\item The first one is that the gadget used in the proof of Lemma \ref{lemma:GOCcoherent} is simpler than the case of cellular automata and intrinsic universality (see \cite{liferokadur} for more details).  In particular, the fact that the communication graph can be chosen freely, allow us to transmit information and perform calculations in less time. 
\item The second one is that the result is an improvement of the result obtained in the cellular automata context since, as we see in Corollary \ref{cor:intrinsicnotstrongly}, intrinsic universality is not enough for strong universlity. Again Theorem \ref{theo:ballgrowthnouniv} provides some insight on how the properties of the communication graph play an important role in terms of the universality. 
\item The third one is that the latter result is an application of a series of results contained in the article which can be easily applied to any other family. In other words, given a family of automata networks one can show the strong universality by simply showing that the family admits a set of coherent gadgets. We stress that this approach allows to derive a perfectly rigorous proof of many facts implied by strong universality (\textit{e.g.} Corollary~\ref{cor:universality} and Theorem~\ref{them:univ-rich-dynamics}) from a rather small set of observations on a finite set of pseudo-orbits of small automata networks (see proof of Lemma~\ref{lemma:GOCcoherent}).
\end{enumerate}
\end{remark}

\subsection{Super-polynomial periods without universality}

A universal family must exhibit super-polynomial periods, however universality is far from necessary to have this dynamical feature. In this subsection we define the family of wire networks to illustrate this.

 In order to do that, we need the following classical result about the growth of Chebyshev function and prime number theorem. 

\begin{lemma} \cite{hardy2008introduction}
  Let $m \geq 2$ and $\mathcal{P}(m) = \{p \leq m \text{ } | \text{ }  p \text{ prime}\}$. If we define $\pi(m) = |\mathcal{P}(m)|$ and  $\theta(m)  = \sum \limits_{p \in \mathcal{P}(m)}\log (p)$  then we have ${\pi (m) \sim \frac{m}{\log(m)}}$ and ${\theta(m) \sim m}$.
  \label{lemma:primes}
\end{lemma}
\newcommand\Gwire{\GG_{w}}

By using the Lemma \ref{lemma:primes} we can construct automata networks with non-polynomial cycles simply by making disjoint union of rotations (\textit{i.e.} network whose interaction graph is a cycle that just rotate the configuration at each step). Indeed, it is sufficient to consider rotations on cycle whose length are successive prime numbers. It turns out that these automata networks are exactly $\Gwire$-networks where $\Gwire$ is a single 'wire gate': ${\Gwire =\{id_B\}}$ where ${id_B}$ is the identity map over $\{0,1\}$.

Formally, according to Definition \ref{def:g-network}, for any $\Gwire$-network $F:Q^V \to Q^V$ there exist a partition $V=C_1\cup C_2 \hdots \cup C_k$ where $C_i = \{u^i_1\hdots,u^i_{l_i}\}$ with ${l_i\geq 2}$ for each $i = 1,\hdots,k$ and $F(x)_{u^i_{s+1\bmod l_i}}= x_{u^i_s}$ for any $x \in Q^V$ and $0\leq s \leq l_i$.


\begin{theorem}\label{theo:non-polyn-cycl}
  Any family $\mathcal{F}$ that has coherent $\Gwire$-gadgets has superpolynomial cycles, more precisely: there is some ${\alpha>0}$ such that for infinitely many ${n\in\N}$, there exists a network $F_n\in\mathcal{F}$ with ${O(n)}$ nodes and a periodic orbit of size ${\Omega(\exp (n^\alpha))}$. 
\end{theorem}
\begin{proof}
  Taking the notations of Lemma~\ref{lemma:primes}, define for any $n$ the $\Gwire$-network $G_n$ made of disjoint union of circuits of each prime length less than $n$. $G_n$ has size at most ${n\pi(n)}$ and if we consider a configuration $x$ which is in state $1$ at exactly one node in each of the $\pi(n)$ disjoint circuit, it is clear that the orbit of $x$ is periodic of period $\exp{\theta(n)}$. Therefore, from Lemma~\ref{lemma:primes}, for any $n$, $G_n$ is a circuit of size ${m\leq n\pi(n)}$ with a periodic orbit of size $\theta(n)\in\Omega(\exp(\sqrt{m\log m}))$.
 By hypothesis there are linear maps $T$ and $S$ such that for any $n$, there is $F_n$ that simulates $G_n$ (by Lemma~\ref{lem:from-gadgets-to-networks}), therefore $F_n$ also has a super-polynomial cycle by Lemma~\ref{lem:simu-dynamic}. 
\end{proof}

\newcommand\Gconj{\GG_{conj}}
\newcommand\Fconj{\mathcal{F}_{conj}}
\subsection{Conjunctive networks and $\Gconj$-networks}

Let $G=(V,E)$ be any directed graph. The conjunctive network associated to $G$ is the automata network ${F_G:\{0,1\}^V\to\{0,1\}^V}$ given by ${F(x)_i = \wedge_{j\in N^-(i)}x_j}$ where $N^-(i)$ denotes the incoming neighborhood of $i$. Conjunctive networks are thus completely determined by the interaction graph and a circuit representation can be deduced from this graph in \DLOG{}. We define the family $\Fconj$ as the set of conjunctive networks together with the standard representation $\Fconj^*$ which are just directed graphs encoded as finite words in a canonical way.

\begin{remark}
  We can of course do the same with disjunctive networks. Any conjunctive network $F_G$ on graph $G$ is conjugated to the disjunctive network $F'_G$ on the same graph by the negation map ${\rho : \{0,1\}^V\to \{0,1\}^V}$ defined by ${\rho(x)_i = 1-x_i}$, formally ${\rho\circ F_G = F'_G\circ \rho}$. In particular, this means that the families of conjunctive and disjunctive networks simulate each other. In the sequel we will only state results for conjunctive networks while they hold for disjunctive networks as well.
\end{remark}

Let us now consider the set $\Gconj = \{\gateAND,\gateCOPY\}$.
$\Gconj$-networks are nothing else than conjunctive networks with the following degree constraints: each node has either in-degree $1$ and out-degree $2$, or in-degree $2$ and out-degree $1$. The following theorem shows that, up to simulation, these constraints are harmless.

\begin{theorem}\label{theo:Gconj-networks}
  The family of $\Gconj$-networks simulates the family ${(\Fconj,\Fconj^*)}$ of conjunctive networks in linear time and polynomial space.
\end{theorem}
\begin{proof}
  Let $F$ be an arbitrary conjunctive network on graph ${G=(V,E)}$ with $n$ nodes. Its maximal in/out degree is at most $n$. For each node of indegree ${i\leq n}$ we can make a tree-like $\Gconj$-gadget with $i$ inputs and $1$ output that computes the conjunction of its $i$ inputs in exactly ${n}$ steps: more precisely, we can build a sub-network of size ${O(n)}$ with $i$ identified 'input' nodes of fanin $1$ and one identified output node of fanout $1$ such that for any ${t\in\N}$ the state of the output node at time $t+n$ is the conjunction of the states of the input nodes at time $t$ (the only sensible aspect is to maintain synchronization in the gadget, see Figure~\ref{fig:faningadget}).
  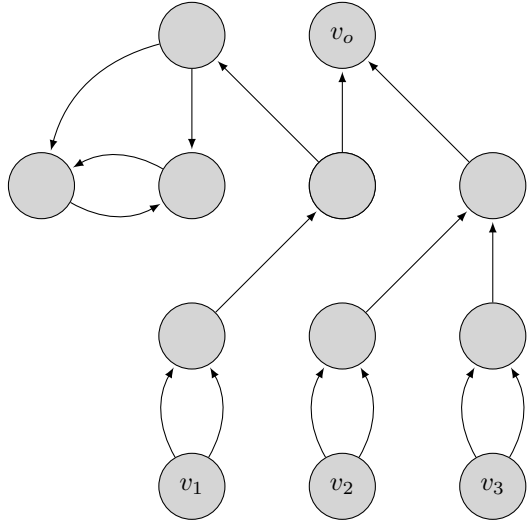
\begin{figure}
    \centering
    \begin{tikzpicture}[shorten >=1pt,node distance=2cm,on grid,auto]
      \tikzstyle{every state}=[fill={rgb:black,2;white,10}]
      \node[state] (q_1)                    {};
      \node[state] (q_2)  [right of=q_1]    {};
      \node[state] (q_3)  [right of=q_2]    {};
      \node[state] (i_1)  [below of=q_1]    {$v_1$};
      \node[state] (i_2)  [below of=q_2]    {$v_2$};
      \node[state] (i_3)  [below of=q_3]    {$v_3$};
      \node[state] (q_4)  [above of=q_3]    {};
      \node[state] (q_5)  [above of=q_2]    {};
      \node[state] (q_6)  [above of=q_5]    {$v_o$};
      \node[state] (q_7)  [above of=q_1]    {};
      \node[state] (q_8)  [above of=q_7]    {};
      \node[state] (q_9)  [right of=q_7]    {};
      \node[state] (q'_6)  [left of=q_7]    {};
      \path[->]
      (q_2) edge (q_4);
      \path[->]
      (q_1) edge (q_5);
      \path[->]
      (q_5) edge (q_8);
      \path[->]
      (q_3) edge (q_4);
      \path[->]
      (q_5) edge (q_6);
      \path[->]
      (q_4) edge (q_6);
      \path[->]
      (q_8) edge (q_7);
      \path[->]
      (q_8) edge[bend right] (q'_6);
      \path[->]
      (q'_6) edge[bend right] (q_7);
      \path[->]
       (q_7) edge[bend right] (q'_6);
      \path[->]
      (i_1) edge[bend left] (q_1);
      \path[->]
      (i_1) edge[bend right] (q_1);
      \path[->]
      (i_2) edge[bend left] (q_2);
      \path[->]
      (i_2) edge[bend right] (q_2);
      \path[->]
      (i_3) edge[bend left] (q_3);
      \path[->]
      (i_3) edge[bend right] (q_3);
    \end{tikzpicture}
    \caption{Fanin gadget of degree $3$. For any configuration $x$, ${F^3(x)_{v_o}=x_{v_1}\wedge x_{v_2}\wedge x_{v_3}}$.}
    \label{fig:faningadget}
  \end{figure}

  We do the same for copying the output of a gate $i$ times and dealing with arbitrary fanout. Then we replace each node of $F$ by a meta node made of the two gadgets to deal with fanin/fanout and connect everything together according to graph $G$ (note that fanin/fanout is granted to be $1$ in the gadgets so connections respect the degree constraints). We obtain in \DLOG{} a $\Gconj$-network of size polynomial in $n$ that simulates $F$ in linear time. 
\end{proof}

\begin{remark}
  The family of conjunctive networks can produce super-polynomial periods but is not universal. There are several ways to show this. It is for instance impossible to produce super-polynomial transients within the family \cite[Theorem 3.20]{De_Schutter_1999} so Corollary~\ref{cor:universality} conclude. One could also use Corollary~\ref{coro:trickyuniversal} since a node in a strongly connected component of a conjunctive network must have a trace period of at most the size of the component (actually much more in known about periods in conjunctive networks through the concept of loop number or cyclicity, see \cite{De_Schutter_1999}).
\end{remark}

\newcommand\Gwfa{\GG_{t}}
\newcommand\stdAND{\mathrm{AND}_{\{0,1\}}}
\newcommand\specAND{\mathrm{AND}_{2}}
\newcommand\gateLoop{\mathrm{\Lambda}}
\newcommand\gateId{\mathrm{Id}}
\newcommand\gateCpy{\mathrm{\Upsilon}}

\subsection{Super-polynomial transients and periods without universality}

Let us consider in this section alphabet ${Q=\{0,1,2\}}$. We are going to define a set $\Gwfa$ such that $\Gwfa$-networks exhibit super-polynomial transients but are not universal. To help intuition, $\Gwfa$-networks can be though as standard conjunctive networks on ${\{0,1\}}$ that can in some circumstances produce state $2$ which is a spreading state (a node switches to state $2$ if one of its incoming neighbors is in state $2$). The extra state $2$ will serve to mark super-polynomial transients, but it cannot escape a strongly connected component once it appears in and, as we will see, $\Gwfa$-networks are therefore too limited in their ability to produce large periodic behavior inside strongly connected components.

$\Gwfa$ is made of the following maps:
\begin{align*}
  \stdAND &: (x,y) \mapsto
                    \begin{cases}
                      2&\text{ if $2\in\{x,y\}$}\\
                      x\wedge y&\text{ else.}
                    \end{cases}\\
  \specAND &: (x,y) \mapsto
                    \begin{cases}
                      2&\text{ if $2\in\{x,y\}$ or $x=y=1$}\\
                      0&\text{ else.}
                    \end{cases}\\
  \gateLoop &: (x,y) \mapsto
                    \begin{cases}
                      2&\text{ if $2\in\{x,y\}$}\\
                      x&\text{ else.}
                    \end{cases}\\
  \gateId &: x\mapsto x.\\
  \gateCpy &: x\mapsto (x,x).
\end{align*}

 $\Gwfa$-networks can produce non-polynomial periods by disjoint union of rotations of prime lengths as in Theorem~\ref{theo:non-polyn-cycl}, but they can also wait for a global synchronization of all rotations and freeze the result of the test for this synchronization condition inside a small feedback loop attached to a ``controlled AND map''.

More precisely, as shown in Figure~\ref{fig:freezeand} we can use in the context of any $\Gwfa$-network a small module ${T(x)}$ of made of five nodes with the following property: if the $\gateLoop$ node of the module is in state $0$ in some initial configuration, then it stays in state $0$ as long as nodes $x$ is not in state $1$, and when $x=1$ at some time step $t$ then from step $t+2$ on the $\gateLoop$ node is in state $2$ at least one step every two steps. This module is the key to control transient behavior.

\begin{figure}
  \centering
  \begin{tikzpicture}[shorten >=1pt,node distance=2cm,on grid,auto]
    \tikzstyle{every state}=[fill={rgb:black,2;white,10}]
    
    \fill[fill=white!90!gray] (-4,0)--(-2,-2)--(-4,-2)--cycle;

    \node[state] (q_1)                    {$\specAND$};
    \node[state] (q_2)  [left of=q_1]    {$\gateCpy_1$};
    \node[state] (q_3)  [below of=q_1]    {$\gateCpy_2$};
    \node[state] (q_4)  [right of=q_1]    {$\gateLoop$};
    \node[state] (q_5)  [below of=q_2]    {$x$};
    \node[state] (q_6)  [below of=q_4]    {$\gateId$};

    \path[<->]
    (q_6) edge (q_4);
    \path[->]
    (q_2) edge (q_1);
    \path[->]
    (q_1) edge (q_4);
    \path[->]
    (q_3) edge (q_1);
    \path[->]
    (q_5) edge (q_2);
    \path[->]
    (q_5) edge (q_3);
    \path[dotted,very thick] (q_5)+(-2,2) edge (q_5);
    \path[dotted,very thick] (q_5)+(-2,0) edge (q_5);
  \end{tikzpicture}
  
  \caption{Freezing the result of a test in a $\Gwfa$-network. The module $T(x)$ is made of the nodes marked $\gateCpy$, $\specAND$, $\gateLoop$ and $\gateId.$  Observe that each node represents some output of its corresponding label (for more details on $\GG$-networks see Definition \ref{def:g-network}). Each gate has one output with the exception of the gate $\Upsilon$ which is represented by two nodes. The module $T(x)$ reads the value of node $x$ belonging to an arbitrary $\Gwfa$-network (represented in light gray inside dotted lines).  The output $\Lambda$ is fed back to its control input via the $\gateId$ node (self-loops are forbidden in $\Gwfa$-networks). Note that $x$  as well as the rest of the network is not influenced by the behavior of the gates of the module $T(x)$.}
  \label{fig:freezeand}
\end{figure}
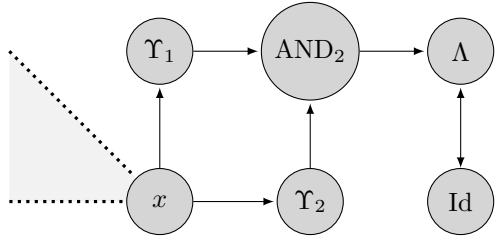

Besides, the map $\stdAND$ behaves like standard Boolean AND map when its inputs are in ${\{0,1\}}$. More generally, by combining such maps in a tree-like fashion, one can build modules ${A(x_1,\ldots,x_k)}$ for any number $k$ of inputs with a special output node which has the following property for some time delay $\Delta\in O(\log(k))$: the output node at time ${t+\Delta}$ is in state $1$ if and only if all nodes $x_i$ (with ${1\leq i\leq k}$) are in state $1$ at time $t$.

Combining these two ingredients, we can build upon the construction of Theorem~\ref{theo:non-polyn-cycl} to obtain non-polynomial transients in any family having coherent $\Gwfa$-gadgets.

\begin{theorem}
  \label{theo:non-poly-transient}
  Any family $\mathcal{F}$ that has coherent $\Gwfa$-gadgets has superpolynomial transients, more precisely: there is some ${\alpha>0}$ such that for any ${n\in\N}$, there exists a network $F_n\in\mathcal{F}$ with ${O(n)}$ nodes and a configuration $x$ such that ${F_n^t(x)}$ is not in an attractor of $F_n$ with ${t\in\Omega(\exp ( n^\alpha))}$.
\end{theorem}
\begin{proof}
  Like in Theorem~\ref{theo:non-polyn-cycl}, the key of the proof is to show that there is a $\Gwfa$-network with transient length as in the theorem statement, then the property immediately holds for networks of the family $\mathcal{F}$ by Lemma~\ref{lem:from-gadgets-to-networks} and Lemma~\ref{lem:simu-dynamic}.

  For any ${n>0}$ we construct a $\Gwfa$-network $G_n$ made of two parts:
  \begin{itemize}
  \item the 'bottom' part of $G_n$ uses a polynomial set of nodes $B_n$ and consists in a disjoint union of circuits for each prime length less than $n$ as in Theorem~\ref{theo:non-polyn-cycl}, but where for each prime $p$, the circuit of length $p$ has a node $v_p$ which implements a copy gate $\gateCOPY$, thus not only sending its value to the next node in the circuit, but also outputting it to the second part of $G_n$;
  \item the 'top' part of $G_n$ is made of a module ${A(x_1,\ldots,x_k)}$ connected to all nodes $v_p$ as inputs and whose output is connected to a test module $T(x)$ as in Figure~\ref{fig:freezeand}.
  \end{itemize}
  Note that the size of $G_n$ is polynomial in $n$.  With this construction we have the following property as soon as the modules ${A(x_1,\ldots,x_k)}$ and $T(x)$ are initialized to state $0$ everywhere: as long as nodes $v_p$ are not simultaneously in state $1$ then the output of the test module $T(x)$ stays in state $0$; moreover, if at some time $t$ nodes $v_p$ are simultaneously in state $t$, then after time $t+O(\log(t))$ the output node of module $T(x)$ is in state $1$ one step every two steps. This means that $t+O(\log(t))$ is a lower bound on the transient of the considered orbit. To conclude the theorem it is sufficient to consider the initial configuration where all nodes are in state $0$ except the successor of node $v_p$ in each circuit of prime length $p$, which are in state $1$. In this case it is clear that the first time $t$ at which all nodes $v_p$ are in state $1$ is the product of prime numbers less than $n$. As in theorem~\ref{theo:non-polyn-cycl}, we conclude thanks to Lemma~\ref{lemma:primes}. 
\end{proof}

As said above, $\Gwfa$-networks are limited in their ability to produce large periods. More precisely, as shown by the following lemma, their behavior is close enough to conjunctive networks so that it can be analyzed as the superposition of the propagation/creation of state $2$ above the behavior of a classical Boolean conjunctive network. To any $\Gwfa$-network $F$ we associate the Boolean conjunctive network $F^*$ with alphabet $\{0,1\}$ as follows: nodes with local map $\stdAND$ or $\specAND$ are simply transformed into nodes with Boolean conjunctive local maps on the same neighbors, nodes with local maps $\gateCpy$ or $\gateId$ are left unchanged (only their alphabet changes), and nodes with map $\gateLoop(x,y)$ are transformed into a node with only $x$ as incoming neighborhood. 

\begin{lemma}\label{lem:superconjunctive}
  Let $F$ be a $\Gwfa$-network with node set $V$ and $F^*$ its associated Boolean conjunctive network. Consider any ${x\in\{0,1,2\}^V}$ and any ${x^*\in\{0,1\}^V}$ such that the following holds: 
  \[\forall v\in V: x_v\in\{0,1\}\Rightarrow x^*_v=x_v,\]
  then the same holds after one step of each network: 
  \[\forall v\in V: F(x)_v\in\{0,1\}\Rightarrow F^*(x^*)_v=F(x)_v.\]
\end{lemma}
\begin{proof}
  It is sufficient to check that if $F(x)_v\neq 2$, it means that all its incoming neighbors are in ${\{0,1\}}$ so $x$ and $x^*$ are equal on these incoming neighbors, and that it only depend on neighbor $a$ in the case of a local map ${\gateLoop(a,b)}$. In any case, we deduce ${F^*(x^*)_v=F(x)_v}$ by definition of $F^\ast$. 
\end{proof}

$\Gwfa$-networks are close to Boolean conjunctive networks as shown by the previous lemma. The following result shows that this translates into strong limitations in their ability to produce large periods and prevents them to be universal.

\begin{theorem}
  \label{theo:transient-nonuniversal}
  The family of $\Gwfa$-networks is not universal.
\end{theorem}
\begin{proof}
  Consider a Boolean conjunctive automata network $F$, a configuration $x$ with periodic orbit under $F$ and some node $v$ such that there is a walk of length $L$ from $v$ to $v$. We claim that $x_v=F^L(x)_v$ so the trace at node $v$ in $x$ is periodic of period less than $L$. Indeed, in a conjunctive network state $0$ is spreading so clearly if $x_v=0$ then $F^L(x)_v=0$ and, more generally, ${F^{kL}(x)_v=0}$ for any ${k\geq 1}$. On the contrary, if ${x_v=1}$ then we can't have ${F^L(x)_v=0}$ because then ${F^{Pk}(x)_v=0}$ with $P$ the period of $x$ which would imply $x_v=0$.

  With the same reasoning, if we consider any $\Gwfa$-network $F$, any configuration $x$ with periodic orbit and some node $v$ such that there is a walk of length $L$ from $v$ to $v$, then it holds: 
  \[x_v=2 \Leftrightarrow F^L(x)_v=2.\]
  We deduce thanks to Lemma~\ref{lem:superconjunctive} that for any configuration $x$ with periodic orbit of some $\Gwfa$-network $F$ with $n$ nodes, and for any node $v$ belonging to some strongly connected component, the period of the trace at $v$ starting from $x$ is less than ${n^2}$: it is a periodic pattern of presence of state $2$ of length less than $n$ superposed on a periodic trace on $\{0,1\}$ of length less than $n$. We conclude that the family of $\Gwfa$-networks cannot be universal thanks to Theorem~\ref{coro:trickyuniversal}. 
\end{proof}

\section{Perspectives}
\label{sec:perspectives}

The main contribution of this paper is a general formalism and a proof technique to show intrinsic universality of families of automata networks, with all the dynamical and computational consequences such a result implies.
As announced earlier, the first perspective is the use of this framework in a companion paper to show how some non-universal concrete families can recover universality by changing the update schedule of the system, thus extending previous results like \cite{goles2020firing,goles2016pspace,goles2014computational}.

However, we believe that several research directions directly connected to the notions developed along with our framework are worth being considered.
We detail some of them below.

\paragraph{Glueing} We think it would be interesting to understand the properties of the glueing process itself and see what information on the result of the glueing process can be deduced from the knowledge of each network to be glued. We are particularly interested in dynamical properties. In addition, it would be very interesting to explore if latter process can be seen in the opposite way, i.e., given an automata network, determine if it is possible to decompose the network into glued blocks satisfying some particular properties as gadgets do.

\paragraph{Simulations and universality}
An obvious working direction following our framework is to classify classical known families with respect to intrinsic universality.
Actually two notions of intrinsic universality are introduced in this paper, and we showed that families coming from intrinsically universal cellular automata are not strongly universal but very close to be. We also see how to build families which are universal but not strongly universal by adding a somewhat artificial mechanism that slows down polynomially any useful computation made by networks in the family, giving examples which are universal but requires a superlinear spatio-temporal rescaling factor. However, we don't have any natural example so far of such 'weakly universal' families and we would like to better understand this territory. In the same spirit, we can ask how a strongly universal family can fail to have coherent $\Gmon$-gadgets (recall that Corollary~\ref{cor:univfrommon} only gives a sufficient condition to be strongly universal). We don't think that strongly universality implies coherent $\Gmon$-gadgets in general, but the implication might at least be true under some additional hypothesis, and possibly in natural families like $\GG$-networks.

\paragraph{$\GG$-networks}
Proposition \ref{prop:closuregadgets} together with theorems~\ref{theo:non-polyn-cycl} and \ref{theo:transient-nonuniversal} provide an interesting starting point to explore the link between different gate sets and the richness of their synchronous closure and the associated family of $\GG$-networks. It is natural to further study the hierarchy between sets of gates and we believe that a promising direction would be to study reversible gate sets such as Toffoli or Fredkin gates. Also, we would like to understand how easy it is to deduce global properties of the family of $\GG$-networks from the knowledge of $\GG$. Typically, one can consider the following decision problem:
\begin{itemize}
\item input: $\GG$
\item question: is the family of $\GG$-networks strongly universal?
\end{itemize}
Is this problem undecidable? If it is the case, what is the minimum number of gates in $\GG$ to obtain undecidability?

\bibliographystyle{plain}

\bibliography{paper.bib}

\end{document}